\documentclass{article}

\RequirePackage{amsthm,amsmath,amsfonts,amssymb}
\RequirePackage{mathrsfs}
\RequirePackage{natbib}
\RequirePackage{hyperref}
\hypersetup{colorlinks=true, linkcolor=red, urlcolor=red, citecolor=blue}
\RequirePackage{booktabs, array}
\RequirePackage{lipsum}
\RequirePackage{graphicx}
\RequirePackage{subcaption}
\RequirePackage{multirow}
\RequirePackage{float}
\RequirePackage[ruled,lined,boxed,linesnumbered]{algorithm2e}
\RequirePackage{pifont}
\RequirePackage{makecell}
\RequirePackage{subcaption}
\RequirePackage[dvipsnames, table]{xcolor}
\RequirePackage{tikz}
\RequirePackage{mathtools}
\RequirePackage{titletoc}

\setcitestyle{numbers,square}

\RequirePackage[ruled,lined,boxed,linesnumbered]{algorithm2e}

\RequirePackage[utf8]{inputenc}
\RequirePackage{setspace}
\RequirePackage{enumerate}
\RequirePackage{lscape}
\RequirePackage[a4paper,
    left=1.3in,
    right=1in]{geometry}
\allowdisplaybreaks
\providecommand{\keywords}[1]{\textbf{Keywords:} #1}

\newtheorem{theorem}{Theorem}[section]

\newtheorem{assumption}{Assumption}[section]
\newtheorem{lemma}{Lemma}[section]
\newtheorem{corollary}{Corollary}[section]
\newtheorem{proposition}{Proposition}[section]

\usepackage{cleveref}
\newtheorem{innercustomthm}{Scenario}
\newenvironment{scenario}[1]
  {\renewcommand\theinnercustomthm{#1}\innercustomthm}
  {\endinnercustomthm}

\providecommand{\keywords}[1]{\textbf{Keywords:} #1}

\RequirePackage{stackengine}

\definecolor{transpred}{RGB}{255, 152, 152}
\definecolor{transpblue}{RGB}{0, 115, 230}
\definecolor{transpyellow}{RGB}{255, 255, 153}

\DeclareMathOperator*{\argmax}{arg\,max}

\title{
\Large
Detecting Changes in Production Frontiers
}

\author{Shakeel Gavioli-Akilagun$^{1,2}$, Yining Chen$^2$, and Flavio Ziegelmann$^3$}

\date{%
    $^1$Department of Decision Analytics and Operations, City University of Hong Kong\\%
    $^2$Department of Statistics, London School of Economics and Political Science\\%
    $^3$ Department of Statistics, Federal University of Rio Grande do Sul\\[2ex]%
    \today
}

\begin{document}

\maketitle

\begin{abstract}
We study the problem of estimating locations in time at which the level of technology in an economy changes when given a sequence of time ordered inputs and outputs. We approach the problem through the lens of nonparametric frontier analysis with frontiers that expand sharply and globally over time, and develop an offline change point detection procedure which achieves the minimax localization rates for the problem at hand up to logarithmic factors. We additionally give a simple method for constructing confidence intervals for the unobserved change point locations. Finally, we explain how the procedure can be modified to accommodate local changes in technology, meaning that efficiency gains are only realized for certain combinations of inputs. Simulation studies and  real data examples are also presented to illustrate the practical value of our methods. 
\end{abstract}


\section{Introduction}

Frontier analysis deals with the problem of estimating production frontier functions, the function characterizing the maximum output an economy may attain given a particular set of inputs, from observational data. A parametric frontier can usually be estimated from data straightforwardly via maximum likelihood \cite[Chapter 3]{kumbhakar2003stochastic}. However, it may not be realistic to assume the functional form of the frontier is known, thus a number of nonparametric estimation methods have been proposed in the literature; see \cite{simar2000statistical, simar2008statistical} for a review of the state of the art.  For nonparametric frontier estimators, popular methods include the free disposal hull estimator \citep{deprins2006measuring}, and data envelopment analysis \citep{charnes1978measuring}. We refer interested readers to \cite{parmeter2014efficiency} for a comprehensive review. 

In reality, advances in technology (almost always) cause the frontier to change and shift outwards over time. See \cite{tonini2012bayesian} and \cite{dhawan1997estimating}. A question which so far has been relatively neglected in the literature concerns rigorously accounting for such productivity changes, which naturally occur over time, see for instance, \citet[Chapter 8]{kumbhakar2003stochastic}. Alas, since the aforementioned approaches have not been designed to account for advances in technology over time, their use in the presence of time-varying productivity will lead to incorrect economic interpretations. A handful of approaches have been put forward in the literature to deal with temporal heterogeneity of the production frontier function, including threshold stochastic frontier models \citep{tsionas2019bayesian}, Markov switching stochastic frontier models \citep{tsionas2004markov}, random coefficient models \citep{tsionas2002stochastic}, and panel data models with cross-sectionally fixed but time-varying parameters \citep{koop2000modeling, koop2000stochastic}. Nonetheless, there remain some gaps in the literature. To name a few, first, existing approaches assume a parametric form for the production frontier function, and will therefore perform poorly when this functional form is misspecified; and second, there is a lack of theory concerning estimation of the time at which each change in technology occurs, which is of separate economic interest. 

In this work we approach the problem from a change point perspective. We study the setting where there might be abrupt changes in the frontier functions over time, and address the problem of accurately estimating the points in time at which the level of technology in an economy changes. As a secondary aim, we also study statistical inference by looking into the construction of confidence intervals for the estimated change point locations. 
 
We now describe the main idea for estimating the points in time at which the level of technology in an economy changes. Its main thrust relies on the fact that given historical data of inputs and outputs produced under several unknown technologies (see Figure~\ref{fig: intuition a}), under the assumption that technology increases monotonically over time (which holds empirically in most settings), the production frontier function corresponding to the most efficient technology can be consistently recovered by estimating the production frontier function using all available data, without any knowledge of when the change occurred (see Figure~\ref{fig: intuition b}). Then, dividing each observed output by the maximum attainable output (over the entire period) given the observed input, according to the aforementioned estimated production frontier function, results in a sequence of random variables with support $[0,1]$ if they are time-indexed after the most recent change, and with support bounded above and away from $1$ if they are time-indexed before the most recent change (see  Figure~\ref{fig: intuition c}). Consequently, the time of the most recent change in technology can be estimated by developing and then applying an appropriate change point detection algorithm to the previously defined sequence (see  Figure~\ref{fig: intuition d}), e.g. on the change in the boundary of the distribution.  Once these locations have been estimated consistently, the time-varying production frontier functions can likewise be recovered consistently by applying a consistent estimator to data indexed between these locations. 

\begin{figure}[!htbp]
\centering
\caption{Intuition for the change point detection procedure. In \eqref{fig: intuition a}: red dots (\textcolor{transpred}{$\bullet$}) represent observed input-output pairs under the old technology while blue dots (\textcolor{transpblue}{$\bullet$}) represent the same under the new technology; red dashed line (\textbf{\textcolor{red}{- - -}}) represents the old production frontier function while blue dashed line (\textbf{\textcolor{blue}{- - -}}) represents the new production frontier function. In \eqref{fig: intuition b}: black dashed line (\textbf{- - -}) represents the Free Disposal Hull (FDH, see Section~\ref{section: FDH estimator} for details) estimator computed on the complete dataset, without knowledge of the change point location. In \eqref{fig: intuition c}: red dots (\textcolor{transpred}{$\bullet$}) represent (time ordered) estimated efficiency scores for observed data under the old technology, while blue dots (\textcolor{transpblue}{$\bullet$}) represent the same under the new technology. In \eqref{fig: intuition d}: the black dashed line (\textbf{- - -}) represents the value of a quasi-likelihood ratio statistic for change in distribution.} 
\label{figure: intuition for change detection}
\begin{subfigure}[b]{0.7\textwidth}
\centering
\includegraphics[width=\textwidth]{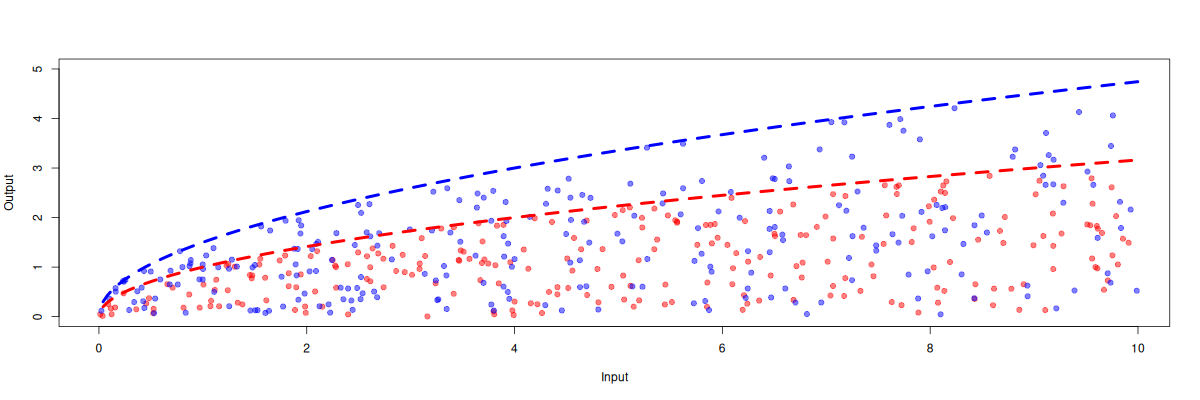}
\caption{}
\label{fig: intuition a}
\end{subfigure}
\begin{subfigure}[b]{0.7\textwidth}
\centering
\includegraphics[width=\textwidth]{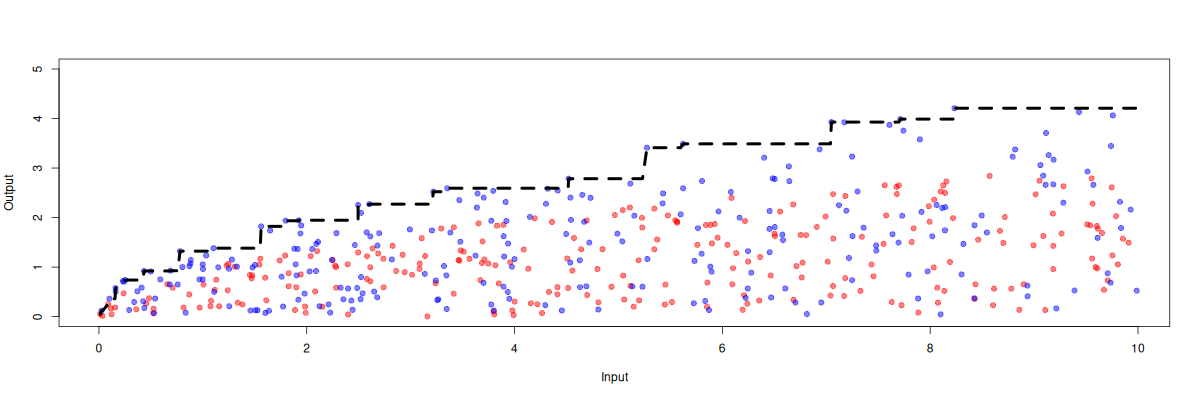}
\caption{}
\label{fig: intuition b}
\end{subfigure}
\begin{subfigure}[b]{0.7\textwidth}
\centering
\includegraphics[width=\textwidth]{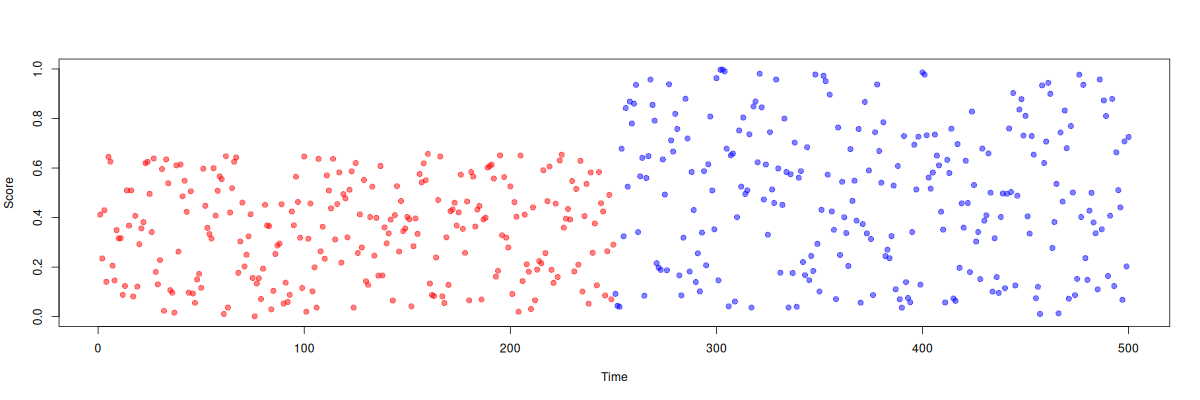}
\caption{}
\label{fig: intuition c}
\end{subfigure}
\begin{subfigure}[b]{0.7\textwidth}
\centering
\includegraphics[width=\textwidth]{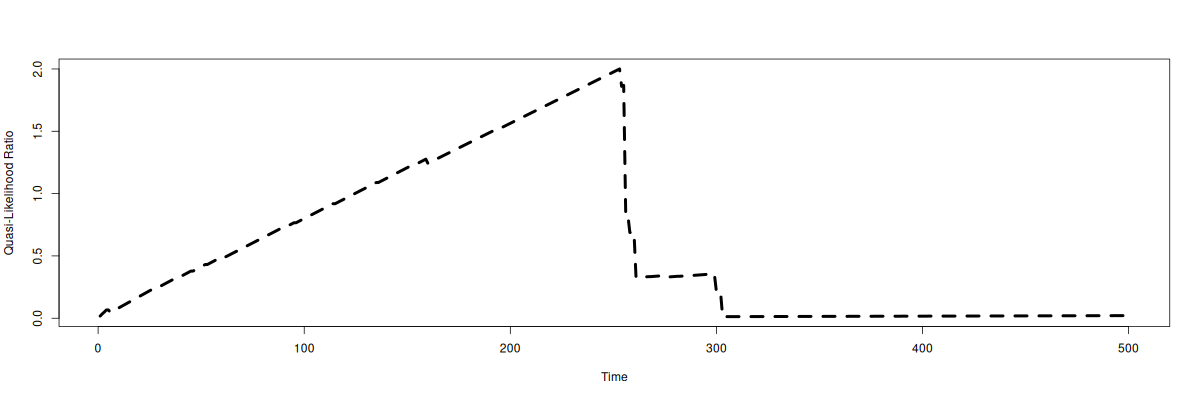}
\caption{}
\label{fig: intuition d}
\end{subfigure}
\end{figure}

It is worth noting that change point analysis deals with the problems where data are observed from a certain stochastic model whose parameters of interest (that could potentially be infinite-dimensional) are, exactly or approximately, piecewise constant. The principal aim is to recover the locations in time at which the parameters change. We refer interested readers to \cite{cho2024data} and \cite{truong2020selective} for comprehensive reviews. As discussed earlier, in this work we recover abrupt change locations in the  frontier model by carefully applying a change point algorithm designed to detect changes in support in the sequence of so called estimated efficiency scores, constructed as described above. Granted, it is tempting to ask whether the problem can be resolved using off-the-shelf change point algorithms. For instance, one may frame our problem as a regression problem with the response represented by the log-output, and attempt to apply algorithms for changes in linear regression \citep{bai1998estimating, bai2003computation, horvath1995detecting} or nonparametric regression \citep{muller1992change, yang2020change}. Alternatively, one may consider the inputs and outputs jointly and apply change point detection algorithms for changes in their joint distribution \citep{mcgonigle2025nonparametric, matteson2014nonparametric, celisse2018new}. However, as we show in Section~\ref{section: simulation studies} such approaches will not be consistent in general unless strong assumptions are placed on the data generating process. 

The remainder of the paper is organized as follows. Section~\ref{section: problem setup} formally introduces the problem, and briefly reviews the Free Disposal Hull (FDH) estimator which will play a central role in the proposed solution. Section~\ref{section: detecting a single change} studies the setting in which the production frontier function may change at most once, and Section~\ref{section: multiple change points} extends the analysis to settings where the production frontier function may change an unknown number of times over the period on which the data is observed. Section~\ref{section: local changes in technology} shows how the methodology can be modified to accommodate local changes in technology, meaning that changes in technology only manifest themselves over certain local regions of the input space. Finally, simulation experiments and real data examples are presented in Section~\ref{section: simulation studies} and Section~\ref{section: real data example}, respectively. All the proofs, as well as further details on the numerical experiments and some possible extensions are provided in the supplementary materials. 

Throughout the paper, we make use of the following notation. For a vector $\boldsymbol{x} = \left ( x_1, \dots, x_d \right )^\top \in \mathbb{R}^d$, we write $\left \| \boldsymbol{x} \right \|_p =  \left ( \sum_{j=1}^d \left | x_j \right | ^p \right ) ^ \frac{1}{p}$ for its $p$ norm ($1 \leq p < \infty $). We write $x \vee y$ for the maximum of two real numbers $x$ and $y$. We write $\mathbb{N} = \{0, 1, 2, \dots\}$ for the set of natural numbers including zero. For any $c > 0$, we write $\mathbb{N}_{< c} = \left \{ n \in \mathbb{N} \mid n < c \right \}$ for the set of natural numbers smaller than $c$. For a set $\mathcal{A}$,  we denote its Lebesgue measure by $\operatorname{mes} \left ( \mathcal{A} \right )$. For a function $f : \mathcal{X} \mapsto \mathbb{R}$, we write $\left \| f \right \|_\infty = \sup_{x \in \mathcal{X}} \left | f (x) \right |$ for its sup-norm. Finally, for an event $A$, the expression $\neg A$ is understood as `not $A$', and $\mathbf{1}_{A}$ as the indicator function taking value $1$ if the event occurs and $0$ if otherwise. 

\section{Problem setup and preliminaries} \label{section: problem setup}

\subsection{Problem setup}

We study the problem in which data $\left \{ \left ( \boldsymbol{X}_t, Y_t \right ) \mid t = 1, \dots, n \right \}$ consisting of pairs of time-ordered inputs and outputs are observed. We consider a model for an economy where at each time step $\boldsymbol{X}_t \in \mathcal{X} \subset \mathbb{R}_+^d$ quantities of a $d$-dimensional input are used to produce $Y_t \in \mathcal{Y} \subset \mathbb{R}_+$ quantities of a scalar valued output. The production set for the economy at each time $t$ is given by 
\begin{equation}
\Psi_t = \left \{ \left ( \boldsymbol{x}, y \right ) \in \mathcal{X} \times \mathcal{Y} \mid \boldsymbol{x} \text{ can produce } y \right \}. 
\end{equation}
The level of technology in the economy is assumed to increase monotonically over time. That is to say: $\Psi_s \subseteq \Psi_t$ for all $s < t$. In addition, we assume the level of technology improves at $K$ discrete locations in time given by $\Theta = \left \{ \eta_1, \dots, \eta_K \right \}$. More precisely, at each time step the production set is given by
\begin{equation}
\Psi_t =  
\begin{cases}
\Psi_{(1)} & \text{ if } \eta_0 < t \leq \eta_1 \\
 & \vdots \\
\Psi_{(K+1)} & \text{ if } \eta_K < t \leq \eta_{K+1},
\end{cases}
\label{equation: sequence of production sets}
\end{equation}
where $\Psi_{(k)} \subset \Psi_{(k+1)}$ for each $k = 1, \dots, K$, and we have put $\eta_0 = 0$ and $\eta_{K+1} = n$ by convention. Both $K$  and $\Theta$ are unknown, and the goal is to estimate these quantities from the data. The boundary of each production set, which may be intuitively understood as the locus of optimal production plans, is defined via 
\begin{equation}
    \partial\Psi_{(k)} = \left \{ \left ( \boldsymbol{x} , y \right ) \in \Psi_{(k)} \mid  \left ( \gamma^{-1} \boldsymbol{x} , \gamma y \right ) \not \in \Psi_{(k)} , \forall \gamma \in (1,\infty) \right \}.
\end{equation}
This quantity can be equivalently characterised by a production frontier function $f_{(k)} : \mathcal{X} \mapsto \mathcal{Y}$ for each $k = 1, \dots, K+1$. That is, we denote by $f_{(1)} (\cdot), \ldots, f_{(K+1)} (\cdot)$ the production frontier functions associated with each of the production sets $\Psi_{(1)},\ldots, \Psi_{(K+1)}$. A pair of inputs and outputs $\left ( \boldsymbol{X}_t, Y_t \right )$ produced with technology $\Psi_t$ is said to have (output) Farrell efficiency $\theta_t^\text{out} = \sup \left \{ r \mid \left ( \boldsymbol{X}_t, r Y_t \right ) \in \Psi_t \right \}$ and (input) Farrell efficiency $\theta_t^\text{in} = \inf \left \{ r \mid \left ( r \boldsymbol{X}_t, Y_t \right ) \in \Psi_t \right \}$; see \cite{farrell1957measurement}. In light of \eqref{equation: sequence of production sets} the inputs and outputs are linked by the equation
\begin{equation}
Y_t = 
\begin{cases}
f_{(1)} (\boldsymbol{X}_t) R_t  & \text{ if } \eta_0 < t \leq \eta_1 \\
 & \vdots \\
f_{(K+1)} (\boldsymbol{X}_t) R_t & \text{ if } \eta_K < t \leq \eta_{K+1},
\end{cases} 
\label{equation: data generating process}
\end{equation}
where each $R_t$ is the inverse Farrell output efficiency measure for the pair $\left ( \boldsymbol{X}_t, Y_t \right )$ under technology $\Psi_t$, which is henceforth referred to as the efficiency scores. In this work, the efficiency scores are modelled as a sequence of random variables each with support $[0,1]$. Besides, we note that in the above formulation the input space $\mathcal{X}$ does not change over time, as otherwise issues of identifiability could arise.

\subsection{Revisiting the Free Disposal Hull (FDH) estimator} \label{section: FDH estimator}

In this section we recall the construction of the Free Disposal Hull (FDH) estimator, which will play an important role in our approach to recovering the change point locations in \eqref{equation: data generating process}. We also provide a result on the finite sample concentration of the FDH estimator, which we use for our theoretical development. Given a sample $\left \{ \left ( \boldsymbol{X}_t, Y_t \right ) \mid t = 1, \dots, n \right \}$ of $d$-dimensional inputs and scalar outputs from model \eqref{equation: data generating process} belonging to a \emph{single} production set $\Psi_{(1)}$ and linked by the equation $Y_t = f_{(1)} ( \boldsymbol{X}_t ) R_t$, where $f_{(1)} (\cdot)$ is the associated production frontier function and $R_t$ is the inverse Farrell output efficiency at time $t$, the FDH estimator for the production set is obtained via
\begin{equation}
\hat{\Psi}_n = \bigcup_{t=1}^n \text{FD} \left ( \boldsymbol{X}_t, Y_t \right ), 
\label{equation: FDH production set estimator} 
\end{equation}
where $\text{FD} \left ( \cdot, \cdot \right )$ gives the free disposal set for an input-output pair $\left ( \boldsymbol{x} ,y \right ) \in \mathcal{X} \times \mathcal{Y}$ and is defined as $\text{FD} \left ( \boldsymbol{x}, y \right ) = \left \{ \left ( \tilde{\boldsymbol{x}}, \tilde{y} \right ) \in \mathcal{X} \times \mathcal{Y} \mid \tilde{y} \leq y \text{ and } \tilde{\boldsymbol{x}} \geq \boldsymbol{x} \right \}$, where the final inequality should be understood entry-wise. Based on \eqref{equation: FDH production set estimator}, one can also define an estimator for the associated production frontier function $f_{(1)} (\cdot)$ as follows:
\begin{equation}
\hat{f}_n ( \boldsymbol{x} ) = \max_{t \in \left \{ s \mid \boldsymbol{X}_s \leq \boldsymbol{x} \right \}}  Y_t, \hspace{2em} \boldsymbol{x} \in \mathcal{X}. 
\label{equation: FDH production frontier function estimator} 
\end{equation}

Compared to other estimators available in the literature, the FDH estimator has the virtue of being fully nonparametric and placing almost no restriction on the production set. The only requirement being that the set is free disposal, meaning that by increasing the quantity of inputs an economy can produce at least the same quantity of output. For a more rigorous result, the following assumptions are needed: 

\begin{assumption}
There are strictly positive constants $\left \{\overline{x}_1, \dots, \overline{x}_{d}, \overline{y} \right \}$ such that $\Psi_{(k)} \subseteq \mathcal{X} \times \mathcal{Y}$ for each $k = 1, \dots, K+1$, where $\mathcal{X} = \times_{j=1}^{d} \mathcal{X}_j$ with $\mathcal{X}_j = \left [ 0, \overline{x}_j \right ]$ for each $j = 1,\ldots, d$ and $\mathcal{Y} = \left [ 0, \overline{y} \right ]$. 
\label{assumption: Z is subset of hypercube}
\end{assumption}

\begin{assumption}
For each $k = 1, \dots, K+1$, the production set $\Psi_{(k)}$ is free disposal. That is: if $(\boldsymbol{x}, y) \in \Psi_{(k)}$ then all $ (\tilde{\boldsymbol{x}}, \tilde{y}) \in \mathcal{X} \times \mathcal{Y}$ with $\tilde{\boldsymbol{x}} \geq \boldsymbol{x}$ and $\tilde{y} \leq y$ are also in $\Psi_{(k)}$. 
\label{assumption: production set is free disposal}
\end{assumption}

\begin{assumption}
\label{assumption: X and R distribution}
For each $t = 1, \dots, n$, the input-output pairs $\left ( \boldsymbol{X}_t, Y_t \right )$ are independently distributed with absolutely continuous joint distribution functions. Moreover:
\begin{enumerate}[(i)]
    \item There are positive constants $C_{1,X}, C_{2,X}$  (that do not depend on $t$) such that the density of each input $\boldsymbol{X}_t$ satisfies $C_{1,X} \leq f_{\boldsymbol{X}_t} \left ( \boldsymbol{x} \right ) \leq C_{2,X} $ for all $\boldsymbol{x} \in \mathcal{X}$. 
    \item There are  positive constants $C_{1,R}, C_{2,R}$ (that do not depend on $t$) such that the conditional density of each efficiency score satisfies $C_{1,R} \leq f_{R_t \mid \boldsymbol{X}_t = \boldsymbol{x}} \left ( r \right ) \leq C_{2,R}$ for all $r \in [0,1]$ and $\boldsymbol{x} \in \mathcal{X}$. 
\end{enumerate}
\end{assumption}
\begin{assumption}
The production frontier functions  $f_{(1)}(\cdot), \dots, f_{(K+1)}(\cdot)$ are $L$-bilipschitz. That is, there is a constant $L \geq 1$ such that for any $\boldsymbol{x}, \tilde{\boldsymbol{x}} \in \mathcal{X}$ and any $k = 1, \dots, K+1$ it holds that $ L^{-1} \left \| \boldsymbol{x} - \tilde{\boldsymbol{x}} \right \|_2 \leq \left | f_{(k)} (\boldsymbol{x}) - f_{(k)} ( \tilde{\boldsymbol{x}} ) \right | \leq L \left \| \boldsymbol{x} - \tilde{\boldsymbol{x}} \right \|_2$. Moreover, it holds that $f_{(k)} \left ( \boldsymbol{0} \right ) = 0 $ for all $k$, where $ \boldsymbol{0} := ( 0, \dots, 0 )^\top \in \mathbb{R}^d $. 
\label{assumption: production frontier function continuous}
\end{assumption}

Some comments about these assumptions are in order. We note that it is common in the literature to assume free disposal production sets (Assumption~\ref{assumption: production set is free disposal}), bounded inputs and outputs (Assumption~\ref{assumption: Z is subset of hypercube}), and continuous production frontier functions (Assumption~\ref{assumption: production frontier function continuous}); see for example the review papers \cite{simar2000statistical, simar2008statistical}. Assumption~\ref{assumption: X and R distribution} assumes that respectively the densities and conditional densities of the inputs and efficiency scores are bounded from below; this is also standard  and is needed to guarantee that enough data is observed close to the boundary of the production set for every point in the input space; see \cite{park2000fdh}. The assumption that the input-output pairs are serially independent is made for ease of exposition, and we show how this assumption can be relaxed in \ref{section: serially dependent} of the supplementary materials. On the other hand, Assumption~\ref{assumption: X and R distribution} does allow for the distributions of the inputs and efficiency scores to change over time, as we do not require them to be identical. With these assumptions in place we have the following result: 

\begin{theorem}
Let $\hat{f}_n \left ( \cdot \right )$ be the FDH estimator from a sample of input-output pairs 
$( \boldsymbol{X}_t, Y_t )$ with $ t = 1, \dots, n$
from a single production set $\Psi_{(1)}$ with associated production frontier function $f_{(1)} \left ( \cdot \right )$ and linked by the equation $Y_t = f_{(1)} (\boldsymbol{X}_t ) R_t$, i.e. with number of changes $K=0$. Grant Assumptions~\ref{assumption: Z is subset of hypercube},~\ref{assumption: production set is free disposal},~\ref{assumption: X and R distribution},~\ref{assumption: production frontier function continuous} hold and put $\psi_n = \left ( \frac{\log (n)}{n} \right ) ^ {\frac{1}{d+1}}$. Then, with $n$ sufficiently large, for any $\zeta > 0$ it holds that
\begin{equation*}
\mathbb{P} \left ( \psi_n^{-1} \left \| f_{(1)} - \hat{f}_n \right \|_\infty > \zeta \right ) \leq C e^{-\zeta},
\end{equation*}
where $C$ is an absolute constant which depends only on the parameters in Assumptions~\ref{assumption: Z is subset of hypercube},~\ref{assumption: production set is free disposal},~\ref{assumption: X and R distribution},~\ref{assumption: production frontier function continuous}.
\label{theorem: exponential concentration in infty distance}
\end{theorem}

This result shows the consistency of the FDH estimator, under mild conditions. It is a slight refinement using similar proof techniques in \cite{korostelev1995efficient} and \cite{korostelev1995estimation}.  The exponential concentration of \eqref{equation: FDH production frontier function estimator} around the truth in the $\ell_\infty$ metric established here will be useful for the theoretical development of our estimation procedure later.

\section{Detecting a single change} \label{section: detecting a single change}

\subsection{Methodology}

In this section we formally introduce the main idea for localizing change points in model \eqref{equation: sequence of production sets} by first studying the simpler problem in which at most one change occurs at an unknown time, i.e., 
\begin{scenario}{S}\label{scenario: AMOC}
The number of change points in the data satisfies $K \in \left \{ 0, 1\right \}$, and with $K=1$ there is a $\mu \in (0,1)$ such that for all $\boldsymbol{x} \in \mathcal{X}$ it holds that $f_{(1)} (\boldsymbol{x}) / f_{(2)} (\boldsymbol{x}) \leq \mu$. 
\end{scenario}
See Figure~\ref{figure: AMOC illustration} for a visual illustration of Scenario~\ref{scenario: AMOC}. Note, however, that the efficiency gain in moving to new technology does not have to be constant over the input space $\mathcal{X}$, as Scenario~\ref{scenario: AMOC} allows for the quantity $f_{(1)}(\boldsymbol{x}) / f_{(2)} (\boldsymbol{x})$ to vary for different values of $\boldsymbol{x}$ provided this quantity is bounded from below by $\mu$. 

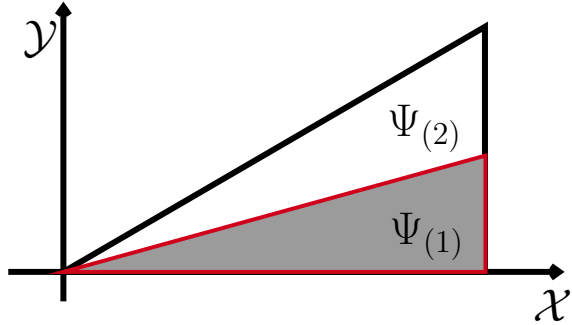
\begin{figure}[!htbp]
\centering
\caption{Illustration of Scenario~\ref{scenario: AMOC}. The grey region (\textcolor{gray}{$\blacksquare$}) corresponds to the production set under the old technology $\Psi_{(1)}$ and the white region ($\square$) corresponds to the production set under the new technology $\Psi_{(2)}$.}
\label{figure: AMOC illustration}

\tikzset{every picture/.style={line width=0.75pt}} 

\begin{tikzpicture}[x=0.75pt,y=0.75pt,yscale=-0.6,xscale=0.6]

\draw [line width=2.25]  (91,243.59) -- (549.57,243.59)(136.86,21.29) -- (136.86,268.29) (542.57,238.59) -- (549.57,243.59) -- (542.57,248.59) (131.86,28.29) -- (136.86,21.29) -- (141.86,28.29)  ;
\draw  [line width=2.25]  (487.57,38.1) -- (136.86,243.59) -- (487.57,243.59) -- cycle ;
\draw  [color={rgb, 255:red, 208; green, 2; blue, 27 }  ,draw opacity=1 ][fill={rgb, 255:red, 155; green, 155; blue, 155 }  ,fill opacity=1 ][line width=1.5]  (488,146.1) -- (136.86,243.59) -- (488,243.59) -- cycle ;

\draw (101,30.74) node [anchor=north west][inner sep=0.75pt]  [font=\LARGE]  {$\mathcal{Y}$};
\draw (525,255.74) node [anchor=north west][inner sep=0.75pt]  [font=\LARGE]  {$\mathcal{X}$};
\draw (406,191.74) node [anchor=north west][inner sep=0.75pt]  [font=\LARGE]  {$\Psi _{( 1)}$};
\draw (404,100.74) node [anchor=north west][inner sep=0.75pt]  [font=\LARGE]  {$\Psi _{( 2)}$};
\end{tikzpicture}
\end{figure}

The thrust of our method for estimating the change point location relies on the fact that if the latter frontier $f_{(2)} \left ( \cdot \right )$ were known, one could then estimate the change point location at the optimal $\mathcal{O}_\mathbb{P} (1)$ rate by forming the ``quasi-efficiency'' scores
\begin{equation}
    \widetilde{R}_t = \frac{Y_t}{f_{(2)} \left ( \boldsymbol{X}_t \right )} , \hspace{2em} \text{for } t = 1, \dots, n, 
    \label{equation: pseudo-efficiency scores}
\end{equation}
and estimating the change point location as the arg-max of two sample test statistics for testing equality of support between $\tilde{R}$'s indexed on $\left \{ 1, \dots, \tau \right \}$ and those indexed on $\left \{ \tau + 1, \dots, n \right \}$ for all $1 \leq \tau < n$; see  Figure~\ref{fig: intuition c}. This is because under Scenario~\ref{scenario: AMOC} the $\tilde{R}_t$'s indexed before the change will be distributed with support $[0, \mu]$ whereas those indexed after the change will be distributed with support $[0,1]$. Consequently, if $f_{(2)} \left ( \cdot \right )$ can be estimated at the $o_\mathbb{P} (1)$ rate, then the change point can likewise be localized at the optimal or near-optimal rate. Although the aforementioned problem appears to rely on the unknown change point location, a key insight is that because $\Psi_{(1)} \subset \Psi_{(2)}$ (as $\mu < 1$), it is in fact possible to estimate $\Psi_{(2)}$, and consequently $f_{(2)} \left ( \cdot \right )$, from data without any knowledge of the change point location. For this task we opt for the nonparametric free disposal hull estimator for simplicity, which was formally introduced and studied in Section~\ref{section: FDH estimator}. In principle, it could be replaced by any other nonparametric approaches, such as Data Envelopment Analysis (DEA).

We are now in a position to describe our estimator for the change point location. Let 
$( \boldsymbol{X}_t, Y_t )$ for $t = 1, \dots, n$
be a sample from \eqref{equation: data generating process} with production frontier functions satisfying Scenario~\ref{scenario: AMOC}, and let $\hat{f}_n (\cdot)$ be the FDH estimator constructed according to \eqref{equation: FDH production frontier function estimator}. Note that although Theorem~\ref{theorem: exponential concentration in infty distance} can be used to show that the $\hat{f}_n (\cdot)$ will be uniformly consistent for $f_{(2)}(\cdot)$, for observations with small input quantities, the corresponding efficiency score cannot be consistently estimated, as both the input and the output would be close to zero. Therefore, for some $\boldsymbol{x}_0 \in \mathcal{X}$ whose choice will be made precise in the sequel, we work only with observations which are entry-wise larger than $\boldsymbol{x}_0$ and consequently define the sequence of estimated pseudo-efficiency scores as follows: 
\begin{equation}
\hat{R}_t = 
\begin{cases}
Y_t / \hat{f}_n \left ( \boldsymbol{X}_t \right ) & \text{ if } \boldsymbol{X}_t > \boldsymbol{x}_0 \\
0 & \text{ else }  
\end{cases},  \hspace{2em} \text{for } t = 1, \dots, n,
\label{equation: estimated scores}
\end{equation}
where the inequality in \eqref{equation: estimated scores} should be understood to be entry-wise for vectors. Consequently, for each $\tau = 1, \dots, n$, define the statistic  
\begin{equation}
\hat{L}_{\tau} = 
-2 N \left ( 1 , \tau , \boldsymbol{x}_0 \right ) \log  \left [ \hat{M} \left ( 1, \tau, \boldsymbol{x}_0 \right ) \right ]  
\label{equation: quasi-LR stat}
\end{equation}
where we follow the convention of $0 \times \infty = 0$, and where for any $1 \leq t_1 < t_2 \leq n$,  
\[
N \left ( t_1 , t_2, \boldsymbol{x}_0 \right ) = \sum_{t=t_1}^{t_2} \mathbf{1}_{\left \{  \boldsymbol{X}_t > \boldsymbol{x}_0 \right \}} \quad \mbox{and} \quad \hat{M} \left ( t_1, t_2, \boldsymbol{x}_0 \right ) = \max \left ( \hat{R}_t  \mid t_1 \leq t \leq t_2, \boldsymbol{X}_t > \boldsymbol{x}_0 \right ).
\]
Here $N \left ( t_1 , t_2, \boldsymbol{x}_0 \right )$ counts the number of observations indexed between $t_1$ and $t_2$  with inputs larger than $\boldsymbol{x}_0$ in an entry-wise manner, while $\hat{M}( t_1, t_2, \boldsymbol{x}_0)$ reports the maximum of the estimated efficiency scores occurring between the aforementioned indices. 

Observe that \eqref{equation: quasi-LR stat} corresponds to the quasi-likelihood ratio statistic, derived from testing the null that a sequence of $\tau$ uniform random variables have support $[0,1]$ against the alternative that they have support $[0,\mu]$ for some $\mu \in (0, 1)$. More precisely, if $U_1, \dots, U_\tau$ are mutually independent Uniform random variables on $[0,\mu]$ the likelihood ratio statistic for testing the aforementioned null and alternative is given by
\begin{align*}
    L_\tau = 2 \log \left [ \sup_{\mu < 1} \prod_{t=1}^\tau \frac{1}{\mu} \boldsymbol{1}_{\{ \max \left ( U_t \mid 1 \leq t \leq \tau \right ) \leq \mu \}} \right ] = - 2 \tau \log \left [ \max \left ( U_t \mid 1 \leq t \leq \tau \right ) \right ]. 
\end{align*}
With these definitions in place, the estimator for the change point location in Scenario~\ref{scenario: AMOC} can be taken to be the arg-max of the sequence of quasi-likelihood ratio statistics: 
\begin{equation}
\hat{\eta} = \argmax_{1 \leq \tau \leq n} \hat{L}_\tau. 
\label{equation: single cpt estimator}
\end{equation} 

\subsection{Theory -- detection and localization}

In order to show the consistency of \eqref{equation: single cpt estimator} for the unknown change point location we need to impose the following assumptions on the amount by which the frontier shifts after the change as well as the effective sample size measured as the minimum of the number of data points occurring to the left and to the right of the change. In the sequel, we have put $x_0 \coloneqq \| \boldsymbol{x}_0 \|_2$, where $\boldsymbol{x}_0$ is the threshold in the input we use for trimming. 

\begin{assumption}
Putting $\delta = \min \left ( \eta, n - \eta \right )$ and letting $L$ be as in Assumption~\ref{assumption: production frontier function continuous} it holds that
\begin{enumerate}[(i)]
    \item there is a positive constant $C_{x,d}$ for which $\delta \geq C_{x,d} \log (n)$, and 
    \item $\delta^{\frac{1}{d+1}} \log (1/\mu) \geq \frac{4 L}{x_0} \log^{3/2} (n)$.
\end{enumerate}
\label{assumption: change point spacing}
\end{assumption}
Part (i) of Assumption~\ref{assumption: change point spacing} is needed  to guarantee that enough data is observed after the change so that $f_{(2)}(\cdot)$ can be estimated well using the FDH estimator, whereas Part (ii) guarantees that the signal from the change, for the purpose of localization, is not dominated by the estimation error from the FDH estimator. In particular, this assumption is automatically satisfied, for $n$ sufficiently large, when $\delta = \mathcal{O} (n)$ and $\log (1 / \mu ) = \mathcal{O} (1)$, which is the typical scenario studied for infill asymptotics in the change point literature. Note also that for the problem at hand the quantity $\log (1/\mu)$ plays the role of the ``size'' of the change when moving from technology $f_{(1)}(\cdot)$ to technology $f_{(2)}(\cdot)$. With these assumptions in place we have the following result: 

\begin{theorem} 
Let $\left \{ \left ( \boldsymbol{X}_t, Y_t \right ) \mid t = 1, \dots, n \right \}$ be a sample from \eqref{equation: data generating process} with production frontier functions satisfying Scenario~\ref{scenario: AMOC}, as well as Assumptions~\ref{assumption: production frontier function continuous}, and~\ref{assumption: change point spacing} and stochastic components satisfying Assumptions~\ref{assumption: Z is subset of hypercube},~\ref{assumption: production set is free disposal}, and \ref{assumption: X and R distribution}. If $\boldsymbol{x}_0$ is chosen such that $\max_{t = 1, \dots, n} \mathbb{P} \left ( \boldsymbol{X}_t \leq \boldsymbol{x}_0 \right ) \leq C_1$ for some $C_1 \in (0,1)$ and the threshold $\lambda$, effective sample size $\delta$, and jump size $\log (1/\mu)$ jointly satisfy
\begin{equation}
C_2 \log (n) \leq \lambda \leq C_3 \delta \log (1/\mu) 
\label{equation: AMOC thresh condition}
\end{equation}
then with  sufficiently large $n$, on a set with probability at least $1 - C_4 n^{-1}$, it holds that: 
\begin{enumerate}[(i)]
    \item if $K = 0$ then $\hat{L}_{\hat{\eta}} \leq \lambda$, 
    \item if $K = 1$ then $\hat{L}_{\hat{\eta}} > \lambda$ and moreover $\left | \hat{\eta} - \eta \right | \leq C_5 \log (n) \vee C_6 \frac{\log(n)}{\log (1/\mu)}$.
\end{enumerate}
where $C_1,\ldots, C_6$ are absolute constants which depend only on the constants stated in Assumptions ~\ref{assumption: Z is subset of hypercube},~\ref{assumption: production set is free disposal},~\ref{assumption: X and R distribution},~\ref{assumption: production frontier function continuous} and~\ref{assumption: change point spacing}.
\label{theorem: single change detection}
\end{theorem}

With this above convergence rate in mind,  we then investigate the global information theoretic minimax lower bounds on the rate at which the change points in model \eqref{equation: sequence of production sets} can be localized. To that end, we have the following result, which reveals that the rate at which the change point is localized by our procedure is unimprovable up to a logarithmic term on $n$. 

\begin{proposition}
Let $\left \{ \left ( \boldsymbol{X}_t, Y_t \right ) \mid t = 1, \dots, n \right \}$ be a sample from \eqref{equation: data generating process} with production frontier functions satisfying Assumption~\ref{assumption: production frontier function continuous} and stochastic components satisfying Assumptions~\ref{assumption: Z is subset of hypercube},~\ref{assumption: production set is free disposal},~\ref{assumption: X and R distribution}.  and having a single change point at location $\eta$, effective sample size $\delta = \min \left ( \eta, n - \eta \right )$, and change size $\mu$. Let $P_{n, \delta, \mu}$ denote the joint law of the data and consider the class
\begin{equation*}
    \mathcal{Q}_n = \left \{ P_{n, \delta, \mu} \mid \delta < n / 2 \text{ and } \delta \log \left ( 1 / \mu \right ) > \zeta_n \right \}
\end{equation*}
for any sequence $\left \{ \zeta_n \mid n > 0 \right \}$ such that $\lim_{n \rightarrow \infty} \zeta_n = \infty$. Then, for all sufficiently large $n$ , it holds that
\begin{equation*}
\inf_{\hat{\eta}} \sup_{P \in \mathcal{Q}_n} \mathbb{E}_{P} \left [ \left | \hat{\eta} - \eta \right | \right ] \geq \frac{e^{-2}}{2 \log (1/\mu)} 
\end{equation*}
where the infimum is over all measurable functions of the data and $\eta (P)$ denotes the change point location for a given $P \in \mathcal{Q}_n$. 
\label{lemma: global change minimax lower bound}
\end{proposition}

While the best rate at which each production frontier function can be estimated depends in a precise way on the dimension of the input space \cite{korostelev1995efficient, korostelev1995estimation}, we see that interestingly the minimax rate for estimating change points in model \eqref{equation: data generating process} is dimension agnostic, and depends only on the amount by which the frontier shifts after the change. This is because estimating the production frontier function is a genuinely multidimensional estimation problem whereas, following the logic in \eqref{equation: pseudo-efficiency scores}, estimating each change point location boils down to a one-dimensional estimation problem in the presence of an (infinite dimensional) nuisance parameter that can be estimated at the $o_\mathbb{P} (1)$ rate  according to Theorem~\ref{theorem: exponential concentration in infty distance}. 

\subsection{Inference} \label{section: inference on a single change}

In this section we describe the asymptotic distribution of the change point estimator defined via \eqref{equation: single cpt estimator}. Consequently, we outline a simple procedure for constructing a confidence interval for the unobserved change point location. We observe that, compared with existing results on the distribution of change point estimators (e.g., see \cite{dumbgen1991asymptotic}), our results are obtained using only elementary calculations due to the simple structure of the quasi-likelihood ratio statistic used to construct the estimator. 
\begin{theorem}
Let $\hat{\eta}$ be as defined in \eqref{equation: single cpt estimator} and suppose that the assumptions in Theorem~\ref{theorem: single change detection} hold. Assume additionally that $\mu$ does not depend on the sample size $n$. As $n \rightarrow \infty$, it holds that 
\begin{enumerate}[(i)]
    \item $\left ( \hat{\eta} - \eta \right ) $ is asymptotically stochastically dominated by $Z - 1$ where 
    \begin{equation*}
        Z \sim \text{Geometric} \left ( \theta \left ( \boldsymbol{x}_0, \mu, C_{1,X}, C_{1,R} \right ) \right ),
    \end{equation*}
    with
    \begin{equation}
        \theta \left ( \boldsymbol{x}_0, \mu, C_{1,X}, C_{1,R} \right ) = C_{1,R} \left ( 1 - \mu \right ) C_{1,X} \operatorname{mes} \left ( \mathcal{X} \setminus \left \{ \boldsymbol{x} \leq  \boldsymbol{x}_0 \right \} \right ). 
        \label{equation: geom param non-iid}
    \end{equation}
    \item If additionally the random variables $\left ( \boldsymbol{X}_t, R_t \right )_{t=1,\dots,n}$ are identically distributed then $\left ( \hat{\eta} - \eta \right ) $ is asymptotically stochastically dominated by $Z - 1$ where
    \begin{equation*}
        Z \sim \text{Geometric} \left ( \theta \left ( \boldsymbol{x}_0, \mu \right ) \right ),
    \end{equation*}
    with
    \begin{equation}
        \theta \left ( \boldsymbol{x}_0, \mu \right ) = \mathbb{P} \left ( R_1 \boldsymbol{1}_{\left \{ \boldsymbol{X}_1 > \boldsymbol{x}_0 \right \}} \geq \mu \right ).
        \label{equation: geom pama iid}
    \end{equation}
\end{enumerate}
\label{lemma: single change inference}
\end{theorem}

Theorem~\ref{lemma: single change inference} implies that one can construct an asymptotically conservative confidence interval for the change point location by consistently estimating \eqref{equation: geom param non-iid} or  \eqref{equation: geom pama iid}. We now outline how these quantities can be estimated. Let $\hat{f}_{(1),n} \left ( \cdot \right )$ be the FDH estimator constructed using data indexed on $\left \{ 1, \dots, \hat{\eta} \right \}$ and $\hat{f}_{(2),n} \left ( \cdot \right )$ be the estimator constructed using data indexed on $\left \{ \hat{\eta} + 1, \dots, n \right \}$. The parameter $\mu$ defined in Scenario~\ref{scenario: AMOC} can be estimated via
\begin{equation}
    \hat{\mu} = \sup_{\boldsymbol{x} > \boldsymbol{x}_0} \frac{\hat{f}_{(1),n} \left ( \boldsymbol{x} \right )}{ \hat{f}_{(2),n} \left ( \boldsymbol{x} \right )}. 
    \label{equation: jump size estimators}
\end{equation}
Similarly, the unobserved efficiency scores can be estimated via
\begin{equation*}
    \hat{R}_t = \frac{Y_t}{\hat{f}_{(1),n} \left ( \boldsymbol{X}_t \right ) \mathbf{1}_{\left \{ 1 \leq t \leq \hat{\eta} \right \}} + \hat{f}_{(2),n} \left ( \boldsymbol{X}_t \right ) \mathbf{1}_{\left \{  \hat{\eta} + 1 \leq t \leq n \right \}}}, \hspace{2em} t = 1, \dots, n. 
\end{equation*}
Here the quantities $C_{1,X}$ and $C_{1,R}$ can be estimated with further assumptions, such as the distributions of $(\boldsymbol{X}_t, R_t)$ vary slowly. Note that without any such assumptions, it would be impossible to estimate $C_{1,X}$ and $C_{1,R}$, because we could only observe one data from $(\boldsymbol{X}_t, R_t)$ for each $t$. Now under the slow varying distribution assumption, we can estimate these quantities using the minimum of histogram density estimators for inputs and the estimated efficiency scores. That is, let $h_n$ and $h_{1,n}, \dots, h_{d,n}$ be bandwidths, let $H_n$ be the window size (e.g. taken as $n/\log(n)$) and let $\mathcal{P} \left ( \mathcal{X}; h_{1,n}, \dots,h_{d,n} \right )$ denote the partition of $\mathcal{X}$ into non-overlapping hyper-cubes with sides of length at most $h_{1,n}, \dots, h_{d,n}$. Then the histogram based estimators for $C_{1,R}$ and $C_{1,X}$ are given by
\begin{align}
& \hat{C}_{1,R} = \min_{1 \le i \le n-H_n+1}\min_{1 \leq j \leq h^{-1}_n} \frac{1}{H_n} \sum_{t=i}^{i+H_n-1} h_n^{-1} \boldsymbol{1}_{\left \{ (j-1) h_n < \hat{R}_t \leq j h_n \text{ and } \boldsymbol{X}_t > \boldsymbol{x}_0 \right \}} \\
& \hat{C}_{1,X} =  \min_{1 \le i \le n-H_n+1} \min_{\mathcal{X}' \in \mathcal{P} \left ( \mathcal{X}; h_{1,n}, \dots,h_{d,n} \right )} \frac{1}{H_n} \sum_{t=i}^{i+H_n-1} \prod_{j=1}^d h_{j,n}^{-1} \boldsymbol{1}_{\left \{ \boldsymbol{X}_t \in \mathcal{X}' \right \}}.
\label{equation: histogram estimators}
\end{align}
Data driven choices for the bandwidths can be found in \cite{devroye2004bin} and \cite{wand1997data}. 
Then one may estimate \eqref{equation: geom param non-iid} via the plug-in estimator $\theta ( \boldsymbol{x}_0, \hat{\mu}, \hat{C}_{1,X}, \hat{C}_{1,R})$. Finally, in the simpler setting where the inputs and efficiency scores are identically distributed, one can estimate the parameter of interest \eqref{equation: geom pama iid} much more conveniently via 
\begin{equation*}
\hat{\theta} \left ( \boldsymbol{x}_0, \hat{\mu} \right ) = \frac{1}{n} \sum_{t=1}^n \mathbf{1}_{\left \{\hat{R}_t \geq \hat{\mu} \text{ and } \boldsymbol{X}_t > \boldsymbol{x}_0 \right \}}. 
\end{equation*}

\section{Detecting multiple changes} \label{section: multiple change points}

\subsection{Methodology}

Next, we study the problem in which multiple change points occur in the data, where after each change,  the frontier shifts upwards uniformly, although not necessarily by the same ratio, over all regions of the input space: 
\begin{scenario}{M}\label{scenario: multiple changes}
The number of change points in the data satisfies $K \in \mathbb{N}_{<n}$, and for each $k \in \left \{ 1, \dots, K \right \}$ there is a $\mu_k \in (0,1)$ such that for all $\boldsymbol{x} \in \mathcal{X}$ it holds that $f_{(k)} (\boldsymbol{x}) / f_{(k+1)} (\boldsymbol{x}) \leq \mu_k$. 
\end{scenario}

We now outline our procedure for multiple change point detection. Following the reasoning in Section~\ref{section: FDH estimator}, under Scenario~\ref{scenario: multiple changes} the production frontier function corresponding to the rightmost stationary segment, that is $ t \in \left \{ \eta_{K} + 1, \dots, \eta_{K+1} \right \}$, can be consistently estimated by applying the FDH estimator to the entire data sequence. Consequently, forming the sequence of pseudo-efficiency scores as in \eqref{equation: estimated scores} the rightmost change point can be estimated using a localized version of the detection statistic defined in \eqref{equation: quasi-LR stat}. That is, with integers $t_1,t_2$ such that $1 \leq t_1 < t_2 \leq n$ for any $\tau \in \{ t_1, \dots, t_2 \}$ we define the localized analogue of \eqref{equation: quasi-LR stat} as follows: 
\begin{equation}
\hat{L}_{t_1, \tau, t_2} = - 2 N \left ( t_1 , \tau , \boldsymbol{x}_0 \right ) \log \left [ \hat{M} \left ( t_1, \tau, \boldsymbol{x}_0 \right ) \right ]. 
\label{equation: localized quasi-LR stat}
\end{equation}
To estimate the rightmost change point, we take the arg-max of \eqref{equation: localized quasi-LR stat} over left-expanding intervals, and estimate the change point location as the arg-max over the narrowest interval on which the maximum of the test statistic exceeds a threshold $\lambda$ to be defined precisely later on; Figure \ref{figure: left-expanding intervals} provides a visual illustration. The intuition is that, up to and including the first interval on which a statistic exceeds the chosen threshold, every such interval will contain either one or no change points. Similar approaches to change point detection have been proposed in \cite{anastasiou2022detecting} and in \cite{fang2016segmentation}. Specifically we take the arg-max of \eqref{equation: localized quasi-LR stat} over the sequence of intervals $\left \{ n - j, \dots, n \right \}$ with $j = 1, \dots, n - 1$. Therefore letting 
\begin{equation}
j^* = \min \left \{ 1 \leq j \leq n - 1 \mid \max_{n-j \leq \tau \leq n} \hat{L}_{n-j,\tau,n} > \lambda \right \}
\label{equation: smallest left interval}
\end{equation}
be the index of the shortest left expanding interval on which a change is detected, the right most-change point is estimated via 
\begin{equation}
\tilde{\eta} =  \argmax_{n-j^* \leq \tau \leq n} \hat{L}_{n-j^*,\tau,n}
\label{equation: pilot change point estimates}
\end{equation}
Here if \eqref{equation: smallest left interval} is undefined we conclude there are no change points in the data.

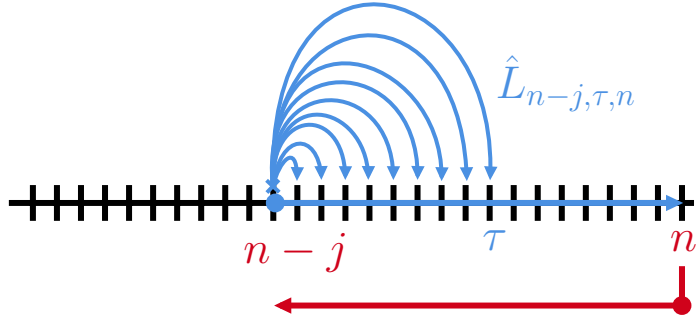
\begin{figure}
\centering
\caption{Illustration of the left-expanding intervals used in Algorithm \ref{algorithm: multiple change detection}. Detecting the right-most change point, at the $j$-th iteration, if an $\hat{L}_{n-j,\tau, n}$ with $\tau \in \{ n-j, \dots, n\}$ exceeds the threshold $\lambda$, then $\tilde{\eta}$ is estimated as the arg-max over the $\tau$'s. Else, $j$ is incremented by $1$ and the local statistics \eqref{equation: localized quasi-LR stat} are again computed for all $\tau$ between $n-j$ and $n$.}

\tikzset{every picture/.style={line width=0.75pt}} 

\begin{tikzpicture}[x=0.75pt,y=0.75pt,yscale=-0.6,xscale=0.6]

\draw [line width=2.25]    (24.33,219) -- (594.33,219) (44.33,204) -- (44.33,234)(64.33,204) -- (64.33,234)(84.33,204) -- (84.33,234)(104.33,204) -- (104.33,234)(124.33,204) -- (124.33,234)(144.33,204) -- (144.33,234)(164.33,204) -- (164.33,234)(184.33,204) -- (184.33,234)(204.33,204) -- (204.33,234)(224.33,204) -- (224.33,234)(244.33,204) -- (244.33,234)(264.33,204) -- (264.33,234)(284.33,204) -- (284.33,234)(304.33,204) -- (304.33,234)(324.33,204) -- (324.33,234)(344.33,204) -- (344.33,234)(364.33,204) -- (364.33,234)(384.33,204) -- (384.33,234)(404.33,204) -- (404.33,234)(424.33,204) -- (424.33,234)(444.33,204) -- (444.33,234)(464.33,204) -- (464.33,234)(484.33,204) -- (484.33,234)(504.33,204) -- (504.33,234)(524.33,204) -- (524.33,234)(544.33,204) -- (544.33,234)(564.33,204) -- (564.33,234)(584.33,204) -- (584.33,234) ;
\draw [color={rgb, 255:red, 208; green, 2; blue, 27 }  ,draw opacity=1 ][line width=2.25]    (250,305) -- (584.33,305) ;
\draw [shift={(584.33,305)}, rotate = 0] [color={rgb, 255:red, 208; green, 2; blue, 27 }  ,draw opacity=1 ][fill={rgb, 255:red, 208; green, 2; blue, 27 }  ,fill opacity=1 ][line width=2.25]      (0, 0) circle [x radius= 5.36, y radius= 5.36]   ;
\draw [shift={(245,305)}, rotate = 0] [fill={rgb, 255:red, 208; green, 2; blue, 27 }  ,fill opacity=1 ][line width=0.08]  [draw opacity=0] (14.29,-6.86) -- (0,0) -- (14.29,6.86) -- cycle    ;
\draw [color={rgb, 255:red, 208; green, 2; blue, 27 }  ,draw opacity=1 ][line width=2.25]    (584.33,305) -- (584.33,272) ;
\draw [color={rgb, 255:red, 74; green, 144; blue, 226 }  ,draw opacity=1 ][fill={rgb, 255:red, 208; green, 2; blue, 27 }  ,fill opacity=1 ][line width=2.25]    (246,219) -- (580.33,219) ;
\draw [shift={(585.33,219)}, rotate = 180] [fill={rgb, 255:red, 74; green, 144; blue, 226 }  ,fill opacity=1 ][line width=0.08]  [draw opacity=0] (14.29,-6.86) -- (0,0) -- (14.29,6.86) -- cycle    ;
\draw [shift={(246,219)}, rotate = 0] [color={rgb, 255:red, 74; green, 144; blue, 226 }  ,draw opacity=1 ][fill={rgb, 255:red, 74; green, 144; blue, 226 }  ,fill opacity=1 ][line width=2.25]      (0, 0) circle [x radius= 5.36, y radius= 5.36]   ;
\draw [color={rgb, 255:red, 74; green, 144; blue, 226 }  ,draw opacity=1 ][line width=1.5]    (244,205) .. controls (248.99,183.69) and (263.75,167.88) .. (264.04,196.73) ;
\draw [shift={(264,200.57)}, rotate = 271.68] [fill={rgb, 255:red, 74; green, 144; blue, 226 }  ,fill opacity=1 ][line width=0.08]  [draw opacity=0] (11.61,-5.58) -- (0,0) -- (11.61,5.58) -- cycle    ;
\draw [color={rgb, 255:red, 74; green, 144; blue, 226 }  ,draw opacity=1 ][line width=1.5]    (244,205) .. controls (250.01,158.25) and (281.37,159.54) .. (284.16,196.17) ;
\draw [shift={(284.33,199.67)}, rotate = 268.57] [fill={rgb, 255:red, 74; green, 144; blue, 226 }  ,fill opacity=1 ][line width=0.08]  [draw opacity=0] (11.61,-5.58) -- (0,0) -- (11.61,5.58) -- cycle    ;
\draw [color={rgb, 255:red, 74; green, 144; blue, 226 }  ,draw opacity=1 ][line width=1.5]    (244,205) .. controls (241.26,139.14) and (302.47,136.65) .. (304.17,196.08) ;
\draw [shift={(304.2,199.8)}, rotate = 270.72] [fill={rgb, 255:red, 74; green, 144; blue, 226 }  ,fill opacity=1 ][line width=0.08]  [draw opacity=0] (11.61,-5.58) -- (0,0) -- (11.61,5.58) -- cycle    ;
\draw [color={rgb, 255:red, 74; green, 144; blue, 226 }  ,draw opacity=1 ][line width=1.5]    (244,205) .. controls (236.32,109.26) and (341.76,111.48) .. (344.16,195.89) ;
\draw [shift={(344.2,199.8)}, rotate = 270.52] [fill={rgb, 255:red, 74; green, 144; blue, 226 }  ,fill opacity=1 ][line width=0.08]  [draw opacity=0] (11.61,-5.58) -- (0,0) -- (11.61,5.58) -- cycle    ;
\draw [color={rgb, 255:red, 74; green, 144; blue, 226 }  ,draw opacity=1 ][line width=1.5]    (244,205) .. controls (231.13,88.75) and (364.62,91.6) .. (364.38,196.46) ;
\draw [shift={(364.33,199.67)}, rotate = 271.59] [fill={rgb, 255:red, 74; green, 144; blue, 226 }  ,fill opacity=1 ][line width=0.08]  [draw opacity=0] (11.61,-5.58) -- (0,0) -- (11.61,5.58) -- cycle    ;
\draw [color={rgb, 255:red, 74; green, 144; blue, 226 }  ,draw opacity=1 ][line width=1.5]    (244,205) .. controls (243.35,117.45) and (320.93,135.65) .. (323.28,195.92) ;
\draw [shift={(323.33,199.67)}, rotate = 270.72] [fill={rgb, 255:red, 74; green, 144; blue, 226 }  ,fill opacity=1 ][line width=0.08]  [draw opacity=0] (11.61,-5.58) -- (0,0) -- (11.61,5.58) -- cycle    ;
\draw [color={rgb, 255:red, 74; green, 144; blue, 226 }  ,draw opacity=1 ][line width=1.5]    (244,205) .. controls (236.41,47.26) and (384.33,91.75) .. (384.38,197.45) ;
\draw [shift={(384.33,200.67)}, rotate = 271.59] [fill={rgb, 255:red, 74; green, 144; blue, 226 }  ,fill opacity=1 ][line width=0.08]  [draw opacity=0] (11.61,-5.58) -- (0,0) -- (11.61,5.58) -- cycle    ;
\draw [color={rgb, 255:red, 74; green, 144; blue, 226 }  ,draw opacity=1 ][fill={rgb, 255:red, 0; green, 0; blue, 0 }  ,fill opacity=0 ][line width=1.5]    (244,205) .. controls (243.77,200.2) and (243.69,195.54) .. (243.77,191) .. controls (243.77,190.94) and (243.77,190.88) .. (243.77,190.82) .. controls (243.78,190.1) and (243.8,189.39) .. (243.82,188.68) .. controls (243.86,187.2) and (243.92,185.73) .. (243.99,184.27) .. controls (244.11,181.9) and (244.26,179.57) .. (244.46,177.28) .. controls (244.58,175.88) and (244.72,174.49) .. (244.87,173.12) .. controls (258.72,47.11) and (395.24,41.14) .. (404.87,190.09) .. controls (404.94,191.1) and (405,192.11) .. (405.05,193.13) .. controls (405.11,194.3) and (405.16,195.49) .. (405.21,196.68) ;
\draw [shift={(405.33,200.67)}, rotate = 268.67] [fill={rgb, 255:red, 74; green, 144; blue, 226 }  ,fill opacity=1 ][line width=0.08]  [draw opacity=0] (11.61,-5.58) -- (0,0) -- (11.61,5.58) -- cycle    ;
\draw [color={rgb, 255:red, 74; green, 144; blue, 226 }  ,draw opacity=1 ][line width=1.5]    (244,205) .. controls (237.64,-15.25) and (425.31,24.18) .. (424.96,189.14) .. controls (424.95,191.42) and (424.91,193.73) .. (424.83,196.05) ;
\draw [shift={(424.67,200)}, rotate = 272.91] [fill={rgb, 255:red, 74; green, 144; blue, 226 }  ,fill opacity=1 ][line width=0.08]  [draw opacity=0] (11.61,-5.58) -- (0,0) -- (11.61,5.58) -- cycle    ;
\draw [shift={(244,205)}, rotate = 313.35] [color={rgb, 255:red, 74; green, 144; blue, 226 }  ,draw opacity=1 ][line width=1.5]    (-7.27,0) -- (7.27,0)(0,7.27) -- (0,-7.27)   ;

\draw (572,240.4) node [anchor=north west][inner sep=0.75pt]  [font=\LARGE,color={rgb, 255:red, 208; green, 2; blue, 27 }  ,opacity=1 ]  {$n$};
\draw (217,241.4) node [anchor=north west][inner sep=0.75pt]  [font=\LARGE,color={rgb, 255:red, 208; green, 2; blue, 27 }  ,opacity=1 ]  {$n-j$};
\draw (416,239.4) node [anchor=north west][inner sep=0.75pt]  [font=\LARGE,color={rgb, 255:red, 74; green, 144; blue, 226 }  ,opacity=1 ]  {$\tau $};
\draw (427,92.4) node [anchor=north west][inner sep=0.75pt]  [font=\LARGE,color={rgb, 255:red, 74; green, 144; blue, 226 }  ,opacity=1 ]  {$\hat{L}_{n-j,\tau ,n}$};
\end{tikzpicture}
\label{figure: left-expanding intervals}
\end{figure}

However, to estimate the next right-most change point we cannot simply repeat the procedure for the remaining data indexed on $\left \{ 1, \dots, \tilde{\eta} - 1 \right \}$. This is because we may have that $\tilde{\eta} > \eta_K + 1$, in which case the data indexed on $\left \{ 1, \dots, \tilde{\eta} - 1 \right \}$ will contain at least one observation from the production set $\Psi_{(K+1)}$, and consequently the FDH estimator based on data indexed on this set is not guaranteed to estimate $f_{(K)} \left ( \cdot \right )$ consistently. To overcome this problem we put 
\begin{equation}
m = n - j^*
\label{equation: m's}
\end{equation}
and instead repeat the procedure on $\left \{ 1, \dots, m \right \}$. As can be seen from the proof of Theorem~\ref{theorem: multiple change detection}, on a high probability set we will have that $m < \eta_{K}$ and so the FDH estimator computed with data indexed on the aforementioned set is once again consistent for $f_{(K)} \left ( \cdot \right )$. The steps described are therefore repeated sequentially until no more change points can be detected in the data. 

Finally, in an attempt to improve finite sample performance, we apply a local refitting step to each of the change point locations recovered by the procedure described above. More precisely, let $\tilde{\Theta} = \left \{ \tilde{\eta}_1, \dots, \tilde{\eta}_{\hat{K}} \right \}$ be the (ordered) set of estimated change point locations recovered as in \eqref{equation: pilot change point estimates} and let $\tilde{\mathcal{M}} = \left \{ m_1, \dots, m_{\hat{K}} \right \}$ be the (likewise ordered) sequence of $m$'s recovered as in \eqref{equation: m's}. Putting $\hat{\eta}_0 = m_1 = 1$ and $\hat{\eta}_{\hat{K}+1} = n$, as well as
\begin{align}
\tilde{s}_k = \frac{\tilde{\eta}_{k-1} + m_k}{2} \hspace{1em} \text{and} \hspace{1em} \tilde{e}_k = m_{k+1}, \hspace{2em}  \mbox{ for }  k = 1, \dots, \hat{K}.
\label{equation: s e refit points}
\end{align}
Then each estimated change point location is locally refitted as
\begin{equation}
\hat{\eta}_k =  \argmax_{\tilde{s}_k \leq \tau \leq \tilde{e}_k} \hat{L}_{\tilde{s}_k,\tau,\tilde{e}_k}, \hspace{1em} k = 1, \dots, \hat{K}, 
\label{equation: local refitting}
\end{equation}
where, however, at each $k$, the estimated quasi-efficiency scores are calculated as in \eqref{equation: estimated scores} using the estimated production frontier function which was used to detect the corresponding $\tilde{\eta}_k$. This is suppressed for economy of notation in the definition of \eqref{equation: local refitting}. Overall, the change point detection procedure is summarized in Algorithm~\ref{algorithm: multiple change detection}. 

\begin{algorithm}[p]
\DontPrintSemicolon
\SetKwInOut{Input}{Input}
\SetKwInOut{Output}{Output}
\Input{the observed data sequence $\left ( \boldsymbol{X}_1, Y_1 \right ), \dots, \left ( \boldsymbol{X}_n, Y_n \right )$, a threshold $\lambda$ against which to compare local quasi-Likelihood ratio statistics, and a minimum value for the norm of the inputs $\boldsymbol{x}_0$.}
\Output{a set of estimated change points $\hat{\Theta} \subset \left \{ 1, \dots, n \right \}$.}
\BlankLine
\Begin{
$\tilde{\Theta} \leftarrow \emptyset$; $\hat{M} \leftarrow \emptyset$; $\hat{\mathcal{F}} \leftarrow \emptyset$ \\
$\hat{K} \leftarrow 0$; $m \leftarrow n$; $\text{test} \leftarrow \texttt{True}$ \\
\While{{\upshape test}}
{
$\text{detection} \leftarrow \texttt{False}$ \\
Estimate $\hat{f}_n \left ( \cdot \right )$ from $\left \{ \left ( \boldsymbol{X}_t, Y_t \right ) \mid t = 1, \dots, m \right \}$ via \eqref{equation: FDH production frontier function estimator} \\
Obtain $\hat{R}_1, \dots, \hat{R}_m$ as in \eqref{equation: estimated scores} \\
\For{$j \in \left \{ 1, \dots, m-1 \right \}$}{
$L^* \leftarrow \max \left \{ \hat{L}_{m-j, \tau , m} \mid \tau = m-j, \dots, m \right \}$ \\
\If{$L^* > \lambda$}{
$\tilde{\eta} \leftarrow  \argmax_{m-j \leq \tau \leq m} \hat{L}_{m-j,\tau,m}$ \\
$m \leftarrow m - j$ \\
$\hat{M} \leftarrow \left \{ m \right \} \cup \hat{M}$ \\
$\tilde{\Theta} \leftarrow \left \{ \tilde{\eta} \right \} \cup \tilde{\Theta}$ \\
$\hat{K} \leftarrow \hat{K} + 1 $ \\
$\hat{\mathcal{F}} \leftarrow \left \{ \hat{f}_m \left ( \cdot \right ) \right \} \cup \hat{\mathcal{F}}$ \\
$\text{detection} \leftarrow \texttt{True}$ \\
BREAK \\
}
}
\If{{\upshape $m < 1$ or $\neg \text{detection}$}}{test $\leftarrow$ \texttt{False}}
}
$\hat{\Theta} \leftarrow \emptyset$ \\
\If{$\hat{K} = 0$}{STOP}
\For{$k = 1, \dots, \hat{K}$}
{
Obtain $\hat{R}_1, \dots \hat{R}_n$ as in \eqref{equation: estimated scores} using $\hat{f}_k \left ( \cdot \right )$ from $\hat{\mathcal{F}}$ \\
Assign $\tilde{s}_k$ and $\tilde{e}_k$ as in \eqref{equation: s e refit points} \\
$\hat{\eta}_k \leftarrow \argmax_{\tilde{s}_k  \leq \tau \leq \tilde{e}_k} \hat{L}_{\tilde{s}_k, \tau, \tilde{e}_k}$ \\
$\hat{\Theta} \leftarrow \hat{\Theta} \cup \left \{ \hat{\eta}_k \right \}$
}
}
\caption{An algorithm for multiple change point detection in model \eqref{equation: data generating process}.}
\label{algorithm: multiple change detection}
\end{algorithm}

\subsection{Theory}

In this section we investigate the consistency and localization rate of Algorithm~\ref{algorithm: multiple change detection}. In order to prove consistency of Algorithm~\ref{algorithm: multiple change detection} we need to impose the following assumption, which generalizes Assumption~\ref{assumption: change point spacing} to the multiple change point setting. 

\begin{assumption}
For each $k = 1, \dots, K$, write $\delta_k = \min \left ( \eta_{k} - \eta_{k-1}, \eta_{k+1} - \eta_k \right )$. Moreover put $\underline{\delta} = \min_{k=1,\dots,K} \delta_k$. Letting $L$ be as in Assumption~\ref{assumption: production frontier function continuous}. It holds that
\begin{enumerate}[(i)]
    \item there is a constant $C_{x,d} > 0$ for which $\underline{\delta} \geq C_{x,d} \log (n)$, and 
    \item $\delta_k^{\frac{1}{d+1}} \log (1/\mu_k) \geq \frac{4 L}{x_0} \log^{3/2} (n)$ for each $k = 1, \dots, K$. 
\end{enumerate}
 
\label{assumption: multiple change point spacing}
\end{assumption}

Similar to Assumption~\ref{assumption: change point spacing}, both (i) and (ii) in Assumption~\ref{assumption: multiple change point spacing} are satisfied with sufficiently large $n$, if $\underline{\delta} = \mathcal{O} (n)$ and $\log (1 / \mu_k ) = \mathcal{O} (1)$ for each $k = 1, \dots, K$, which is the standard setup in the infill asymptotics analysis. With this assumption in place, we have the following result: 

\begin{theorem} 
Let $\left \{ \left ( \boldsymbol{X}_t, Y_t \right ) \mid t = 1, \dots, n \right \}$ be a sample from \eqref{equation: data generating process} with production frontier functions satisfying Scenario~\ref{scenario: multiple changes}, as well as Assumptions~\ref{assumption: production frontier function continuous} and~\ref{assumption: multiple change point spacing} and stochastic components satisfying Assumptions~\ref{assumption: Z is subset of hypercube},~\ref{assumption: production set is free disposal},~\ref{assumption: X and R distribution}. If $\boldsymbol{x}_0$ is chosen such that $\max_{t = 1, \dots, n} \mathbb{P} \left ( \boldsymbol{X}_t \leq \boldsymbol{x}_0 \right ) \leq C_1$ for some $C_1 \in (0,1)$ and the threshold $\lambda$, effective sample sizes $\delta_k$, and jump size $\log (1/\mu_k)$ jointly satisfy
\begin{equation}
C_2 \log (n) \leq \lambda \leq C_3 \delta_k \log (1/\mu_k) \quad \text{ for all } k = 1, \dots, K
\label{equation: multiple change threshold thresh condition}
\end{equation}
then for $n$ sufficiently large, on a set with probability at least $1 - C_4 n^{-1}$, it holds  that
\begin{enumerate}[(i)]
    \item $\hat{K} = K$, and 
    \item $\left | \hat{\eta}_k - \eta_k \right | \leq C_5 \log (n) \vee C_6 \frac{\log(n)}{\log (1/\mu_k)}$ for each $k = 1, \dots, K$,
\end{enumerate}
where $C_1,\ldots, C_6$ are absolute constants depending only on the constants stated in Assumptions~\ref{assumption: Z is subset of hypercube},~\ref{assumption: production set is free disposal},~\ref{assumption: X and R distribution},~\ref{assumption: production frontier function continuous}, and~\ref{assumption: multiple change point spacing}. 
\label{theorem: multiple change detection}
\end{theorem}

Note that similarly to what was done in  Section~\ref{section: inference on a single change}, one can construct asymptotically valid confidence intervals for each of the change point locations returned by Algorithm~\ref{algorithm: multiple change detection} by applying the results of Theorem~\ref{lemma: single change inference} to each of the change point locations in turn.

\section{Local changes in technology} \label{section: local changes in technology}

\subsection{Problem setup}

In this section, we show how the procedure introduced in Section~\ref{section: multiple change points} can be extended to the setting in which large changes in technology occur only over a subset of the input space. In practice, changes in technology may not manifest themselves over the entire input space. For example, a particular new technology might only lead to more efficient production when the input is sufficiently large. This setting is characterized by the following scenario: 
\begin{scenario}{L}
\label{scenario: local change jump sizes}
$\Psi_{(1)} \subset \cdots \subset \Psi_{(K+1)}$. For each $k \in \left \{ 1, \dots, K \right \}$ there is a $\mu_k \in (0,1)$ and a hyper-cube $\mathcal{X}_{(k)} \subseteq \mathcal{X}$ with all sides of equal length $\chi_{(k)} > 0$ such that $f_{(k)} (\boldsymbol{x}) / f_{(k+1)} (\boldsymbol{x}) \le \mu_k$ for all $\boldsymbol{x} \in \mathcal{X}_{(k)}$.
\end{scenario}
See also Figure~\ref{figure: local change illustration} for a visual illustration of Scenario~\ref{scenario: local change jump sizes}.  In this scenario, the procedure outlined in Section~\ref{section: multiple change points} may perform poorly, since estimated efficiency scores before a change point may still attain values close to $1$, if the corresponding input lies in a region of the input space over which the production frontier function does not change. This makes it impossible for the vanilla form of the method to detect changes in the technology, which would require the estimated efficiency scores to be bounded away from 1 before the change.

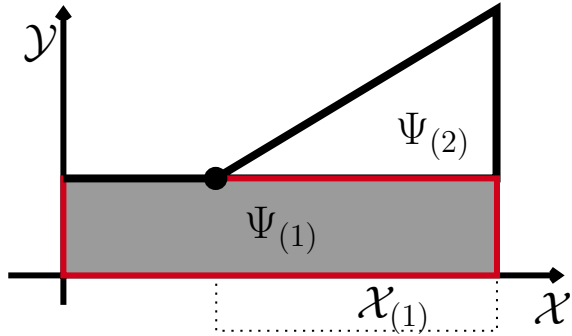
\begin{figure}[!htbp]
\centering
\caption{Visual illustration of local changes in technology with $K =1$ change points.  The grey region (\textcolor{gray}{$\blacksquare$}) corresponds to the production set under the old technology $\Psi_{(1)}$ and the white region ($\square$) corresponds to the production set under the new technology $\Psi_{(2)}$. The dotted rectangle $\mathcal{X}_{(1)}$ corresponds to the subset of the input space over which gains in productivity are realised under the new technology.
}
\label{figure: local change illustration}

\tikzset{every picture/.style={line width=0.75pt}} 

\begin{tikzpicture}[x=0.75pt,y=0.75pt,yscale=-0.6,xscale=0.6]

\draw  [dash pattern={on 0.84pt off 2.51pt}] (244,160.66) -- (477.57,160.66) -- (477.57,288.66) -- (244,288.66) -- cycle ;
\draw  [line width=3]  (477.2,19.17) -- (243.57,160.95) -- (477.2,160.95) -- cycle ;
\draw [line width=2.25]  (71,242.25) -- (529.57,242.25)(116.86,19.95) -- (116.86,266.95) (522.57,237.25) -- (529.57,242.25) -- (522.57,247.25) (111.86,26.95) -- (116.86,19.95) -- (121.86,26.95)  ;
\draw  [color={rgb, 255:red, 208; green, 2; blue, 27 }  ,draw opacity=1 ][fill={rgb, 255:red, 155; green, 155; blue, 155 }  ,fill opacity=1 ][line width=2.25]  (116.86,160.95) -- (477.57,160.95) -- (477.57,242.25) -- (116.86,242.25) -- cycle ;
\draw [line width=3]    (116.86,160.95) -- (243.57,160.95) ;
\draw [shift={(243.57,160.95)}, rotate = 0] [color={rgb, 255:red, 0; green, 0; blue, 0 }  ][fill={rgb, 255:red, 0; green, 0; blue, 0 }  ][line width=3]      (0, 0) circle [x radius= 6.37, y radius= 6.37]   ;

\draw (81,29.4) node [anchor=north west][inner sep=0.75pt]  [font=\LARGE]  {$\mathcal{Y}$};
\draw (505,254.4) node [anchor=north west][inner sep=0.75pt]  [font=\LARGE]  {$\mathcal{X}$};
\draw (265,182.4) node [anchor=north west][inner sep=0.75pt]  [font=\LARGE]  {$\Psi _{( 1)}$};
\draw (391,103.4) node [anchor=north west][inner sep=0.75pt]  [font=\LARGE]  {$\Psi _{( 2)}$};
\draw (358,249.4) node [anchor=north west][inner sep=0.75pt]  [font=\LARGE]  {$\mathcal{X}_{( 1)}$};

\end{tikzpicture}
\end{figure}

\subsection{Methodology} \label{section: local change methodology}

We now outline a procedure for multiple change point detection under Scenario~\ref{scenario: local change jump sizes}, which is a modification based on that from Section~\ref{section: multiple change points}. The main idea is to define a multi-scale grid over the input space $\mathcal{X}$, and then for every element in the grid,  search for a change as was done in Section~\ref{section: multiple change points} by restricting ourselves to observations whose inputs fall in these given sets. 

Recall that $\mathcal{X} = \times_{j=1}^{d} \mathcal{X}_j$ with $\mathcal{X}_j = \left [ 0, \overline{x}_j \right]$ and  $\underline{x} = \min_j \overline{x}_j$. So as not to overburden the notation, in the following, we assume that $\overline{x}_1=\cdots=\overline{x}_d = \underline{x}$. Note that we can always rescale the input to make this happen without affecting our procedure, because both the FDH and change detection in the support of a distribution are invariant to scaling.
To this end, for some strictly increasing sequence $\left \{ A_n \mid n \in \mathbb{N} \right \}$, we introduce the following grid 
\begin{align}
& \mathfrak{X}  = \bigcup_{k = 0}^{\left \lceil \log_{2} \left ( \underline{x} A_n^{1/d} \right ) \right \rceil} \Big\{ \mathcal{X}' = \times_{j=1}^d \left [ x_j \underline{x} 2^{-k}, \left ( x_j + 1 \right ) \underline{x} 2^{-k}  \right ] \Big| \Big . \nonumber \\
& \hspace{11em} \Big . \left ( x_1, \dots, x_d \right )^\top \in  \frac{1}{2} \mathbb{N}^{d}, \mathcal{X}' \subseteq \mathcal{X} \setminus \left \{ \boldsymbol{x} \mid  \boldsymbol{x} \leq \boldsymbol{x}_0 \right \} \Big \},
\label{equation: multi-scale grid}
\end{align}
which can be recognised as the multi-scale partition of $\mathcal{X}$ into axis-aligned hyper-cubes with sides of equal length decaying dyadically from $\underline{x}$ to $A_n^{-1/d}$. See Figure~\ref{figure: multiscale grid illustration} for a visual illustration. See also \cite{hotz2012locally} for a similar construction. Consequently, we introduce an analogue of the detection statistic defined in \eqref{equation: localized quasi-LR stat} for data whose input lies in a particular element of \eqref{equation: multi-scale grid}. That is, with integers $t_1,t_2$ such that $1 \leq t_1 < t_2 \leq n$ for any $\tau \in \{ t_1, \dots, t_2 \}$ and some $\mathcal{X}' \in \mathfrak{X}$ we put
\begin{equation}
\hat{L}_{t_1, \tau, t_2}^{\mathcal{X}'} = - 2 N \left ( t_1, \tau, \mathcal{X}' \right ) \log \left [ \hat{M} (t_1, \tau, \mathcal{X}') \right ] 
\label{equation: localized quasi-LR stat on mathcal X}
\end{equation}
where we follow the convention $0 \times \infty = 0$, and for $1 \leq t_1 \leq t_2 \leq n$ we have put 
\begin{equation*}
N \left ( t_1 , t_2, \mathcal{X}' \right ) = \sum_{t=t_1}^{t_2} \mathbf{1}_{\left \{ \boldsymbol{X}_t \in \mathcal{X}' \right \}} \quad \text{and} \quad \hat{M} \left ( t_1, t_2, \mathcal{X}' \right ) = \max \left ( \hat{R}_t \mid t_1 \leq t \leq t_2, X_t \in \mathcal{X}' \right );     
\end{equation*}
in particular, the first quantity counts the number of inputs indexed between $t_1$ and $t_2$ taking values in the set $\mathcal{X}'$, and the second quantity takes the maximum of $\hat{R}$'s whose corresponding $\boldsymbol{X}$ takes values in the same set. 
   
\begin{figure}[!htbp]
\centering
\caption{Illustration of the axis aligned hyper-cubes used in the construction of the multi-scale grid \eqref{equation: multi-scale grid} at scales $k=0,1,2$. Each dashed cube represents a given scale, and the grey and white cube contained within represent the next smallest scale.}
\label{figure: multiscale grid illustration}

\tikzset{every picture/.style={line width=0.75pt}} 

\begin{tikzpicture}[x=0.75pt,y=0.75pt,yscale=-0.6,xscale=0.6]

\draw  [fill={rgb, 255:red, 155; green, 155; blue, 155 }  ,fill opacity=0.5 ] (154.59,58.25) -- (187.6,25.24) -- (264.61,25.24) -- (264.61,120.59) -- (231.61,153.59) -- (154.59,153.59) -- cycle ; \draw   (264.61,25.24) -- (231.61,58.25) -- (154.59,58.25) ; \draw   (231.61,58.25) -- (231.61,153.59) ;
\draw  [fill={rgb, 255:red, 155; green, 155; blue, 155 }  ,fill opacity=0.5 ] (21,102.25) -- (54.01,69.25) -- (131.02,69.25) -- (131.02,164.6) -- (98.01,197.6) -- (21,197.6) -- cycle ; \draw   (131.02,69.25) -- (98.01,102.25) -- (21,102.25) ; \draw   (98.01,102.25) -- (98.01,197.6) ;
\draw  [fill={rgb, 255:red, 155; green, 155; blue, 155 }  ,fill opacity=0.5 ] (61.87,58.25) -- (94.87,25.24) -- (171.88,25.24) -- (171.88,120.59) -- (138.88,153.59) -- (61.87,153.59) -- cycle ; \draw   (171.88,25.24) -- (138.88,58.25) -- (61.87,58.25) ; \draw   (138.88,58.25) -- (138.88,153.59) ;
\draw  [fill={rgb, 255:red, 155; green, 155; blue, 155 }  ,fill opacity=0.5 ] (21,209.13) -- (54.01,176.12) -- (131.02,176.12) -- (131.02,271.47) -- (98.01,304.48) -- (21,304.48) -- cycle ; \draw   (131.02,176.12) -- (98.01,209.13) -- (21,209.13) ; \draw   (98.01,209.13) -- (98.01,304.48) ;
\draw  [fill={rgb, 255:red, 155; green, 155; blue, 155 }  ,fill opacity=0.5 ] (61.87,165.12) -- (94.87,132.12) -- (171.88,132.12) -- (171.88,227.46) -- (138.88,260.47) -- (61.87,260.47) -- cycle ; \draw   (171.88,132.12) -- (138.88,165.12) -- (61.87,165.12) ; \draw   (138.88,165.12) -- (138.88,260.47) ;
\draw  [fill={rgb, 255:red, 155; green, 155; blue, 155 }  ,fill opacity=0.5 ] (113.73,209.13) -- (146.74,176.12) -- (223.75,176.12) -- (223.75,271.47) -- (190.74,304.48) -- (113.73,304.48) -- cycle ; \draw   (223.75,176.12) -- (190.74,209.13) -- (113.73,209.13) ; \draw   (190.74,209.13) -- (190.74,304.48) ;
\draw  [fill={rgb, 255:red, 155; green, 155; blue, 155 }  ,fill opacity=0.5 ] (154.59,165.12) -- (187.6,132.12) -- (264.61,132.12) -- (264.61,227.46) -- (231.61,260.47) -- (154.59,260.47) -- cycle ; \draw   (264.61,132.12) -- (231.61,165.12) -- (154.59,165.12) ; \draw   (231.61,165.12) -- (231.61,260.47) ;
\draw  [fill={rgb, 255:red, 255; green, 255; blue, 255 }  ,fill opacity=0.9 ] (113.73,102.25) -- (146.74,69.25) -- (223.75,69.25) -- (223.75,164.6) -- (190.74,197.6) -- (113.73,197.6) -- cycle ; \draw   (223.75,69.25) -- (190.74,102.25) -- (113.73,102.25) ; \draw   (190.74,102.25) -- (190.74,197.6) ;
\draw  [dash pattern={on 0.84pt off 2.51pt}] (10,85.49) -- (90.16,5.33) -- (277.18,5.33) -- (277.18,236.89) -- (197.03,317.05) -- (10,317.05) -- cycle ; \draw  [dash pattern={on 0.84pt off 2.51pt}] (277.18,5.33) -- (197.03,85.49) -- (10,85.49) ; \draw  [dash pattern={on 0.84pt off 2.51pt}] (197.03,85.49) -- (197.03,317.05) ;
\draw  [fill={rgb, 255:red, 155; green, 155; blue, 155 }  ,fill opacity=0.5 ] (431.03,190.7) -- (446.12,175.61) -- (481.34,175.61) -- (481.34,219.21) -- (466.24,234.3) -- (431.03,234.3) -- cycle ; \draw   (481.34,175.61) -- (466.24,190.7) -- (431.03,190.7) ; \draw   (466.24,190.7) -- (466.24,234.3) ;
\draw  [fill={rgb, 255:red, 155; green, 155; blue, 155 }  ,fill opacity=0.5 ] (369.94,210.82) -- (385.03,195.73) -- (420.24,195.73) -- (420.24,239.33) -- (405.15,254.43) -- (369.94,254.43) -- cycle ; \draw   (420.24,195.73) -- (405.15,210.82) -- (369.94,210.82) ; \draw   (405.15,210.82) -- (405.15,254.43) ;
\draw  [fill={rgb, 255:red, 155; green, 155; blue, 155 }  ,fill opacity=0.5 ] (388.62,190.7) -- (403.71,175.61) -- (438.93,175.61) -- (438.93,219.21) -- (423.84,234.3) -- (388.62,234.3) -- cycle ; \draw   (438.93,175.61) -- (423.84,190.7) -- (388.62,190.7) ; \draw   (423.84,190.7) -- (423.84,234.3) ;
\draw  [fill={rgb, 255:red, 155; green, 155; blue, 155 }  ,fill opacity=0.5 ] (369.94,259.7) -- (385.03,244.6) -- (420.24,244.6) -- (420.24,288.21) -- (405.15,303.3) -- (369.94,303.3) -- cycle ; \draw   (420.24,244.6) -- (405.15,259.7) -- (369.94,259.7) ; \draw   (405.15,259.7) -- (405.15,303.3) ;
\draw  [fill={rgb, 255:red, 155; green, 155; blue, 155 }  ,fill opacity=0.5 ] (388.62,239.57) -- (403.71,224.48) -- (438.93,224.48) -- (438.93,268.08) -- (423.84,283.17) -- (388.62,283.17) -- cycle ; \draw   (438.93,224.48) -- (423.84,239.57) -- (388.62,239.57) ; \draw   (423.84,239.57) -- (423.84,283.17) ;
\draw  [fill={rgb, 255:red, 155; green, 155; blue, 155 }  ,fill opacity=0.5 ] (412.34,259.7) -- (427.43,244.6) -- (462.65,244.6) -- (462.65,288.21) -- (447.56,303.3) -- (412.34,303.3) -- cycle ; \draw   (462.65,244.6) -- (447.56,259.7) -- (412.34,259.7) ; \draw   (447.56,259.7) -- (447.56,303.3) ;
\draw  [fill={rgb, 255:red, 155; green, 155; blue, 155 }  ,fill opacity=0.5 ] (431.03,239.57) -- (446.12,224.48) -- (481.34,224.48) -- (481.34,268.08) -- (466.24,283.17) -- (431.03,283.17) -- cycle ; \draw   (481.34,224.48) -- (466.24,239.57) -- (431.03,239.57) ; \draw   (466.24,239.57) -- (466.24,283.17) ;
\draw  [fill={rgb, 255:red, 255; green, 255; blue, 255 }  ,fill opacity=0.9 ] (412.34,210.82) -- (427.43,195.73) -- (462.65,195.73) -- (462.65,239.33) -- (447.56,254.43) -- (412.34,254.43) -- cycle ; \draw   (462.65,195.73) -- (447.56,210.82) -- (412.34,210.82) ; \draw   (447.56,210.82) -- (447.56,254.43) ;
\draw  [dash pattern={on 0.84pt off 2.51pt}] (358.73,207.06) -- (397.2,168.59) -- (486.95,168.59) -- (486.95,273.91) -- (448.48,312.37) -- (358.73,312.37) -- cycle ; \draw  [dash pattern={on 0.84pt off 2.51pt}] (486.95,168.59) -- (448.48,207.06) -- (358.73,207.06) ; \draw  [dash pattern={on 0.84pt off 2.51pt}] (448.48,207.06) -- (448.48,312.37) ;
\draw  [fill={rgb, 255:red, 155; green, 155; blue, 155 }  ,fill opacity=0.5 ] (627.03,252.54) -- (633.83,245.74) -- (649.7,245.74) -- (649.7,265.39) -- (642.9,272.18) -- (627.03,272.18) -- cycle ; \draw   (649.7,245.74) -- (642.9,252.54) -- (627.03,252.54) ; \draw   (642.9,252.54) -- (642.9,272.18) ;
\draw  [fill={rgb, 255:red, 155; green, 155; blue, 155 }  ,fill opacity=0.5 ] (599.51,261.61) -- (606.31,254.81) -- (622.18,254.81) -- (622.18,274.45) -- (615.38,281.25) -- (599.51,281.25) -- cycle ; \draw   (622.18,254.81) -- (615.38,261.61) -- (599.51,261.61) ; \draw   (615.38,261.61) -- (615.38,281.25) ;
\draw  [fill={rgb, 255:red, 155; green, 155; blue, 155 }  ,fill opacity=0.5 ] (607.93,252.54) -- (614.73,245.74) -- (630.59,245.74) -- (630.59,265.39) -- (623.8,272.18) -- (607.93,272.18) -- cycle ; \draw   (630.59,245.74) -- (623.8,252.54) -- (607.93,252.54) ; \draw   (623.8,252.54) -- (623.8,272.18) ;
\draw  [fill={rgb, 255:red, 155; green, 155; blue, 155 }  ,fill opacity=0.5 ] (599.51,283.62) -- (606.31,276.82) -- (622.18,276.82) -- (622.18,296.47) -- (615.38,303.27) -- (599.51,303.27) -- cycle ; \draw   (622.18,276.82) -- (615.38,283.62) -- (599.51,283.62) ; \draw   (615.38,283.62) -- (615.38,303.27) ;
\draw  [fill={rgb, 255:red, 155; green, 155; blue, 155 }  ,fill opacity=0.5 ] (607.93,274.56) -- (614.73,267.76) -- (630.59,267.76) -- (630.59,287.4) -- (623.8,294.2) -- (607.93,294.2) -- cycle ; \draw   (630.59,267.76) -- (623.8,274.56) -- (607.93,274.56) ; \draw   (623.8,274.56) -- (623.8,294.2) ;
\draw  [fill={rgb, 255:red, 155; green, 155; blue, 155 }  ,fill opacity=0.5 ] (618.62,283.62) -- (625.41,276.82) -- (641.28,276.82) -- (641.28,296.47) -- (634.48,303.27) -- (618.62,303.27) -- cycle ; \draw   (641.28,276.82) -- (634.48,283.62) -- (618.62,283.62) ; \draw   (634.48,283.62) -- (634.48,303.27) ;
\draw  [fill={rgb, 255:red, 155; green, 155; blue, 155 }  ,fill opacity=0.5 ] (627.03,274.56) -- (633.83,267.76) -- (649.7,267.76) -- (649.7,287.4) -- (642.9,294.2) -- (627.03,294.2) -- cycle ; \draw   (649.7,267.76) -- (642.9,274.56) -- (627.03,274.56) ; \draw   (642.9,274.56) -- (642.9,294.2) ;
\draw  [fill={rgb, 255:red, 255; green, 255; blue, 255 }  ,fill opacity=0.9 ] (618.62,261.61) -- (625.41,254.81) -- (641.28,254.81) -- (641.28,274.45) -- (634.48,281.25) -- (618.62,281.25) -- cycle ; \draw   (641.28,254.81) -- (634.48,261.61) -- (618.62,261.61) ; \draw   (634.48,261.61) -- (634.48,281.25) ;
\draw  [dash pattern={on 0.84pt off 2.51pt}] (591.71,258.17) -- (611.28,238.6) -- (656.94,238.6) -- (656.94,292.95) -- (637.37,312.52) -- (591.71,312.52) -- cycle ; \draw  [dash pattern={on 0.84pt off 2.51pt}] (656.94,238.6) -- (637.37,258.17) -- (591.71,258.17) ; \draw  [dash pattern={on 0.84pt off 2.51pt}] (637.37,258.17) -- (637.37,312.52) ;
\draw [line width=1.5]    (171.88,132.12) -- (358.66,197.66) ;
\draw [shift={(362.43,198.99)}, rotate = 199.34] [fill={rgb, 255:red, 0; green, 0; blue, 0 }  ][line width=0.08]  [draw opacity=0] (13.4,-6.43) -- (0,0) -- (13.4,6.44) -- (8.9,0) -- cycle    ;
\draw [shift={(171.88,132.12)}, rotate = 19.34] [color={rgb, 255:red, 0; green, 0; blue, 0 }  ][fill={rgb, 255:red, 0; green, 0; blue, 0 }  ][line width=1.5]      (0, 0) circle [x radius= 4.36, y radius= 4.36]   ;
\draw [line width=1.5]    (438.93,224.48) -- (577.24,272.49) ;
\draw [shift={(581.02,273.8)}, rotate = 199.14] [fill={rgb, 255:red, 0; green, 0; blue, 0 }  ][line width=0.08]  [draw opacity=0] (13.4,-6.43) -- (0,0) -- (13.4,6.44) -- (8.9,0) -- cycle    ;
\draw [shift={(438.93,224.48)}, rotate = 19.14] [color={rgb, 255:red, 0; green, 0; blue, 0 }  ][fill={rgb, 255:red, 0; green, 0; blue, 0 }  ][line width=1.5]      (0, 0) circle [x radius= 4.36, y radius= 4.36]   ;
\draw   (113.73,102.25) -- (146.74,69.25) -- (223.75,69.25) -- (223.75,164.6) -- (190.74,197.6) -- (113.73,197.6) -- cycle ; \draw   (223.75,69.25) -- (190.74,102.25) -- (113.73,102.25) ; \draw   (190.74,102.25) -- (190.74,197.6) ;
\draw   (412.34,210.82) -- (427.43,195.73) -- (462.65,195.73) -- (462.65,239.33) -- (447.56,254.43) -- (412.34,254.43) -- cycle ; \draw   (462.65,195.73) -- (447.56,210.82) -- (412.34,210.82) ; \draw   (447.56,210.82) -- (447.56,254.43) ;

\draw (68.26,318.81) node [anchor=north west][inner sep=0.75pt]  [font=\LARGE]  {$k=0$};
\draw (366.73,318.81) node [anchor=north west][inner sep=0.75pt]  [font=\LARGE]  {$k=1$};
\draw (577.92,316.48) node [anchor=north west][inner sep=0.75pt]  [font=\LARGE]  {$k=2$};

\end{tikzpicture}
\end{figure}
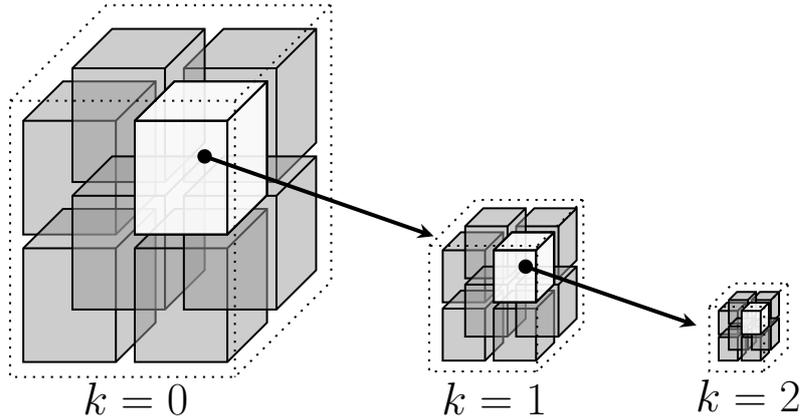

With the above definitions in place, we are in the position to describe in full the procedure for detecting multiple change points under scenario~\ref{scenario: local change jump sizes}. Similar to Algorithm~\ref{algorithm: multiple change detection}, to estimate the rightmost change point we take the arg-max of \eqref{equation: localized quasi-LR stat on mathcal X} over left-expanding intervals. However, within each interval we take the supremum of the corresponding detection statistic over all elements in the grid \eqref{equation: multi-scale grid}. Therefore, putting $j^* = \min \left \{ 1 \leq j \leq n \mid \max_{n-j \leq \tau \leq n} \max_{\mathcal{X}' \in \mathfrak{X}} \hat{L}_{n-j, \tau, n}^{\mathcal{X}'} \right \}$ the rightmost change point is estimated via
\begin{equation*}
\tilde{\eta} = \argmax_{n-j^* \leq \tau \leq n} \max_{\mathcal{X}' \in \mathfrak{X}} \hat{L}_{n-j^*, \tau, n}^{\mathcal{X}'}. 
\end{equation*}
If $j^*$ is undefined, we conclude that the data does not contain any change points. Otherwise, when such a $j^*$ is found, similar to Algorithm~\ref{algorithm: multiple change detection} we re-start the search for the next change point recursively at the index $m = n - j^*$. Let $\tilde{\Theta} = \left \{ \tilde{\eta}_1, \dots, \tilde{\eta}_{\hat{K}} \right \}$ be the (ordered) set of estimated change point locations recovered obtained by performing the above steps until either no data is left or no local statistic has value larger than $\lambda$. In an attempt to improve finite sample performance, we apply the following local re-fitting step. Let $(\tilde{s}_k, \tilde{e}_k)_{k=1,\dots,\hat{K}}$ be as defined in \eqref{equation: s e refit points}. For $k = 1, \dots, \hat{K}$, each estimated change point is locally refitted as
\begin{equation}
\hat{\eta}_k =  \argmax_{\tilde{s}_k \leq \tau \leq \tilde{e}_k} \max_{\mathcal{X}' \in \mathfrak{X}} \hat{L}_{\tilde{s}_k,\tau,\tilde{e}_k}^{\mathcal{X}'}, 
\label{equation: local change cpt extimator}
\end{equation}
where, again, at each $k$, the local statistics are computed using scores obtained from the estimated production frontier function which was used to detect the corresponding $\tilde{\eta}_k$. 

\subsection{Theory for local changes under Scenario~\ref{scenario: local change jump sizes}} \label{section: local change theory}

In this Section we show that, with appropriately chosen tuning parameters, under Scenario~\ref{scenario: local change jump sizes} the procedure described in Section~\ref{section: local change methodology} is able to detect all change point locations with high probability and moreover that all change points are well localized. To achieve a good rate of change point localization  under Scenario~\ref{scenario: local change jump sizes}, we need to impose the following assumption:    
\begin{assumption}
With $\underline{\delta} $ defined as in Assumption \ref{assumption: multiple change point spacing} and $L$ as in Assumption~\ref{assumption: production frontier function continuous}, it holds that
\begin{enumerate}[(i)]
    \item there is a positive constant $C_\chi$  for which $\min_k \operatorname{mes} \left ( \mathcal{X}_{(k)} \right ) \geq C_\chi \log^{-1} (n)$,
    \item there is a positive constant $C_{x,d}$ for which $\underline{\delta} \geq C_{x,d} \log^3 (n)$, and 
    \item $\delta_k^{\frac{1}{d+1}} \log (1/\mu_k) \geq \frac{4 L}{x_0} \log^{3/2} (n)$ for each $k = 1, \dots, K$. 
\end{enumerate}
\label{assumption: size of regions}
\end{assumption}

Part (i) of Assumption~\ref{assumption: size of regions} allows the measure of each region of the input space on which a prominent change occurs to vanish as $n \rightarrow \infty$, but maintains that the rate of decay is slow enough that with $n$ sufficiently large at least one element in $\mathfrak{X}$ will be contained within the region.  Part (ii) controls the spacing the change points, while Part (iii) ensures that, with high probability, for each $k = 1, \dots, K$ enough data points fall in an element of $\mathfrak{X}$ contained within $\mathcal{X}_{(k)}$, which matches Part (ii) of Assumption \ref{assumption: multiple change point spacing}. With this assumption in place, we have the following result: 

\begin{theorem}
Let $\left \{ \left ( \boldsymbol{X}_t, Y_t \right ) \mid t = 1, \dots, n \right \}$ be a sample from \eqref{equation: data generating process} with production frontier functions satisfying Scenario~\ref{scenario: local change jump sizes}, as well as Assumptions~\ref{assumption: size of regions} and \ref{assumption: production frontier function continuous},  and stochastic components satisfying Assumptions~\ref{assumption: Z is subset of hypercube},~\ref{assumption: production set is free disposal},~\ref{assumption: X and R distribution}. If $\boldsymbol{x}_0$ is chosen such that $\max_{t = 1, \dots, n} \mathbb{P} \left ( \boldsymbol{X}_t \leq \boldsymbol{x}_0 \right ) \leq C_1$ for some $C_1 \in (0,1)$, $A_n$ is chosen such that $A_n \leq C_2 \log(n)$, and the threshold $\lambda$, effective sample sizes $\delta_k$, and jump sizes $\log (1/\mu_k)$ jointly satisfy
\begin{equation}
C_3 \log (n) \leq \lambda \leq C_4 \left ( \operatorname{mes} (\mathcal{X}_{(k)}) \right ) \delta_k \log (1/\mu_k) \quad \text{ for all } k = 1, \dots, K,
\label{equation: multiple change threshold thresh condition 2}
\end{equation}
then for any sufficiently large $n$, on a set with probability at least $1 - C_5 n^{-1} (\log(n))^3 $, it holds  that
\begin{enumerate}[(i)]
    \item $\hat{K} = K$, and 
    \item $\left | \hat{\eta}_k - \eta_k \right | \leq C_6 \log^2 (n) \vee C_7 \log^2 (n) \left ( \left ( \operatorname{mes} (\mathcal{X}_{(k)}) \right ) \log \left ( 1/\mu_k \right ) \right )^{-1} $ for each $k = 1, \dots, K$,
\end{enumerate}
where $C_1, \ldots, C_7$ are absolute constants depending only on the constants stated in Assumptions \ref{assumption: Z is subset of hypercube}, \ref{assumption: production set is free disposal}, \ref{assumption: X and R distribution}, \ref{assumption: production frontier function continuous}, and \ref{assumption: size of regions}. 
\label{theorem: consistent detection under local alternatives}
\end{theorem} 

As before, we also investigate the global information theoretic lower bounds on the rate at which the change points in model \eqref{equation: sequence of production sets} can be localized under local changes in technology. To this end, we have the following result:

\begin{proposition}
Let $\left \{ \left ( \boldsymbol{X}_t, Y_t \right ) \mid t = 1, \dots, n \right \}$ be a sample from \eqref{equation: data generating process} with production frontier functions satisfying Assumption~\ref{assumption: production frontier function continuous} and stochastic components satisfying Assumptions~\ref{assumption: Z is subset of hypercube},~\ref{assumption: production set is free disposal},~\ref{assumption: X and R distribution}  having a single change point at location $\eta$, effective sample size $\delta = \min \left ( \eta, n - \eta \right )$, and change size $\mu$ occurring on a region $\mathcal{X}' \subseteq \mathcal{X}$ of the input space. Let $P_{n, \delta, \mu, \mathcal{X}'}$ denote the joint law of the data and consider the class
\begin{equation*}
\mathcal{Q}_n = \left \{ P_{n, \delta, \mu, \mathcal{X}'} \mid \delta < n / 2 \text{ and } \delta \operatorname{mes} \left ( \mathcal{X}' \right ) \log \left ( 1 / \mu \right ) > \zeta_n \right \}
\end{equation*}
for any sequence $\left \{ \zeta_n \mid n > 0 \right \}$ such that $\lim_{n \rightarrow \infty} \zeta_n = \infty$, and where $\operatorname{mes} \left ( \mathcal{A} \right )$ denotes the Lebesgue measure of a set $\mathcal{A}$. Then, for all $n$ sufficiently large it holds that
\begin{equation*}
\inf_{\hat{\eta}} \sup_{P \in \mathcal{Q}_n} \mathbb{E}_{P} \left [ \left | \hat{\eta} - \eta \right | \right ] \geq \frac{C}{\operatorname{mes} \left ( \mathcal{X}' \right ) \log (1/\mu)} 
\end{equation*}
where the infimum is over all measurable functions of the data, $\eta (P)$ denotes the change point location of $P \in \mathcal{Q}_n$, and $C > 0$ is an absolute constant depending on the constants in Assumption~\ref{assumption: production frontier function continuous} and Assumptions~\ref{assumption: Z is subset of hypercube},~\ref{assumption: production set is free disposal},~\ref{assumption: X and R distribution}. 
\label{lemma: local change minimax lower bound}
\end{proposition}

In light of Proposition~\ref{lemma: local change minimax lower bound}, we conclude that the rate at which change points are able to be localized by the procedure described in Section~\ref{section: local change methodology} is minimax optimal up to a log factor.

\section{Simulation studies} \label{section: simulation studies}

In this section, we assess the performance of our proposed methodology on several synthetic data examples. We compare the proposed method with existing approaches from the nonparametric change point detection literature, since one conceivable alternative for detecting changes in the production frontier function in a nonparametric fashion would involve looking for changes in the joint distribution of inputs and outputs.  Specifically, we compare the proposed approach with the moving sum algorithm based on distances between joint characteristic functions (MOJO) proposed by \cite{mcgonigle2025nonparametric}, the E-divisive (ECP) procedure of \cite{matteson2014nonparametric}, and Kernel change point detection (KCP) procedure of \cite{celisse2018new} and \cite{arlot2019kernel} as implemented by \cite{cabrieto2022r}. In the sequel we refer to our methodology as FCP (\underline{f}rontier \underline{c}hange \underline{p}oint detection) for short. MOJO, ECP, and KCP are calibrated via bootstrap, and in each case we set the number of bootstrap replications to $100$. For FCP, we set the threshold parameter to be $\lambda = \log^2(n)$, where $n$ is the number of data points observed, and the trimming parameter $\boldsymbol{x}_0$ to be the coordinate-wise $\alpha$-th quantile of the inputs using $\alpha = 0.1$. All of the results in this section can be reproduced using the \texttt{R} code available at \url{https://anonymous.4open.science/r/frontier-change-num-ex-978B/}

\subsection{Simulation setup} \label{section: simulation studies M}

We simulate $n = 1000$ data points from model \eqref{equation: data generating process} with $K \in \{ 2, 3, 4 \}$ change points being equidistant (subject to rounding) from each other and from the start and end time at $t=1$ and $t=n$. The production frontier function on each stationary segment is given by $f_{(k)} (\cdot) = 1.75 \times f_{(k-1)} (\cdot)$ for each $k = 2, \dots, K +1$ (i.e. corresponding to $\mu_k = 1/1.75 \approx 0.57$) and with $\boldsymbol{x} = (x_1, \dots, x_d)^\top \in \mathcal{X}$ the baseline production frontier function $f_{(1)} (\boldsymbol{x})$ is chosen as follows: 
\begin{enumerate}[(M1)]
    \item Constant: $f (\boldsymbol{x}) = A$.
    \item Additive: $f (\boldsymbol{x}) = A + \sum_{i=1}^d \alpha_i x_i $.
    \item Cobb–Douglas: $f (\boldsymbol{x}) = A \prod_{i=1}^d x_i^{\alpha_i} $. 
    \item Logistic: $f (\boldsymbol{x}) = e^{A \left (\sum_{i=1}^d \alpha_i x_i -B \right )} \times \left ( 1 + e^{A \left (\sum_{i=1}^d \alpha_i x_i -B \right )} \right )^{-1}$. 
\end{enumerate}
The values of $A$, $B$, and $\alpha_i$'s are reported in Table~\ref{table: parameter values} below.

\begin{table}[!ht]
\centering
\caption{Parameter values set for models (M1) - (M5) used in Section~\ref{section: simulation studies M} and Section~\ref{section: simulation studies L}}
\begin{tabular}{|l|c c c c c |}
\hline
\multirow{2}{*}{Model} & \multicolumn{5}{c|}{Parameter values} \\
 & A & B & $\omega$ & $\alpha_1$ & $\alpha_2$ \\
\hline
(M1) Constant & $1.000$ & \texttt{NA} & \texttt{NA} & \texttt{NA} & \texttt{NA} \\
(M2) Additive & $3.000$ & \texttt{NA} & \texttt{NA} & $3.000$ & $3.000$ \\
(M3) Cobb-Douglas & $1.000$ & \texttt{NA} & \texttt{NA} & $0.300$ & $0.300$ \\
(M4) Logistic & $4.000$ & $0.500$ & \texttt{NA} & $1.000$ & $1.000$ \\
(M5) Piecewise Linear; \emph{d = 1} & 1.000 & -0.875 & 1.250 & 1.500 & \texttt{NA} \\
(M5) Piecewise Linear; \emph{d = 2} & 1.000 & -2.750 & 2.500 & 1.500 & 1.500 \\
\hline
\end{tabular}
\label{table: parameter values}
\end{table}

We consider inputs with dimensions $d \in \{1,2\}$. For each $t = 1, \dots, n$, the inputs are mutually independent uniform random variables with support $[1,2]$, and the efficiency scores are drawn from one of the following distributions: 
\begin{enumerate}[(R1)]
    \item Uniform distribution with support $[0,1]$ for all $t$. 
    \item Truncated Normal distribution $N(\mu_t,\sigma^2)$ on $[0,1]$ with $\mu_t = 0.5 + t/n$ and $\sigma^2 = 0.1$. 
    \item Truncated Normal distribution $N(\mu_t,\sigma^2)$ on $[0,1]$ with  $\mu_t = 1.5 - t / n$ and $\sigma^2 = 0.1$.
\end{enumerate}
Note that simulations with efficiency scores distributed as (R1) model the scenario where the level of technology increases over time while efficiency level remains constant, simulations with efficiency scores distributed as (R2) model the scenario where both technology and efficiency increase over time, and simulations with efficiency scores distributed as (R3) model the scenario where the level of technology increases over time while the level of efficiency decreases. 

\subsection{Simulation results}

We simulate data from each combination of production frontier function and efficiency score distribution $500$ times, and on each iteration, record the Hausdorff distance between the change points recovered by each change point detection algorithm and the truth, as well as the absolute difference between the number of estimated and true change points by each algorithm. The results of the simulation study are presented in Table~\ref{table: d1 simulation study} for the case $d = 1$ and in Table~\ref{table: d2 simulation study} for the case $d = 2$. 


We observe that in the case of one-dimensional inputs, FCP uniformly outperforms all off-the-shelf nonparametric change point detection methods which we compare to. In the case of two-dimensional inputs, FCP remains the best performing method in terms of detecting all changes present in the data; however, the localization rate is occasionally slightly worse than that of ECP as evidenced by the Hausdorff distance measurement.  We believe that this is due to the fully nonparametric nature of the FDH. In practice, if one is willing to impose additional continuity or shape constraints, the localization rate can be improved by smoothing or post-processing the FDH estimator. For instance, one may replace the FDH estimator with the data envelopment analysis estimator of \cite{charnes1978measuring}, which corresponds to taking the convex hull of the FDH estimator and offers better finite sample performance at the cost of the additional concavity assumption on the shape of the frontier. Finally, we remark that here all the competitors are based on detecting changes nonparametrically in the distributions, so tend to over-estimate the number of changes in (R2) and (R3), in which the efficiency score distributions vary slowly. Their performance would also deteriorate significantly in the settings where there is no change in the frontiers, but with  pronounced changes in the joint distribution of the inputs and outputs.

\subsection{Simulation studies under Scenario~\ref{scenario: local change jump sizes}} \label{section: simulation studies L}

Here we present simulation studies under Scenario~\ref{scenario: local change jump sizes}, in which the production frontier function only changes over certain localized regions of the input space. We investigate the performance of the proposed extension in Section~\ref{section: local change methodology} for dealing with local changes in technology, which we will refer to as MS-FCP (\underline{m}ulti-\underline{s}cale \underline{f}rontier \underline{c}hange \underline{p}oint detection), for short. We  additionally choose the distribution of the efficiency scores to illustrate how off-the-shelf nonparametric change point detection procedures can again fail at the task of detecting changes in production frontiers.

We simulate $n = 1000$ data points from model \eqref{equation: data generating process} with $K = 2$ change points being  equidistant (subject to rounding) from each other and from $t=1$ and $t=1000$.  We again consider $d \in \{ 1, 2 \}$ dimensional inputs simulated as in Section~\ref{section: simulation studies M}. The production frontier function is given by the constant function (M1) on the first stationary segment, by (M5) detailed below on the second stationary segment
\begin{enumerate}[(M5)]
    \item Piecewise Linear: $f (\boldsymbol{x}) = \begin{cases}
A & \text{ if } \left \| \boldsymbol{x} \right \|_1 < \omega  \\
B + \sum_{j=1}^d \alpha_j x_j & \text{ otherwise}
\end{cases}$,
\end{enumerate}
and by $f_{(3)}(\cdot) = 1.75 \times f_{(2)}(\cdot)$ on the third stationary segment. For (M5), the values of $A$, $B$, $\omega$, and $\alpha_i$ are reported in Table~\ref{table: parameter values} in Section~\ref{section: simulation studies M}. Note that the first change point is local in the sense of Scenario~\ref{scenario: local change jump sizes}, as efficiency gains are only realized for inputs having $\ell_1$ norm greater than or equal to $\omega$, whereas the second change is global and satisfies Scenario~\ref{scenario: multiple changes}. The efficiency scores are drawn according to the following distribution: 
\begin{enumerate}[(R4)]
\item Uniform distribution on $[0,1]$ if $t \leq n / 2$, and a mixture of uniform distributions on $[0,1]$ and $[0.8,1]$ with equal weights of a half if $t > n/2$. 
\end{enumerate}
Therefore, the joint distribution of the inputs and outputs changes at times $\left \lfloor n / 3 \right \rfloor$, $\left \lfloor n / 2 \right \rfloor$, and $\left \lfloor 2n / 3 \right \rfloor$ whereas the production frontier function only changes locally at time $\left \lfloor n / 3 \right \rfloor$, and globally at time $\left \lfloor 2n / 3 \right \rfloor$.

\begin{table}[!hb]
\centering
\caption{Performance of FCP, MOJO, ECP, and KCP on simulated data with $d=1$ dimensional inputs, over $500$ replications. See Section~\ref{section: simulation studies M} for details. In the Setup column: $K$ represents the number of change points in the data and $F_R$ represents the distribution of the efficiency scores. For each production frontier function (Constant, Additive, Cobb-Douglas, or Logistic): the quantity $d_H (\Theta, \hat{\Theta})$ represents the average Hausdorff distance between the estimated set of change points and the truth, and the quantity $d (K,\hat{K})$ represents the average absolute difference between the true number of change points and the estimated number of change points. Lowest values are in bold.}
\resizebox{\columnwidth}{!}{%
\begin{tabular}{|c|l|rc|rc|rc|rc|}
  \hline
  \multirow{2}{*}{Setup} & \multirow{2}{*}{Method} & \multicolumn{2}{c|}{Constant} & \multicolumn{2}{c|}{Additive} & \multicolumn{2}{c|}{Cobb–Douglas} & \multicolumn{2}{c|}{Logistic} \\
& & $d_H (\Theta, \hat{\Theta})$ & $d(K,\hat{K})$ & $d_H (\Theta, \hat{\Theta})$ & $d(K,\hat{K})$ & $d_H (\Theta, \hat{\Theta})$ & $d(K,\hat{K})$ & $d_H (\Theta, \hat{\Theta})$ & $d(K,\hat{K})$ \\ 
\hline 
& FCP & \textbf{2.74} & \textbf{0.00} & \textbf{2.80} & \textbf{0.00} & \textbf{2.76} & \textbf{0.00} & \textbf{2.72} & \textbf{0.00} \\ 
$K = 2$ & MOJO & 278.11 & 1.01 & 309.02 & 1.16 & 287.48 & 1.03 & 279.59 & 0.99 \\ 
$F_R = \text{(R1)}$ & ECP & 11.75 & 0.09 & 12.68 & 0.09 & 11.85 & 0.08 & 11.16 & 0.07 \\ 
& KCP & 326.04 & 0.99 & 326.38 & 0.99 & 326.58 & 0.99 & 326.67 & 0.99 \\ 
\hline
& FCP & \textbf{0.75} & \textbf{0.00} & \textbf{0.81} & \textbf{0.00} & \textbf{0.76} & \textbf{0.00} & \textbf{0.76} & \textbf{0.00} \\ 
$K = 2$ & MOJO & 1.98 & 0.00 & 7.15 & 0.03 & 1.86 & 0.00 & 1.85 & 0.00 \\ 
$F_R = \text{(R2)}$ & ECP & 71.00 & 0.83 & 48.95 & 0.54 & 64.45 & 0.72 & 68.27 & 0.79 \\ 
& KCP & 330.02 & 1.00 & 329.91 & 1.00 & 329.98 & 1.00 & 330.01 & 1.00 \\ 
\hline
& FCP & \textbf{0.76} & \textbf{0.00} & \textbf{0.79} & \textbf{0.00} & \textbf{0.79} & \textbf{0.00} & \textbf{0.78} & \textbf{0.00} \\ 
$K = 2$ & MOJO & 155.61 & 0.46 & 255.30 & 0.77 & 174.19 & 0.52 & 157.13 & 0.47 \\ 
$F_R = \text{(R3)}$ & ECP & 95.38 & 0.93 & 93.88 & 0.92 & 94.77 & 0.91 & 98.21 & 0.94 \\ 
& KCP & 158.19 & 0.48 & 168.10 & 0.51 & 158.94 & 0.48 & 161.50 & 0.49 \\ 
\hline
& FCP & \textbf{3.40} & \textbf{0.00} & \textbf{3.46} & \textbf{0.00} & \textbf{3.47} & \textbf{0.00} & \textbf{3.49} & \textbf{0.00} \\ 
$K = 3$ & MOJO & 262.44 & 1.55 & 269.56 & 1.72 & 266.87 & 1.61 & 262.13 & 1.56 \\ 
$F_R = \text{(R1)}$ & ECP & 17.50 & 0.11 & 16.35 & 0.08 & 17.65 & 0.11 & 16.91 & 0.10 \\ 
& KCP & 249.05 & 1.93 & 249.12 & 1.94 & 249.09 & 1.93 & 249.10 & 1.93 \\ 
\hline
& FCP & \textbf{1.15} & \textbf{0.00} & \textbf{1.21} & \textbf{0.00} & \textbf{1.18} & \textbf{0.00} & \textbf{1.21} & \textbf{0.00} \\ 
$K = 3$ & MOJO & 6.02 & 0.02 & 94.53 & 0.38 & 15.17 & 0.06 & 4.58 & 0.01 \\ 
$F_R = \text{(R2)}$ & ECP & 18.12 & 0.22 & 13.99 & 0.18 & 16.43 & 0.21 & 17.74 & 0.23 \\ 
& KCP & 248.84 & 1.93 & 248.76 & 1.91 & 248.78 & 1.92 & 248.80 & 1.92 \\ 
\hline
& FCP & \textbf{1.07} & \textbf{0.00} & \textbf{1.18} & \textbf{0.00} & \textbf{1.17} & \textbf{0.00} & \textbf{1.21} & \textbf{0.00} \\ 
$K = 3$ & MOJO & 248.52 & 1.01 & 248.76 & 1.05 & 247.95 & 1.01 & 248.47 & 1.01 \\ 
$F_R = \text{(R3)}$ & ECP & 40.05 & 0.54 & 36.99 & 0.48 & 39.33 & 0.52 & 38.76 & 0.51 \\ 
& KCP & 249.43 & 2.00 & 249.43 & 2.00 & 249.44 & 2.00 & 249.44 & 2.00 \\ 
\hline
& FCP & \textbf{4.56} & \textbf{0.00} & \textbf{4.73} & \textbf{0.00} & \textbf{4.66} & \textbf{0.00} & \textbf{4.73} & \textbf{0.00} \\ 
$K = 4$ & MOJO & 235.64 & 2.25 & 254.02 & 2.38 & 239.42 & 2.25 & 237.88 & 2.25 \\ 
$F_R = \text{(R1)}$ & ECP & 117.99 & 0.60 & 109.08 & 0.55 & 120.80 & 0.60 & 119.33 & 0.60 \\ 
& KCP & 215.81 & 2.95 & 213.24 & 2.93 & 214.84 & 2.94 & 215.11 & 2.95 \\ 
\hline
& FCP & \textbf{1.43} & \textbf{0.00} & \textbf{1.51} & \textbf{0.00} & \textbf{1.49} & \textbf{0.00} & \textbf{1.56} & \textbf{0.00} \\ 
$K = 4$ & MOJO & 90.10 & 0.45 & 191.23 & 1.00 & 149.69 & 0.75 & 106.83 & 0.54 \\ 
$F_R = \text{(R2)}$ & ECP & 9.22 & 0.10 & 8.50 & 0.09 & 9.80 & 0.12 & 9.04 & 0.10 \\ 
& KCP & 202.62 & 2.92 & 202.69 & 2.91 & 202.61 & 2.92 & 202.66 & 2.93 \\ 
\hline
& FCP & \textbf{1.50} & \textbf{0.00} & \textbf{1.55} & \textbf{0.00} & \textbf{1.47} & \textbf{0.00} & \textbf{1.61} & \textbf{0.00} \\ 
$K = 4$ & MOJO & 199.85 & 1.68 & 199.97 & 1.82 & 199.81 & 1.69 & 199.84 & 1.67 \\ 
$F_R = \text{(R3)}$ & ECP & 22.72 & 0.32 & 19.85 & 0.27 & 23.51 & 0.33 & 22.98 & 0.32 \\ 
& KCP & 200.72 & 3.00 & 200.77 & 3.00 & 200.72 & 3.00 & 200.72 & 3.00 \\ 
\hline
\end{tabular}
}
\label{table: d1 simulation study}
\end{table}
\clearpage

\begin{table}[!ht]
\centering
\caption{Performance of FCP, MOJO, ECP, and KCP on simulated data with $d=2$ dimensional inputs, over $500$ replications. See Section~\ref{section: simulation studies M} for details. In the Setup column: $K$ represents the number of change points in the data and $F_R$ represents the distribution of the efficiency scores. Here $d_H (\Theta, \hat{\Theta})$ represents the average Hausdorff distance between the estimated set of change points and the truth, and the quantity $d (K,\hat{K})$ represents the average absolute difference between the true and estimated number of change points. Lowest values are in bold.}

\resizebox{\columnwidth}{!}{%
\begin{tabular}{|c|l|rc|rc|rc|rc|}
  \hline
  \multirow{2}{*}{Setup} & \multirow{2}{*}{Method} & \multicolumn{2}{c|}{Constant} & \multicolumn{2}{c|}{Additive} & \multicolumn{2}{c|}{Cobb–Douglas} & \multicolumn{2}{c|}{Logistic} \\
& & $d_H (\Theta, \hat{\Theta})$ & $d(K,\hat{K})$ & $d_H (\Theta, \hat{\Theta})$ & $d(K,\hat{K})$ & $d_H (\Theta, \hat{\Theta})$ & $d(K,\hat{K})$ & $d_H (\Theta, \hat{\Theta})$ & $d(K,\hat{K})$ \\ 
  \hline
 & FCP & \textbf{6.09} & \textbf{0.00} & \textbf{9.37} & \textbf{0.02} & \textbf{8.44} & \textbf{0.01} & \textbf{6.15} & \textbf{0.00} \\ 
$K=2$ & MOJO & 251.31 & 0.89 & 302.15 & 1.10 & 278.59 & 1.02 & 246.01 & 0.89 \\ 
$F_R = \text{(R1)}$ & ECP & 11.12 & 0.07 & 12.06 & 0.08 & 12.42 & 0.08 & 11.33 & 0.07 \\ 
 & KCP & 327.53 & 1.00 & 327.44 & 1.00 & 327.46 & 1.00 & 327.53 & 1.00 \\ 
\hline
 & FCP & \textbf{1.34} & \textbf{0.00} & \textbf{2.28} & \textbf{0.00} & \textbf{2.20} & \textbf{0.00} & \textbf{1.35} & \textbf{0.00} \\ 
$K=2$ & MOJO & 1.94 & 0.01 & 4.13 & 0.02 & 2.00 & 0.01 & 1.94 & 0.01 \\ 
$F_R = \text{(R2)}$ & ECP & 73.69 & 0.84 & 55.77 & 0.60 & 57.94 & 0.61 & 68.56 & 0.78 \\ 
 & KCP & 329.92 & 1.00 & 329.95 & 1.00 & 329.96 & 1.00 & 329.92 & 1.00 \\ 
\hline
 & FCP & \textbf{2.69} & \textbf{0.00} & \textbf{4.16} & \textbf{0.00} & \textbf{3.93} & \textbf{0.00} & \textbf{2.68} & \textbf{0.00} \\ 
 $K=2$ & MOJO & 77.86 & 0.23 & 253.32 & 0.76 & 128.32 & 0.38 & 73.21 & 0.22 \\ 
 $F_R = \text{(R3)}$ & ECP & 103.50 & 0.94 & 101.89 & 0.96 & 101.27 & 0.95 & 102.90 & 0.95 \\ 
 & KCP & 205.61 & 0.62 & 215.98 & 0.65 & 211.43 & 0.64 & 205.61 & 0.62 \\ 
\hline
 & FCP & \textbf{11.00} & \textbf{0.00} & 21.03 & \textbf{0.02} & 17.90 & \textbf{0.01} & \textbf{11.19} & \textbf{0.00} \\ 
 $K=3$ & MOJO & 266.66 & 1.60 & 276.76 & 1.76 & 273.45 & 1.68 & 267.67 & 1.62 \\ 
 $F_R = \text{(R1)}$ & ECP & 15.62 & 0.08 & \textbf{17.19} & 0.08 & \textbf{16.27} & 0.07 & 16.58 & 0.09 \\ 
 & KCP & 249.45 & 1.99 & 249.39 & 1.99 & 249.39 & 1.98 & 249.42 & 1.99 \\ 
\hline
 & FCP & \textbf{3.50} & \textbf{0.00} & \textbf{5.32} & \textbf{0.00} & \textbf{5.01} & \textbf{0.00} & \textbf{3.50} & \textbf{0.00} \\ 
 $K = 3$ & MOJO & 4.20 & 0.01 & 80.47 & 0.32 & 13.88 & 0.05 & 3.81 & 0.01 \\ 
 $F_R = \text{(R2)}$ & ECP & 15.33 & 0.19 & 15.73 & 0.19 & 14.48 & 0.18 & 16.33 & 0.20 \\ 
 & KCP & 248.98 & 1.99 & 248.95 & 1.97 & 248.97 & 1.98 & 248.98 & 1.99 \\ 
\hline
 & FCP & \textbf{7.65} & \textbf{0.00} & \textbf{13.41} & \textbf{0.01} & \textbf{11.43} & \textbf{0.01} & \textbf{7.68} & \textbf{0.00} \\ 
 $K = 3$ & MOJO & 248.46 & 1.00 & 248.78 & 1.05 & 248.83 & 1.01 & 247.98 & 1.00 \\ 
 $F_R = \text{(R3)}$ & ECP & 42.72 & 0.54 & 41.67 & 0.52 & 43.03 & 0.55 & 43.41 & 0.54 \\ 
 & KCP & 249.39 & 2.00 & 249.40 & 2.00 & 249.40 & 2.00 & 249.39 & 2.00 \\ 
\hline
 & FCP & \textbf{25.45} & \textbf{0.04} & \textbf{57.97} & \textbf{0.23} & \textbf{49.34} & \textbf{0.18} & \textbf{25.60} & \textbf{0.04} \\ 
 $K = 4$ & MOJO & 229.84 & 2.26 & 253.83 & 2.38 & 245.35 & 2.37 & 229.32 & 2.27 \\ 
 $F_R = \text{(R1)}$ & ECP & 144.64 & 0.73 & 111.21 & 0.57 & 137.97 & 0.70 & 148.08 & 0.74 \\ 
 & KCP & 229.40 & 2.98 & 229.83 & 2.98 & 230.26 & 2.98 & 229.40 & 2.98 \\ 
\hline
 & FCP & \textbf{8.67} & \textbf{0.01} & 16.29 & \textbf{0.03} & 13.61 & \textbf{0.01} & \textbf{8.70} & \textbf{0.01} \\ 
 $K = 4$ & MOJO & 71.11 & 0.36 & 188.68 & 0.97 & 163.02 & 0.82 & 69.19 & 0.35 \\ 
 $F_R = \text{(R2)}$ & ECP & 11.18 & 0.14 & \textbf{8.56} & 0.09 & \textbf{8.89} & 0.09 & 10.38 & 0.13 \\ 
 & KCP & 205.02 & 2.98 & 206.15 & 2.95 & 205.40 & 2.96 & 205.01 & 2.97 \\ 
\hline
 & FCP & \textbf{17.53} & \textbf{0.04} & 40.06 & \textbf{0.16} & 34.57 & \textbf{0.13} & \textbf{17.57} & \textbf{0.04} \\ 
 $K = 4$ & MOJO & 199.65 & 1.52 & 199.98 & 1.81 & 199.71 & 1.63 & 199.73 & 1.52 \\ 
 $F_R = \text{(R3)}$& ECP & 22.52 & 0.32 & \textbf{22.70} & 0.33 & \textbf{22.16} & 0.32 & 23.67 & 0.33 \\ 
 & KCP & 200.84 & 3.00 & 200.93 & 3.00 & 200.91 & 2.99 & 200.83 & 3.00 \\ 
\hline
\end{tabular}
}
\label{table: d2 simulation study}
\end{table}
\clearpage

We repeat the simulation described above $500$ times. For MS-FCP we set the threshold and trimming parameter as was done for FCP in Section~\ref{section: simulation studies M}, and additionally set $\underline{x}A_n^{1/d} = 4$. The results are reported in Table~\ref{table: local change simualtion}. MS-FCP uniformly outperforms the competing methods with both one- and two-dimensional inputs. Unsurprisingly, the previous FCP method is able to detect the second change, which satisfies Scenario~\ref{scenario: multiple changes}, but generally cannot detect the first change, which is local. Finally, as evidenced by the relatively high average Hausdorff distances and the absolute difference between the estimated and true number of changes, the competing nonparametric change point detection methods tend to spuriously detect the change in the distribution of the efficiency scores at $\left \lfloor n / 2 \right \rfloor$, which is not of interest for the problem being studied here.

\begin{table}[!htbp]
\centering
\caption{
Performance of MS-FCP, FCP, MOJO, ECP, and KCP on simulated data with $d=1$ and $d=2$ dimensional inputs, over $500$ replications. See Section~\ref{section: simulation studies L} for details. The quantity $d_H (\Theta, \hat{\Theta})$ represents the average Hausdorff distance between the estimated set of change points and the truth, and the quantity $d (K,\hat{K})$ represents the average absolute difference between the true number of change points and the estimated number of change points. Lowest values are in bold.}
\begin{tabular}{|l|rr|rr|}
\hline
\multirow{2}{*}{Method} & \multicolumn{2}{c|}{$d = 1$} & \multicolumn{2}{c|}{$d = 2$} \\ 
& $d_H (\Theta, \hat{\Theta})$ & $d(K,\hat{K})$ & $d_H (\Theta, \hat{\Theta})$ & $d(K,\hat{K})$ \\
  \hline
MS-FCP & \textbf{13.95} & \textbf{0.00} & \textbf{57.19} & \textbf{0.29} \\ 
FCP & 306.37 & 0.89 & 230.90 & 0.59 \\ 
MOJO & 298.72 & 0.85 & 251.08 & 0.74 \\ 
ECP & 122.23 & 0.35 & 97.40 & 0.60 \\ 
KCP & 331.37 & 1.00 & 331.02 & 1.00 \\ 
   \hline
\end{tabular}
\label{table: local change simualtion}
\end{table}

\section{Real data examples} \label{section: real data example}

In this section, we illustrate the practical value of our proposed methodology by analyzing two real world datasets. In both data examples we choose the trimming parameter $\boldsymbol{x}_0$ to be entry-wise the $\alpha$-th quantile of the inputs setting $\alpha = 0.1$, and as in the simulation studies in Section~\ref{section: simulation studies M} we set the threshold parameter to $\lambda = \log^2(n)$ where $n$ is the number of data points observed. In practice, when only given a small number of observations, one might not be able to distinguish between a temporary but universal drop in the efficiencies and a down-shift of the production frontier. To resolve this identifiability issue and to reinforce the assumption of monotonic increases in technology, we analyse the data with a slight variant of Algorithm~\ref{algorithm: multiple change detection}, which is described in \ref{section: real data algorithm} of the supplementary materials. We remark that this modification comes with the same theoretical guarantees and does not lead to noticeable differences in our simulation experiments. 

\subsection{Analysis of a large Brazilian bank}  \label{section: Brazilian bank analysis}

We analyse quarterly revenue data from a large Brazilian commercial bank. The period of observation is from the first quarter of 2019 to the first quarter of 2025, and we observe financial statements from $1791$ unique branches for a total of $n = 8,604$ data points. We use the administrative expenses as the input variable (i.e. $d=1$), which are measured as an aggregate of personnel expenses, fixed expenses, and promotion and advertising expenses. For the output variable, we use the operating revenues, which are measured as an aggregate of financial revenue (revenue derived from deposit transactions with other financial institutions), credit revenue (revenue derived from credit transactions), and service revenue (revenue derived from the provision of services to the institution's members). Prior to analysis we rescale both the inputs and outputs by the number of active account holders reported by each branch in the given period. We also apply min-max scaling to the inputs and output to facilitate interpretation.

Applying our algorithm to the aforementioned data, we recover a single change point in the first quarter of 2022. There are a number of possible explanations for this change point location. For instance, on 22 April 2022, the Brazilian Ministry of Health declared an end to the national public health emergency initially declared in response to the COVID-19 pandemic. Therefore, it is reasonable to believe that after the first quarter of 2022 consumption activities would begin to rebound. Additionally, there is evidence that the COVID-19 pandemic changed the way consumers engage with banks, predominately by shifting preferences from in-person to online banking \citep{baicu2020impact}. The estimated efficiency scores obtained by dividing each observed output by the maximum output given by the FDH estimate of the most efficient production frontier function via Equation~\eqref{equation: estimated scores}, as well as the estimated change point location, are shown in Figure~\ref{figure: Brazilian bank change point location}. Additionally, the production frontier functions estimated using the FDH using data before and after the estimated change point location are shown in Figure~\ref{figure: Brazilian bank FDH}, demonstrating how the production frontier shifts upward before and after the change. 

\begin{figure}[!ht]
    \centering
    \caption{Coloured points (\textcolor{red}{x}\textcolor{blue}{x}x\textcolor{green}{x}) represent estimated efficiency scores, via Equation~\eqref{equation: estimated scores}, for branches of a large Brazilian bank. Black dashed line (\textbf{- - -}) represents the change point location recovered by our procedure. Points are coloured according to the branch IDs.}
    \includegraphics[width=0.55\linewidth]{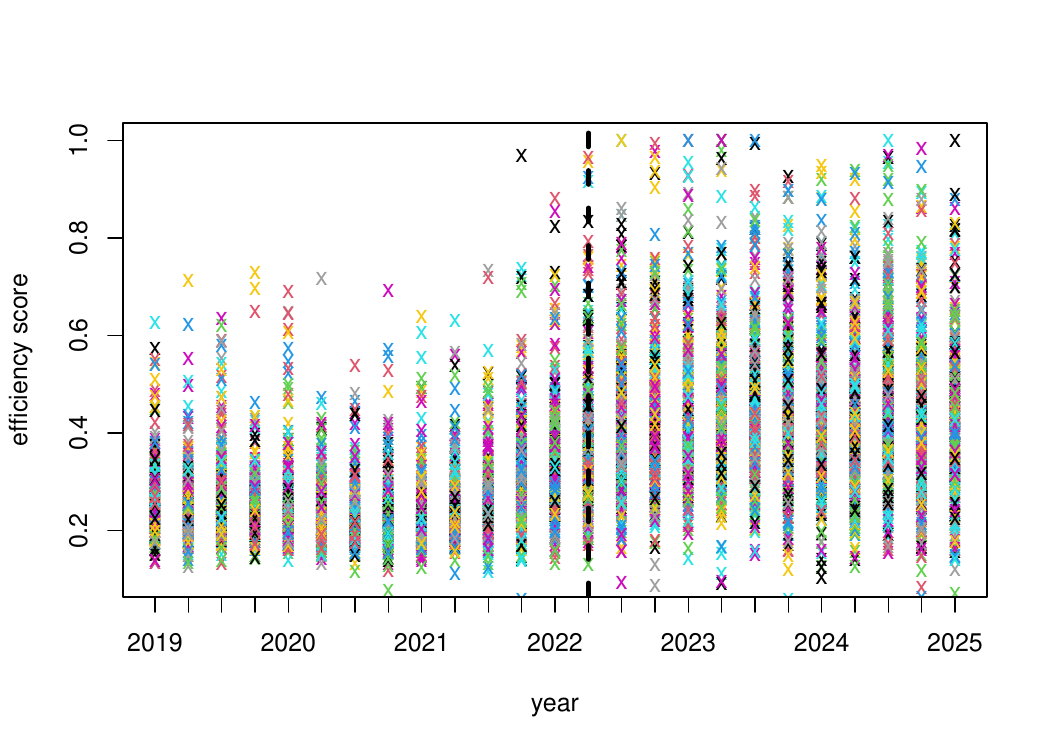}
    \label{figure: Brazilian bank change point location}
\end{figure}
\begin{figure}[!htbp]
    \centering
    \caption{Coloured lines (\textcolor{red}{\textbf{---}}, \textbf{\textcolor{blue}{---}}) represent the estimated production frontier functions before and after the change point location shown in Figure~\ref{figure: Brazilian bank change point location} using the FDH estimator.}
    \includegraphics[width=0.55\linewidth]{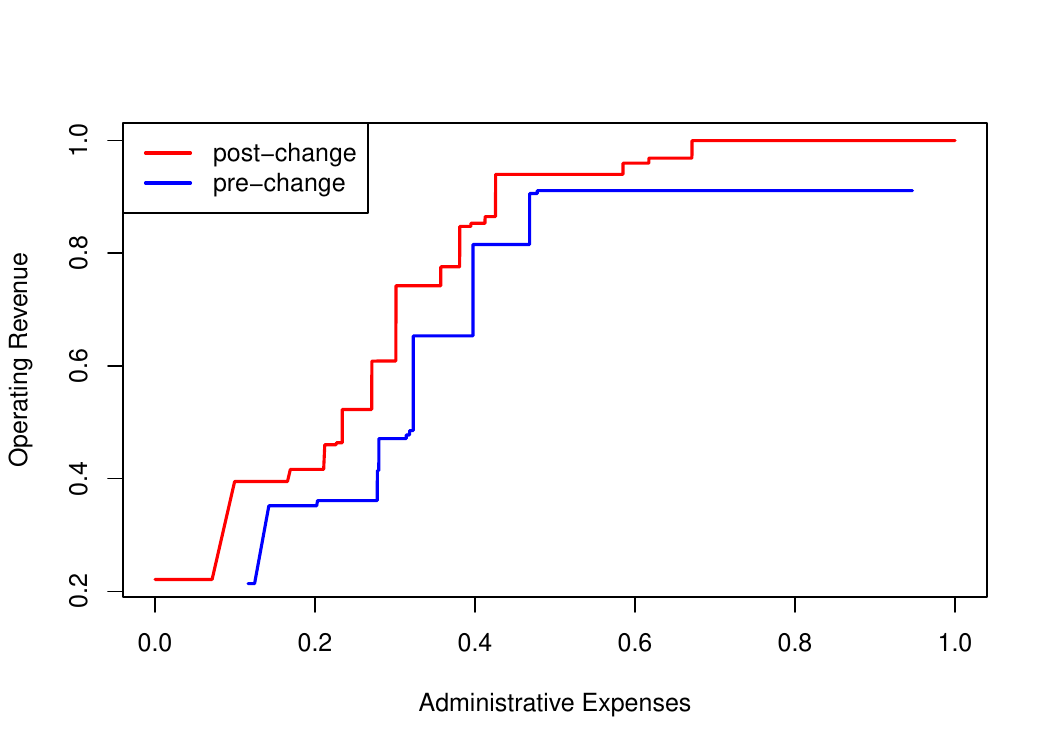}
    \label{figure: Brazilian bank FDH}
\end{figure}

\subsection{Analysis of South American economies} \label{Section: South America analysis}

We analyse national account data for six South American economies, namely, Argentina, Brazil, Colombia, Chile, Uruguay and Ecuador, obtained from the Penn World Table dataset originated from \cite{summers1991penn}. The same data were previously analysed by \cite{kumar2002technological} via DEA, and by \cite{tsionas2004markov} using a parametric Markov stochastic frontier model. The data consists of yearly observations from $1950$ to $2019$ for a total of $n = 318$ data points. We take as inputs the human capital index, based on years of schooling and returns to education, as well as the average number of hours worked per person (i.e. $d=2$). We take as output the expenditure side gross domestic product  (GDP). Prior to analysis we rescale both the inputs and outputs by each country's population in the given period. 

Applying Algorithm~\ref{algorithm: varaint of algo 1} to the aforementioned data, we recover a single change point in the year 1971. The estimated efficiency scores as well as the estimated change point location are shown in Figure~\ref{figure: South America change point location}. The 1970's marked a decisive turning point in the economic trajectory of Latin America, as the region became increasingly dependent on external borrowing to sustain growth. This debt-led expansion was closely tied to the global Oil Crises, which reshaped international financial flows and created new vulnerabilities for developing economies. For governments across the region, many of which were pursuing Import Substitution Industrialization strategies, the availability of cheap credit appeared to offer a unique opportunity to accelerate industrialization and modernize infrastructure. Brazil, experiencing what was termed its “Economic Miracle,” relied heavily on external borrowing to finance energy imports and expand state-owned enterprises, while Argentina similarly turned to foreign loans to sustain industrial growth. The result was a dramatic surge in sovereign debt during the mid-to-late 1970s, facilitated by lenders who regarded government borrowing as a secure investment. In this context, debt accumulation was not merely a financial strategy but a central pillar of development policy. 

\begin{figure}[!ht]
    \centering
    \caption{Coloured points 
    represent estimated efficiency scores, via equation \eqref{equation: estimated scores}, for six South American economies in the Penn World Table dataset described in~\ref{Section: South America analysis}. Black dashed line (\textbf{- - -}) represents the change point location recovered by our procedure. Points are coloured according to the country code.}
    \includegraphics[width=0.7\linewidth]{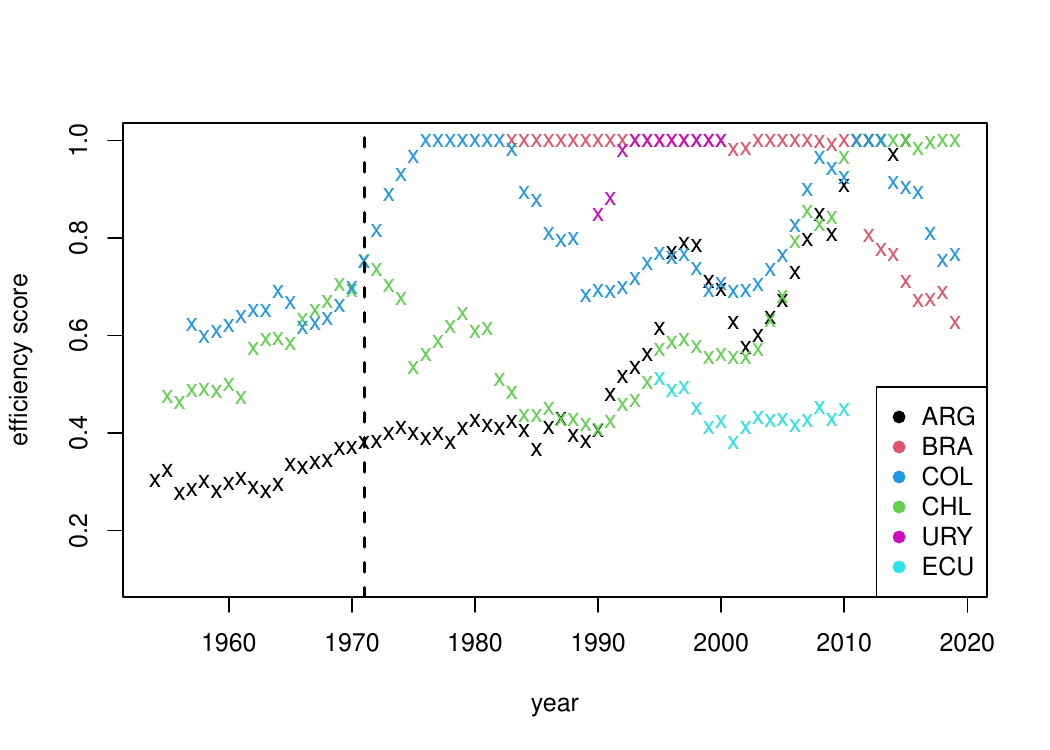}
    \label{figure: South America change point location}
\end{figure}

\section{Discussion}

In this work we studied the problem of estimating locations in time at which the level of technology in an economy changes, assuming the economy's underlying production function(s) expand sharply at discrete points in time. Two algorithms based on the FDH estimator for the production frontier function were presented, respectively for the settings in which the production frontier expands globally and locally. Localization rates for these algorithms were provided and shown to be unimprovable from a minimax perspective. Additionally, in the setting where the production function expands globally a simple method was proposed for constructing asymptotically conservative confidence intervals for the change point locations. 

There are a number of directions in which the present work could be extended. 

First, it is possible to extend our approach to handle serially dependent inputs and outputs, which is outlined in the \ref{section: serially dependent} of the supplementary materials. 

Second, in practice it is possible to observe outlying observations \citep{sexton1986methodology, simar2003detecting} for which the corresponding efficiency of the output is greater than $100\%$; this may occur for instance as a result of measurement error. In such settings the FDH estimator is no longer consistent, and our proposed methodology may therefore perform poorly. To remedy this issue one may use a quantile variant of the FDH estimator similar to the one proposed in \cite{aragon2005nonparametric}, which is robust to outlying observations. More precisely, for some $\alpha \in (0,1)$ one may replace the FDH estimator with the estimator 
\begin{equation}
    \hat{f}_{n,\alpha} ( \boldsymbol{x} ) = Q_{\alpha} \left ( Y_t \mid t \in \{ s : \boldsymbol{X}_s \leq \boldsymbol{x} \} \right ), \quad \boldsymbol{x} \in \mathcal{X}
    \label{equation: quantile FDH}
\end{equation}
where $Q_\alpha (Z_1, \dots, Z_m)$ represents the $\alpha$-th empirical quantile of the sample $Z_1, \dots, Z_m$. Provided the number of outlying observations grows at a slower rate than the stochastic order of the $\ell_\infty$ error of the FDH estimator, one can let $\alpha$ go to one with $n$ at an appropriate rate and obtain an analogous result to Theorem~\ref{theorem: exponential concentration in infty distance} for the quantity $\| \hat{f}_{n, \alpha} - f_{(1)} \|_\infty$. Consequently, using \eqref{equation: quantile FDH} in place of \eqref{equation: FDH production frontier function estimator} for estimating the efficiency scores would result in a method with the same theoretical guarantees as the methods studied in this paper. 

Finally, the FDH estimator suffers from the curse of dimensionality, in the sense that the rate of convergence slows as the dimensionality of the inputs increases. Therefore, the methodology proposed in the main paper is not appropriate for high dimensional data. One possible direction for alleviating this issue is to impose additional constraints on the production functions, for example, by assuming that most inputs do not impact the frontier, so that a variable selection procedure can be incorporated into our method a priori.

The extensions discussed in this section are left to be fully explored in future research. 

\addcontentsline{toc}{section}{References}
\bibliography{ref}
\bibliographystyle{alpha}

\newpage

\appendix

\makeatletter
\renewcommand{\theassumption}{\Alph{section}.\arabic{assumption}}
\setcounter{figure}{0}
\makeatother

\begin{center}

{\Large\bf SUPPLEMENTARY MATERIALS}

\end{center}

\section{Further details on numerical illustrations}

\subsection{Further details on real data analysis}

This section provides further details on the real data illustrations presented in Section~\ref{section: Brazilian bank analysis} and Section~\ref{Section: South America analysis}. Figures~\ref{figure: Brazilian bank inputs and ouputs distributions} and~\ref{figure: South America inputs and ouputs distributions} show scatter plots of the inputs and outputs over time for the real data examples in Sections~\ref{section: Brazilian bank analysis} and~\ref{Section: South America analysis} respectively. Regarding the analysis of South American economies presented in Section~\ref{Section: South America analysis}: in Figure \eqref{fig: gdp} it can be seen that the output is increasing monotonically over time, whereas Figures \eqref{fig: hc} and \eqref{fig: work} show that the inputs likewise exhibit monotonic behaviour. Note that this is not the case for the input in the real data example in Section~\ref{section: Brazilian bank analysis}, as can be seen in Figure~\ref{fig: expenses}.  We remark that when the support of the input distribution changes in a monotonic manner, the issue of identifiability may arise in the frontier models. One possible approach for resolving this is to explicitly model the panel structure of the data; currently the data are treated as a single time series, potentially with multiple observations at the same time point. Detecting changes in the production frontier function for panel data is an interesting direction for future research. 

\begin{figure}[!ht]
     \centering
     \caption{Administrative expenses \eqref{fig: expenses} and operating revenues \eqref{fig: revenue} for branches of a large Brazilian bank described in Section~\ref{section: Brazilian bank analysis}. Each quantity is normalized by the number of active customers of the given branch in the given year, after which min-max scaling is applied. Points are coloured according to the corresponding branch ID.}
     \begin{subfigure}[b]{0.3\textwidth}
         \centering
         \includegraphics[width=\textwidth]{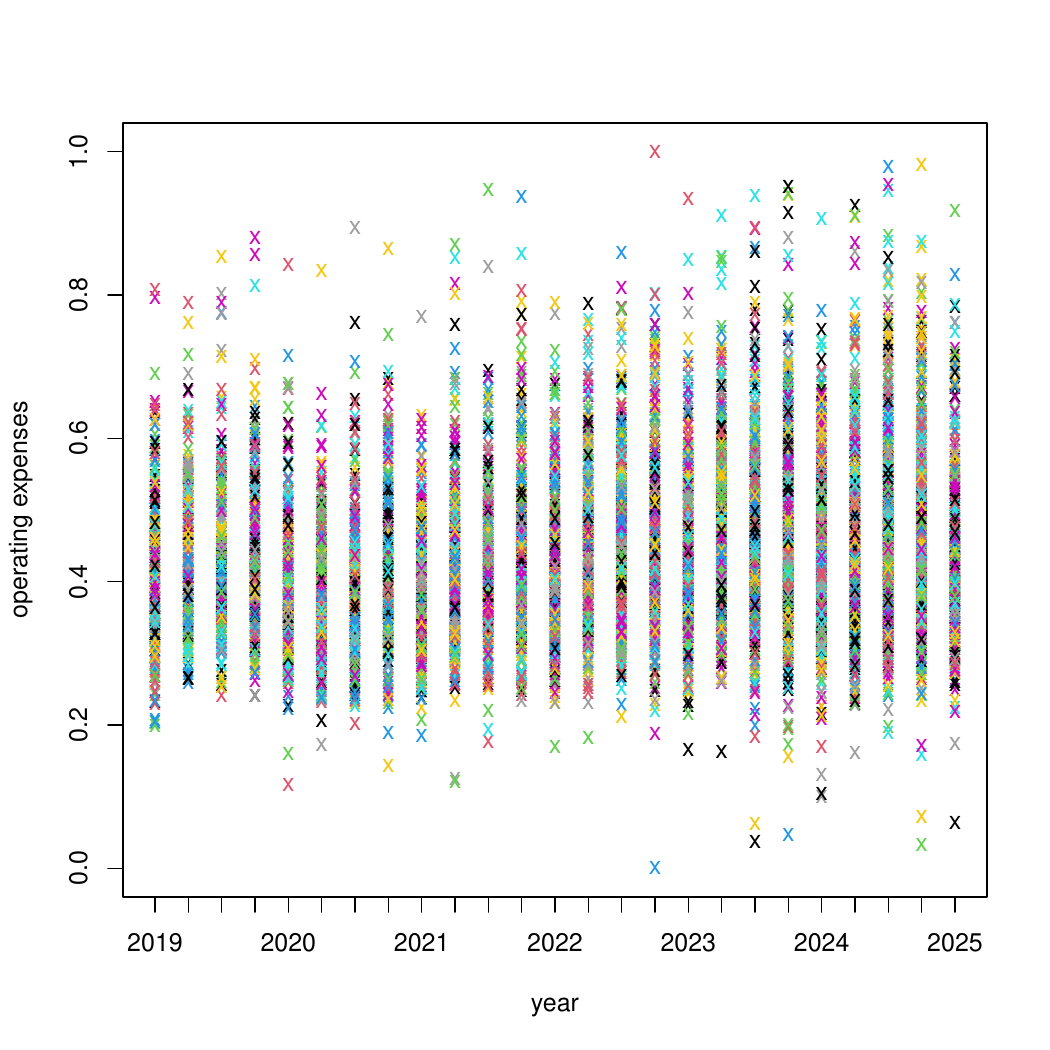}
         \caption{}
         \label{fig: expenses}
     \end{subfigure}
     \begin{subfigure}[b]{0.3\textwidth}
         \centering
         \includegraphics[width=\textwidth]{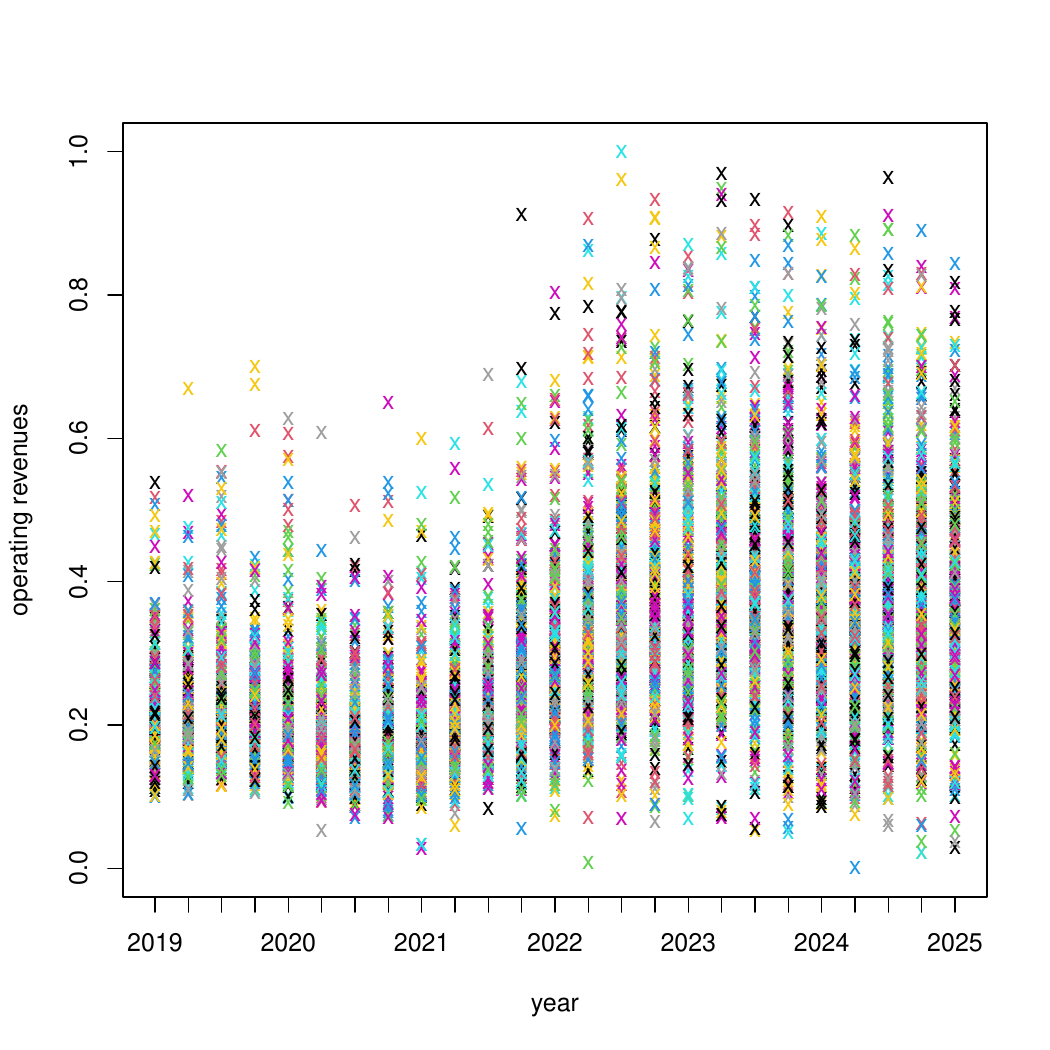}
         \caption{}
         \label{fig: revenue}
     \end{subfigure}
    \label{figure: Brazilian bank inputs and ouputs distributions}
\end{figure}

\begin{figure}[!ht]
     \centering
     \caption{Human capital \eqref{fig: hc}, hours worked \eqref{fig: work}, and GDP \eqref{fig: gdp} over time for the six South American economies in the Penn World Table dataset. Each quantity is normalized by the population of the given country in the given year. Points are coloured according to the country code.}
     \begin{subfigure}[b]{0.3\textwidth}
         \centering
         \includegraphics[width=\textwidth]{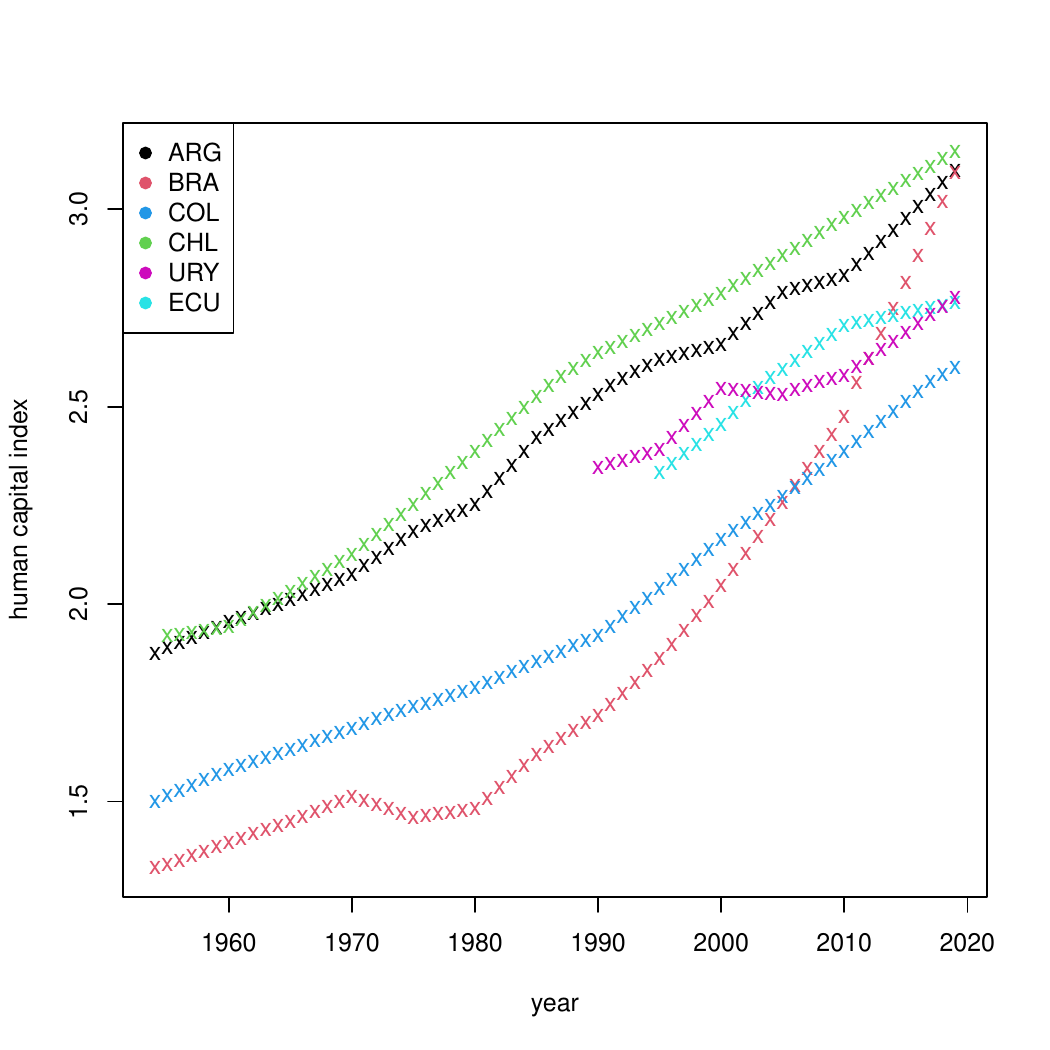}
         \caption{}
         \label{fig: hc}
     \end{subfigure}
     \hfill
     \begin{subfigure}[b]{0.3\textwidth}
         \centering
         \includegraphics[width=\textwidth]{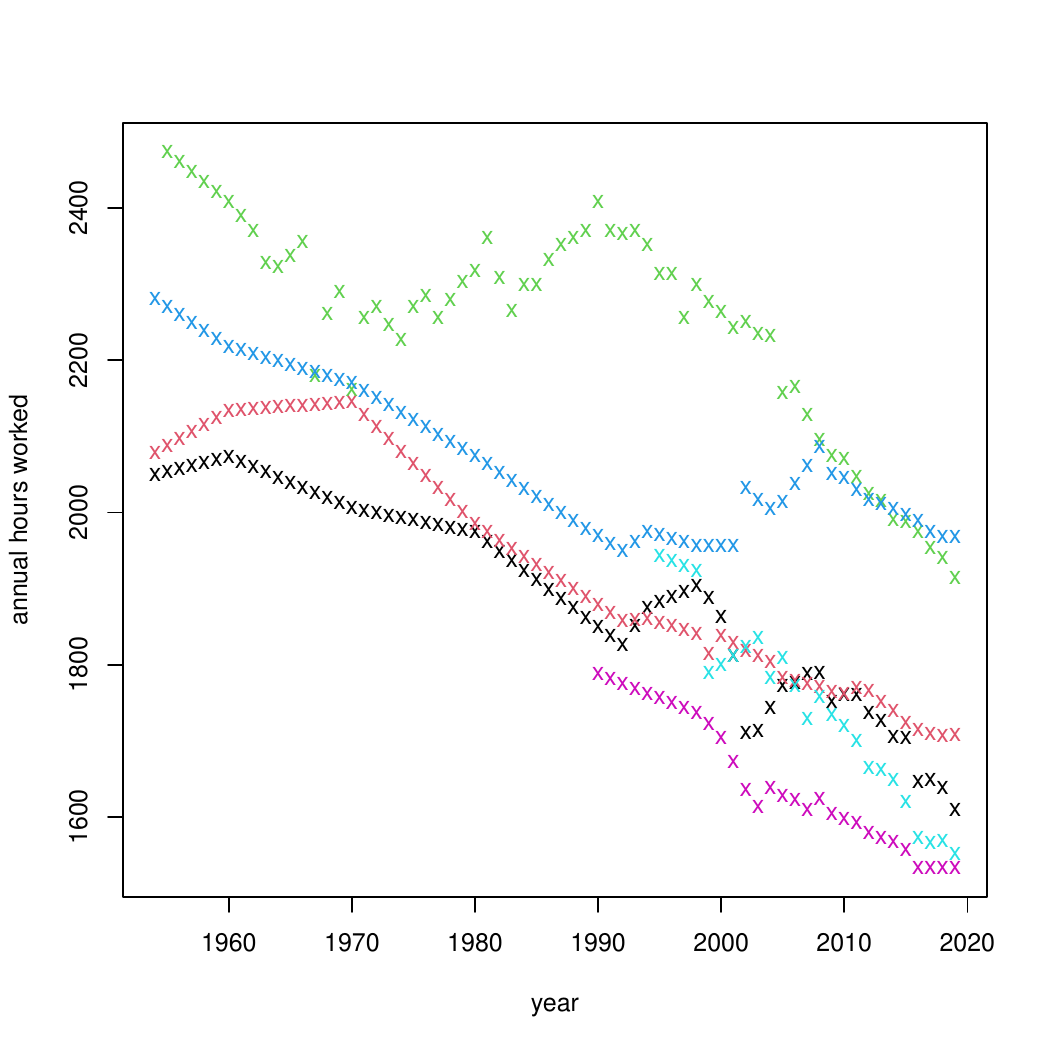}
         \caption{}
         \label{fig: work}
     \end{subfigure}
     \hfill
     \begin{subfigure}[b]{0.3\textwidth}
         \centering
         \includegraphics[width=\textwidth]{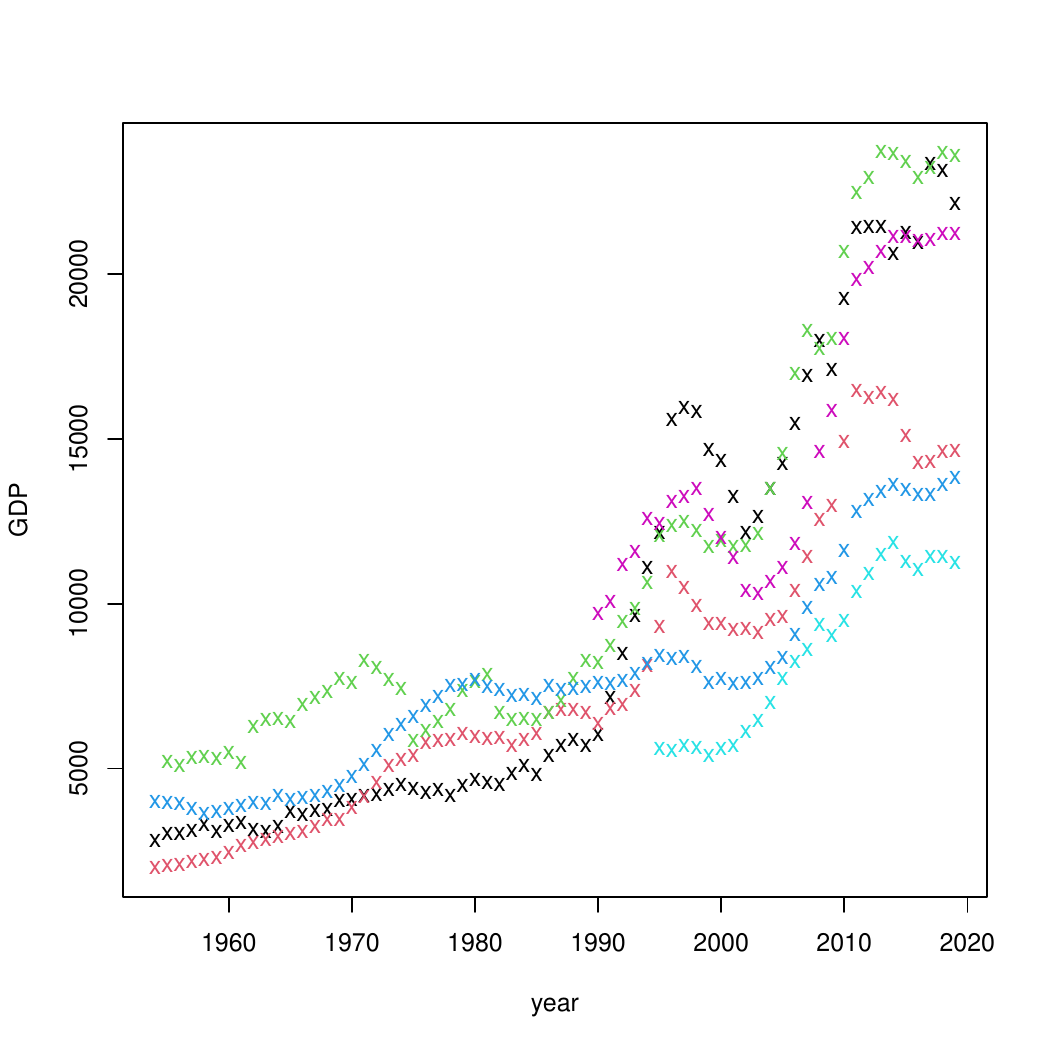}
         \caption{}
         \label{fig: gdp}
     \end{subfigure}
    \label{figure: South America inputs and ouputs distributions}
\end{figure}

\subsection{Algorithm used in the real data examples} \label{section: real data algorithm}

In this section we present a slight variant of Algorithm~\ref{algorithm: multiple change detection} designed to be robust to temporary but universal drops in the efficiencies, which can occur in practice. The main idea is as follows: to detect the first change we construct the pseudo-efficiency scores as in \eqref{equation: pseudo-efficiency scores}; however, instead of computing $\hat{L}$, defined in \eqref{equation: localized quasi-LR stat}, over left-expanding intervals, we compute the entire sequence $\{ \hat{L}_{1, \tau, n} \}_{\tau = 1, \dots, n}$; if all the $\hat{L}$'s are smaller than the threshold $\lambda$ we conclude there are no change points present, otherwise we estimate the right-most change point as the largest $\tau$ for which $\hat{L}_{1,\tau,n} > \lambda$ (call this value $\hat{\eta}$) and re-start the algorithm at $\hat{\eta} - \left \lfloor C \log (n) \right \rfloor$ for some $C > 0$. Pseudocode is provided in Algorithm~\ref{algorithm: varaint of algo 1}. As discussed earlier, this modification comes with the same theoretical guarantees. It does not lead to any noticeable differences in our simulation experiments, but does help handling the potential temporary drop of efficiencies in the real data analysis.

\begin{algorithm}[!htbp]
\DontPrintSemicolon
\SetKwInOut{Input}{Input}
\SetKwInOut{Output}{Output}
\Input{the observed data sequence $\left ( \boldsymbol{X}_1, Y_1 \right ), \dots, \left ( \boldsymbol{X}_n, Y_n \right )$, a threshold $\lambda$ against which to compare local quasi-Likelihood ratio statistics, a minimum value for the norm of the inputs $\boldsymbol{x}_0$, and a constant $C > 0$. }
\Output{a set of estimated change points $\hat{\Theta} \subset \left \{ 1, \dots, n \right \}$.}
\BlankLine
\Begin{
$\hat{\Theta} \leftarrow \emptyset$ \\
$\hat{K} \leftarrow 0$ \\
$m \leftarrow n$\\
$\text{test} \leftarrow \texttt{True}$ \\
\While{{\upshape test}}
{
$\text{detection} \leftarrow \texttt{False}$ \\
Estimate $\hat{f}_n \left ( \cdot \right )$ from $\left \{ \left ( \boldsymbol{X}_t, Y_t \right ) \mid t = 1, \dots, m \right \}$ via \eqref{equation: FDH production frontier function estimator} \\
Obtain $\hat{R}_1, \dots, \hat{R}_m$ as in \eqref{equation: estimated scores} \\
$L^* \leftarrow \max \left \{ \hat{L}_{1, \tau , m} \mid \tau = 1, \dots, m \right \}$ \\
\If{$L^* > \lambda$}{
$\hat{\eta} \leftarrow  \max \tau \quad \text{s.t.} \quad \hat{L}_{1,\tau,m} > \lambda$ \\
$m \leftarrow \max \left ( 0, \hat{\eta} - \left \lfloor C \log (n) \right \rfloor \right )$ \\
$\hat{\Theta} \leftarrow \left \{ \hat{\eta} \right \} \cup \hat{\Theta}$ \\
$\hat{K} \leftarrow \hat{K} + 1 $ \\
$\text{detection} \leftarrow \texttt{True}$ \\
BREAK \\
}
\If{{\upshape $m < 1$ or $\neg \text{detection}$}}{test $\leftarrow$ \texttt{False}}
}
}
\caption{A variant of Algorithm~\ref{algorithm: multiple change detection} designed to be robust to localized periods in which the level of technology temporarily decreases.}
\label{algorithm: varaint of algo 1}
\end{algorithm}

\section{Extension: serially dependent inputs and outputs} \label{section: serially dependent}

In this Section we extend the results of Section~\ref{section: multiple change points} to the setting in which the sequence of inputs and efficiency scores is permitted to be serially dependent. We work under the assumption of strong mixing of \cite{bradley2005basic}, although in practice the results in the sequel can be obtained by imposing any appropriate notion of weak dependence. Let $\left \{ Z_t \mid t \in \mathbb{Z} \right \}$ be a sequence of random variables equipped with a probability triple $\left ( \Omega, \mathcal{F}, \mathbb{P} \right )$, where $\mathcal{F} = \sigma \left ( Z_t \mid t \in \mathbb{Z} \right )$ denotes the sigma field generated by the $Z$'s. We briefly recall that the mixing coefficient of the random sequence is given by 
\begin{equation*}
\alpha (n) = \sup_{j \in \mathbb{Z}} \sup_{\substack{A \in \sigma \left ( Z_t \mid - \infty \leq t \leq j \right ) \\ B \in \sigma \left ( Z_t \mid j + n \leq t \leq +\infty \right )}} \left | \mathbb{P} \left ( A \cap B \right ) - \mathbb{P} \left ( A \right )\mathbb{P} \left ( B \right ) \right |. 
\end{equation*}

Showing consistency of Algorithm~\ref{algorithm: multiple change detection} under strong mixing boils down to showing that the FDH estimator continues to exhibit exponential concentration as in Theorem~\ref{theorem: exponential concentration in infty distance}. With such a result in place, the behaviour of the local tests \eqref{equation: localized quasi-LR stat} can be controlled on a high probability set using standard techniques. In order to prove an analogous result to Theorem~\ref{theorem: exponential concentration in infty distance} we need to assume the dependence between the time indexed inputs and efficiency scores, as measured by the mixing coefficient, decays sufficiently fast. We impose the following restriction:

\begin{assumption}
The sequence of inputs and efficiency scores $\left ( \boldsymbol{X}_t, R_t \right )_{t=1, \dots,n}$ constitutes a geometrically strong mixing sequence with mixing coefficient bounded from above as $\alpha \left ( n \right ) \leq A^n$ for some $A \in (0,1)$. 
\label{assumption: strong mixing}
\end{assumption}

With Assumption~\ref{assumption: strong mixing} in place we have the following result. 

\begin{corollary}
Grant Assumptions~\ref{assumption: Z is subset of hypercube},~\ref{assumption: production set is free disposal},~\ref{assumption: production frontier function continuous}, and~\ref{assumption: strong mixing} hold and let $\psi_n$ be as in Theorem~\ref{theorem: exponential concentration in infty distance}. For any $\zeta > 0$ we have that $\mathbb{P} \left ( {\psi}_n^{-1} \left \| f - \hat{f}_n \right \|_\infty > \zeta \right ) \leq C e^{-\zeta}$, where $C$ is a constant which depends only on the parameters in Assumptions~\ref{assumption: Z is subset of hypercube},~\ref{assumption: production set is free disposal},~\ref{assumption: production frontier function continuous}, and~\ref{assumption: strong mixing}.
\label{corollary: exponential concentration of FDH estimator strong mixing}
\end{corollary}

With the above result, we can derive an analogue to Theorem~\ref{theorem: multiple change detection} using the proof techniques in Section~\ref{section: proofs}. The details are omitted for the sake of brevity.

\section{Proofs} \label{section: proofs}

\subsection{Intermediate lemmas and their proofs}

\begin{theorem}
For any $n \in \mathbb{N}$, any $\left ( p_1, \dots, p_n \right ) ' \in [0,1]^n$, and any convex function $g : \mathbb{R} \mapsto \mathbb{R}$, if $Z_1 \sim \text{Poisson-Binomial} \left ( p_1, \dots, p_n \right )$ and $Z_2 \sim \text{Binomial} \left ( n, \overline{p} \right ) $ where in particular $\overline{p} = \left ( p_1 + \dots + p_n \right ) / n $ it holds that $\mathbb{E} \left [ g \left ( Z_1 \right ) \right ] \leq \mathbb{E} \left [ g \left ( Z_2 \right ) \right ]$. 
\label{theorem: poisson-binomial convex expectation}
\end{theorem}

\begin{proof}
See Theorem 3 in \cite{hoeffding1956distribution}. 
\end{proof}

\begin{lemma}
Let $\left ( R_1, \boldsymbol{X}_1 \right ) , \dots, \left ( R_n, \boldsymbol{X}_n \right ) $ be a sequence of random variables satisfying Assumption~\ref{assumption: X and R distribution}. For any $1 \leq t_1 < t_2 \leq n$ let $N \left ( t_1, t_2, \boldsymbol{x}_0 \right )$ be as defined in Section~\ref{section: detecting a single change} and put $M \left ( t_1, t_2, \boldsymbol{x}_0 \right ) = \max \left ( R_t \mid t_1 \leq t \leq t_2, \boldsymbol{X}_t > \boldsymbol{x}_0 \right )$. Therefore define the random variable 
\begin{equation*}
\boldsymbol{M} \left ( \boldsymbol{x}_0 \right ) = \max_{ 1 \leq t_1 < t_2 \leq n }  -2 N \left ( t_1, t_2, \boldsymbol{x}_0 \right ) \log \left [ M \left ( t_1, t_2, \boldsymbol{x}_0 \right ) \right ]. 
\end{equation*}
for some $\boldsymbol{x}_0 \in \mathcal{X}$. For any $\lambda \geq C_{3,R} \log (n)$ where $C_{3,R} = 12 \log (2 C_{2,R}) / \left ( 1-e^{-1} C_{1,R} \right )$ it holds that $\boldsymbol{M} \left ( \boldsymbol{x}_0 \right ) \leq \lambda$ on a set with probability at least $1 - \frac{4}{n}$. 
\label{lemma: max bound}
\end{lemma}

\begin{proof}
For each $t = 1, \dots, n$ write $p_{t,\boldsymbol{x}_0} = \mathbb{P} \left ( \boldsymbol{X}_t > \boldsymbol{x}_0 \right )$ and for any $1 \leq t_1 \leq t_2 \leq n$ write $\overline{p}_{t_1, t_2, \boldsymbol{x}_0} = (p_{t_1, \boldsymbol{x}_0} + \dots + p_{t_2, \boldsymbol{x}_0} / (t_2 - t_1 + 1)$. For integers $t_1 < t_2$ and $k \in \mathbb{N}_{\leq t_2-t_1+1}$ write $F_{k,t_1,t_2,}$ for the set of all $k$-subsets of $\{ t_1, \dots, t_2 \}$, and with a slight abuse of notation for each $A \in F_{k,t_1,t_2,}$ write $A^c$ for $\left \{ t_1, \dots, t_2 \right \} \setminus A $. Putting $\nu = \lambda / 4 \log (2 C_{2,R})$ we have that
\begin{subequations}
\begin{align}
& \mathbb{P} \left ( \boldsymbol{M} \left ( \boldsymbol{x}_0 \right ) > \lambda \right ) \label{equation: M union bound} \\
& \leq \sum_{\substack{1 \leq t_1 < t_2 \leq n}} \mathbb{P} \left ( M \left ( t_1, t_2, \boldsymbol{x}_0 \right ) \leq e^{- \lambda / 2 N \left ( t_1, t_2, \boldsymbol{x}_0 \right )} \right ) \nonumber \\ 
& = \sum_{\substack{1 \leq t_1 \leq t_2 \leq n \\ t_2 - t_1 + 1 \leq \nu}} \sum_{k=0}^{t_2 - t_1 + 1 } \sum_{A \in F_{k,t_1, t_2}} \prod_{t \in A} p_{t,\boldsymbol{x}_0} \mathbb{P} \left ( R_t \leq e^{-\lambda / 2k} \mid \boldsymbol{X}_t > \boldsymbol{x}_0 \right ) \prod_{s \in A^c} \left ( 1 - p_{s,\boldsymbol{x}_0} \right ) \label{equation: M T1} \\ 
& + \sum_{\substack{1 \leq t_1 \leq t_2 \leq n \\ t_2 - t_1 + 1 > \nu}} \sum_{k=0}^{t_2 - t_1 + 1 } \sum_{A \in F_{k,t_1, t_2}} \prod_{t \in A} p_{t,\boldsymbol{x}_0} \mathbb{P} \left ( R_t \leq e^{-\lambda / 2k} \mid \boldsymbol{X}_t > \boldsymbol{x}_0 \right ) \prod_{s \in A^c} \left ( 1 - p_{s,\boldsymbol{x}_0} \right ) \label{equation: M T2}
\end{align}
\end{subequations}
For term \eqref{equation: M T1} we have that
\begin{subequations}
\begin{align}
\eqref{equation: M T1} & \leq \sum_{\substack{1 \leq t_1 \leq t_2 \leq n \\ t_2 - t_1 + 1 \leq \nu}} \sum_{k=0}^{t_2 - t_1 + 1 } \sum_{A \in F_{k,t_1, t_2}} \prod_{t \in A} p_{t,\boldsymbol{x}_0} \prod_{s \in A^c} \left ( 1 - p_{s,\boldsymbol{x}_0} \right ) \left ( C_{2,R} e^{-\lambda / 2k} \right ) ^ k \nonumber \\ 
& = e^{-\lambda/2} \sum_{\substack{1 \leq t_1 \leq t_2 \leq n \\ t_2 - t_1 + 1 \leq \nu}} \sum_{k=0}^{t_2 - t_1 + 1 } \sum_{A \in F_{k,t_1, t_2}} \prod_{t \in A} p_{t,\boldsymbol{x}_0} \prod_{s \in A^c} \left ( 1 - p_{s,\boldsymbol{x}_0} \right ) \left ( C_{2,R} \right ) ^ k \nonumber \\ 
& \leq e^{-\lambda/2} \sum_{\substack{1 \leq t_1 \leq t_2 \leq n \\ t_2 - t_1 + 1 \leq \nu}} \sum_{k=0}^{t_2 - t_1 + 1 } \binom{t_2 - t_1 + 1}{k} \left ( \overline{p}_{t_1, t_2, \boldsymbol{x}_0} C_{2,R} \right )^k \left ( 1 - \overline{p}_{t_1, t_2, \boldsymbol{x}_0} \right )^{(t_2 - t_1 + 1) - k} \label{equation: M T1 - Hoef} \\ 
& = e^{-\lambda/2} \sum_{\substack{1 \leq t_1 \leq t_2 \leq n \\ t_2 - t_1 + 1 \leq \nu}} \left ( 1 - \overline{p}_{t_1, t_2, \boldsymbol{x}_0} + \overline{p}_{t_1, t_2, \boldsymbol{x}_0} C_{2,R} \right )^{(t_2 - t_1 + 1)} \label{equation: M T1 - binom} \\ 
& \leq n e^{-\lambda/2} \sum_{j=1}^{\nu} \left ( C_{2,R} \right )^j < n e^{-\lambda/2} \sum_{j=1}^{\nu} \left ( 2 \times C_{2,R} \right )^j \label{equation: M T1 - p bound} \\ 
& \leq \left ( \frac{2 C_{2,R}}{2 C_{2,R} - 1} \right ) 2 C_{2,R}^\nu \times n e^{-\lambda/2} \\ 
& \leq 2 n e^{-\lambda / 4}.
\label{equation: M T1 - final}
\end{align}
\end{subequations}
where we have used Theorem~\ref{theorem: poisson-binomial convex expectation} in \eqref{equation: M T1 - Hoef}, the Binomial Theorem in \eqref{equation: M T1 - binom}, and the fact that necessarily $C_{2,R} \geq 1$ since the $R$'s have support $[0,1]$ in \eqref{equation: M T1 - p bound}. Turning to term \eqref{equation: M T2} we have that
\begin{subequations}
\begin{align}
\eqref{equation: M T2} & = \sum_{\substack{1 \leq t_1 \leq t_2 \leq n \\ t_2 - t_1 + 1 > \nu}} \sum_{k=0}^{\nu} \sum_{A \in F_{k,t_1, t_2}} \prod_{t \in A} p_{t,\boldsymbol{x}_0} \mathbb{P} \left ( R_t \leq e^{-\lambda / 2k} \mid \boldsymbol{X}_t > \boldsymbol{x}_0 \right ) \prod_{s \in A^c} \left ( 1 - p_{s,\boldsymbol{x}_0} \right ) \label{equation: M T2-1} \\ 
& + \sum_{\substack{1 \leq t_1 \leq t_2 \leq n \\ t_2 - t_1 + 1 > \nu}} \sum_{k=\nu + 1}^{t_2 - t_1 + 1} \sum_{A \in F_{k,t_1, t_2}} \prod_{t \in A} p_{t,\boldsymbol{x}_0} \mathbb{P} \left ( R_t \leq e^{-\lambda / 2k} \mid \boldsymbol{X}_t > \boldsymbol{x}_0 \right ) \prod_{s \in A^c} \left ( 1 - p_{s,\boldsymbol{x}_0} \right ) \label{equation: M T2-2}
\end{align} 
\end{subequations}
and these terms can in turn be bounded as
\begin{align}
\eqref{equation: M T2-1} & \leq \sum_{\substack{1 \leq t_1 \leq t_2 \leq n \\ t_2 - t_1 + 1 > \nu}} \sum_{k=0}^{\nu} \sum_{A \in F_{k,t_1, t_2}}  \prod_{t \in A} p_{t,\boldsymbol{x}_0} \prod_{s \in A^c} \left ( 1 - p_{s,\boldsymbol{x}_0} \right ) \left ( C_{2,R} e^{-\lambda / 2k } \right )^k \nonumber \\ 
& \leq e^{-\lambda / 2} \left ( C_{2,R} \right ) ^ \nu \sum_{\substack{1 \leq t_1 \leq t_2 \leq n \\ t_2 - t_1 + 1 > \nu}} \sum_{k=0}^{\nu} \sum_{A \in F_{k,t_1, t_2}} \prod_{t \in A} p_{t,\boldsymbol{x}_0} \prod_{s \in A^c} \left ( 1 - p_{s,\boldsymbol{x}_0} \right ) \nonumber \\
& \leq n^2 e^{-\lambda / 4}, \label{equation: M T2-1 final bound}
\end{align}
and 
\begin{subequations}
\begin{align}
\eqref{equation: M T2-2} & \leq \sum_{\substack{1 \leq t_1 \leq t_2 \leq n \\ t_2 - t_1 + 1 > \nu}} \sum_{k=\nu + 1}^{t_2 - t_1 + 1} \sum_{A \in F_{k,t_1, t_2}} \prod_{t \in A} p_{t,\boldsymbol{x}_0} \prod_{s \in A^c} \left ( 1 - p_{s,\boldsymbol{x}_0} \right ) \left ( 1 - C_{1,R} \left ( 1 - e^{-\lambda / 2k} \right ) \right )^k \nonumber \\ 
& \leq \sum_{\substack{1 \leq t_1 \leq t_2 \leq n \\ t_2 - t_1 + 1 > \nu}} \sum_{k=\nu + 1}^{t_2 - t_1 + 1} \sum_{A \in F_{k,t_1, t_2}} \prod_{t \in A} p_{t,\boldsymbol{x}_0} \prod_{s \in A^c} \left ( 1 - p_{s,\boldsymbol{x}_0} \right ) \left ( 1 - \frac{\lambda}{k} \frac{C_{1,R}}{4 \log (2 C_{2,R})} \left ( 1 - e^{-1} \right ) \right )^k \label{equation: M T2-2 exp 1} \\ 
& \leq e^{-\lambda C_{1,R} (1 - e^{-1}) / (4 \log (2 C_{2,R}))} \sum_{\substack{1 \leq t_1 \leq t_2 \leq n \\ t_2 - t_1 + 1 > \nu}} \sum_{k=\nu + 1}^{t_2 - t_1 + 1} \sum_{A \in F_{k,t_1, t_2}} \prod_{t \in A} p_{t,\boldsymbol{x}_0} \prod_{s \in A^c} \left ( 1 - p_{s,\boldsymbol{x}_0} \right ) \label{equation: M T2-2 exp 2} \\ 
& \leq n^2 e^{-\lambda C_{1,R} (1 - e^{-1}) / (4 \log (2 C_{2,R}))} \label{equation: M T2-2 final}
\end{align}
\end{subequations}
where we have used the bound $1-e^{\frac{x}{y}} \geq \frac{x}{ay} (1-e^{-1})$ for $x,y \in \mathbb{N}$ and $a \geq 1$ where $a \times y \geq x$ in \eqref{equation: M T2-2 exp 1}, and the bound $(1+\frac{x}{n}) \leq e^{x}$ for $n \geq 1$ and $|x| \leq n$ in \eqref{equation: M T2-2 exp 2}. Finally plugging \eqref{equation: M T1 - final}, \eqref{equation: M T2-1 final bound}, and \eqref{equation: M T2-2 final} into \eqref{equation: M union bound} and applying the definition of $\lambda$ proves the desired result. 
\end{proof}

\begin{lemma}
Let $N \left ( t_1, t_2, \boldsymbol{x}_0 \right )$ be as defined in Section~\ref{section: detecting a single change} and let $\boldsymbol{x}_0$ be such that 
\begin{equation*}
\max_{t = 1, \dots, n} \mathbb{P} \left ( \boldsymbol{X}_t \leq \boldsymbol{x}_0 \right ) \leq \frac{1}{4}. 
\end{equation*}
It holds with probability at least $1 - C/ n$ that $N(t_1,t_2,\boldsymbol{x}_0) \geq \frac{1}{4} (t_2 - t_1 + 1)$ uniformly over all $1 \leq t_1 < t_2 \leq n$ with $t_2 - t_1 + 1 > 4 \log (n)$, where $C > 0$ is an absolute constant. 
\label{lemma: trimming}
\end{lemma}

\begin{proof}
A union bound argument followed by an application of the Binomial tail bound $\mathbb{P} \left ( Z \leq k \right ) \leq \exp \left ( -2m (p-k/m)^2 \right ) $ with $Z \sim \text{Binom} \left ( m, p \right )$ and $k \leq mp$ gives the following
\begin{align*}
& \mathbb{P} \left ( \cup_{\substack{1 \leq t_1, < t_2 \leq n  \\ t_2 - t_1 + 1 > 4 \log (n)}} \left \{ N \left ( t_1, t_2, \boldsymbol{x}_0 \right )< \frac{1}{4} (t_2 - t_1 + 1) \right \} \right ) \\
& \leq \sum_{\substack{1 \leq t_1 < t_2 \leq n  \\ t_2 - t_1 + 1 > 4 \log (n)}} \mathbb{P}_{Z \sim \text{Binom} \left ( (t_2 - t_1 + 1), 3/4 \right )} \left ( Z \leq \frac{1}{4} (t_2 - t_1 + 1) \right ) \\ 
& = \sum_{x = 4 \log (n) + 1}^n \left ( n - x + 1 \right ) e^{-x/2} \\
& \leq \frac{e^{\frac{1}{2}}}{1-e^{-\frac{1}{2}}} n e^{-2 \log (n)}
\end{align*}
from which the desired result follows. 
\end{proof}

\begin{lemma}
Let $\left ( R_1, \boldsymbol{X}_1 \right ) , \dots, \left ( R_n, \boldsymbol{X}_n \right ) $ be a sequence of random variables satisfying Assumption~\ref{assumption: X and R distribution} and let $\mathfrak{X}$ be as defined in \eqref{equation: multi-scale grid}. For any $1 \leq t_1 < t_2 \leq n$ and any $\mathcal{X}' \in \mathfrak{X}$ let $N \left ( t_1, t_2, \mathcal{X}' \right )$ be as defined in Section~\ref{section: local changes in technology} and put $M \left ( t_1, t_2, \mathcal{X}' \right ) = \max \left ( R_t \mid t_1 \leq t \leq t_2, \boldsymbol{X}_t \in \mathcal{X}' \right )$. Therefore define the random variable 
\begin{equation*}
\boldsymbol{M} \left ( \mathfrak{X} \right ) = \max_{\substack{\mathcal{X}' \in \mathfrak{X} \\ 1 \leq t_1 < t_2 \leq n}}  -2 N \left ( t_1, t_2, \mathcal{X}' \right ) \log \left [ M \left ( t_1, t_2, \mathcal{X}' \right ) \right ]. 
\end{equation*}
For any any $\lambda \geq C_{3,R} \log (n)$ where $C_{3,R} = 12 \log (2 C_{2,R}) / \left ( (1-e^{-1}) C_{1,R} \right )$ it holds that $\boldsymbol{M} \left ( \mathfrak{X} \right ) \leq \lambda$ on a set with probability at least $1 - C A_n n^{-1}$ where $C$ is an absolute constant depending only on $d$ and $\operatorname{mes} (\mathcal{X})$. 
\label{lemma: multiscale gird max bound}
\end{lemma}

\begin{proof}
For each $t = 1, \dots, n$ and each $\mathcal{X}' \in \mathfrak{X}$ write $p_{t, \mathcal{X}'} = \mathbb{P} \left ( \boldsymbol{X}_t \in \mathcal{X}' \right )$. For each $t_1 \leq t_2$ and $k \in \mathbb{N}_{t_2 - t_1 +1}$ let $F_{t_1, t_2, k}$ be as defined in the proof of Lemma~\ref{lemma: max bound}. We therefore have that
\begin{align*}
& \mathbb{P} \left ( \boldsymbol{M} \left ( \mathfrak{X} \right ) > \lambda \right ) \\ 
& \quad \quad \leq \sum_{\mathcal{X}' \in \mathfrak{X}} \sum_{1 \leq t_1 \leq t_2 \leq n} \sum_{k=0}^{t_2-t_1+1} \sum_{A \in F_{t_1,t_2,k}} \prod_{t \in A} p_{t, \mathcal{X}'} \mathbb{P} \left ( R_t \leq e^{-\lambda / 2k} \mid \boldsymbol{X}_t \in \mathcal{X}' \right ) \prod_{s \in A^c} \left ( 1 - p_{t, \mathcal{X}'} \right ). 
\end{align*}
For each $\mathcal{X}' \in \mathfrak{X}$ the quantity above following the first sum can be bounded as in the proof of Lemma~\ref{lemma: max bound}. Note also that that the cardinality of $\mathfrak{X}$, as defined in \eqref{equation: multi-scale grid}, is bounded as
\begin{equation}
\left | \mathfrak{X} \right | \leq  2^d \sum_{k=0}^{\left \lfloor \log_2 \left ( \underline{x} A_n^{1/d} \right )  \right \rfloor} \operatorname{mes} \left ( \mathcal{X} \right ) \underline{x}^{-d} 2^{dk} \leq 2^{d+2} \operatorname{mes} \left ( \mathcal{X} \right ) A_n. 
\label{equation: frac X cardinality}
\end{equation}
These facts together complete the proof. 
\end{proof}

\begin{lemma}
 With $N \left ( t_1, t_2, \mathcal{X}' \right )$ and $\mathfrak{X}$ defined as in Section~\ref{section: local changes in technology} it holds with probability at least $1 - n^{-1} C A_n^3$ that 
\begin{equation}
N \left ( t_1, t_2, \mathcal{X}' \right ) \geq \left ( t_2 - t_1 + 1 \right ) \frac{C_{1,X} \operatorname{mes} \left ( \mathcal{X}' \right )}{2}
\label{equation: N mathcal X event}
\end{equation}
uniformly over all $1 \leq t_1 < t_2 \leq n$ with $t_2 - t_1 + 1 > 4 C_{1,X}^{-2} \log (n) A_n^2$ and all $\mathcal{X}' \in \mathfrak{X}$. In particular $C > 0$ is an absolute constant. 
\label{lemma: trimming multiscale grid}
\end{lemma}

\begin{proof}
Arguing as in the proof of Lemma~\ref{lemma: trimming}, we obtain that
\begin{align*}
& \mathbb{P} \left ( \cup_{\mathcal{X}' \in \mathfrak{X}} \cup_{\substack{1 \leq t_1, < t_2 \leq n  \\ t_2 - t_1 + 1 > 4 C_{1,X}^{-2} \log (n) A_n^2 }} \left \{ N \left ( t_1, t_2, \boldsymbol{x}_0 \right ) < \left ( t_2 - t_1 + 1 \right ) \frac{C_{1,X} \operatorname{mes} \left ( \mathcal{X}' \right )}{2} \right \} \right ) \\
& \leq \sum_{\mathcal{X}' \in \mathfrak{X}} \sum_{\substack{1 \leq t_1 < t_2 \leq n  \\ t_2 - t_1 + 1 > 4 C_{1,X}^{-2} \log (n) A_n^2 }} \mathbb{P}_{Z \sim \text{Binom} \left ( (t_2 - t_1 + 1), C_{1,X} \operatorname{mes} \left ( \mathcal{X}' \right ) \right )} \left ( Z \leq \left ( t_2 - t_1 + 1 \right ) \frac{C_{1,X} \operatorname{mes} \left ( \mathcal{X}' \right )}{2} \right ) \\ 
& = n \sum_{\mathcal{X}' \in \mathfrak{X}} \sum_{x = 4 C_{1,X}^{-2} \log (n) A_n^2 + 1}^n e^{-\frac{x}{2} \left ( C_{1,X} \operatorname{mes} \left ( \mathcal{X}' \right ) \right )^2} \\
& \leq n \sum_{\mathcal{X}' \in \mathfrak{X}} \frac{e^{-2 \log (n)}}{e^{ \frac{1}{2} \left ( C_{1,X} \operatorname{mes} \left ( \mathcal{X}' \right ) \right )^2} - 1} \\
& \leq n \left | \mathfrak{X} \right | 2 C_{1,X}^{-2} A_n^2 e^{-2 \log (n)}. 
\end{align*}
Applying the bound \eqref{equation: frac X cardinality}, the desired result therefore follows. 
\end{proof}

\subsection{Proofs of the main results}

\subsubsection{Proof of Theorem~\ref{theorem: exponential concentration in infty distance}} \label{section: proof of exponential concentration in Haus. distance}

\begin{proof}

We adapt the proof techniques from \cite{korostelev1995efficient} and \cite{korostelev1995estimation}. The former paper studies the risk of the FDH estimator for $\Psi$, in terms of the Hausdorff distance, when the $Z$'s are uniformly distributed on $[0,1]^2$, whereas the latter paper studies the risk in terms of the Lebesgue measure of the symmetric difference of sets.

Let $B \left ( u, r \right ) $ be the sphere in $\mathbb{R}^{d+1}$ with centre $u$ and radius $r$, and put $W \left ( u, r \right ) =  B \left ( u, r \right ) \cap \Psi$. Let $\mathcal{N} = \left \{ u_1, \dots, u_M \right \}$ be a minimal $\psi_n$-net on $\partial \Psi$, and for each $u \in \mathcal{N}$ let $r_k := r_k \left ( u \right )$ be radii defined so that 
\begin{equation}
\operatorname{mes} \left ( W \left ( u, r_k \right ) \right ) = C_W \psi_n^{d+1} \left ( \frac{d}{d+1} + k \right ),  
\label{equation: mes W bound}
\end{equation}
where $\operatorname{mes} \left ( A \right )$ stands for the Lebesgue measure of the set $A$ and the leading constant is given by
\begin{equation}
C_W := \left ( \frac{C_f}{C_{1,X} C_{1,R}} \right ) ^ \frac{1}{d+1} \frac{\sqrt{d+1} \pi^{\frac{d+1}{2}}}{\Gamma \left ( \frac{d+1}{2} + 1 \right )}
\label{equation: C_W}
\end{equation}
and in particular $C_f := \sup_{\boldsymbol{x} \in \mathcal{X}} f (\boldsymbol{x})$. Introduce the events
\begin{equation}
A_{k,n} = \bigcap_{l=1}^M \left \{ \exists \left ( \boldsymbol{X}_t, Y_t \right ) \in W \left ( u_l, r_k \right ) \right \}, \hspace{2em} k = 1, \dots, n^{\frac{1}{d+1}}. 
\label{equation: psi-net events}
\end{equation}
We turn our attention to the probability of interest in Theorem~\ref{theorem: exponential concentration in infty distance}. We have that: 
\begin{align}
\mathbb{P} \left ( \psi_n^{-1} \left \| f - \hat{f}_n \right \|_\infty > \zeta \right ) & \leq e^{-\zeta} \mathbb{E} \left [ e^{\psi_n^{-1} \left \| f - \hat{f}_n \right \|_\infty} \right ] \nonumber \\ 
& = e^{-\zeta} \left [ 1 + \sum_{l=1}^\infty \frac{1}{l!} \mathbb{E} \left [ \left ( \psi_n^{-1} \left \| f - \hat{f}_n \right \|_\infty \right ) ^ l \right ] \right ] \nonumber \\ 
& \leq e^{-\zeta} \left [ 1 + \sum_{l=1}^\infty \frac{1}{l!} \left \{ \mathbb{E} \left [ \left ( \psi_n^{-1} \left \| f - \hat{f}_n \right \|_\infty \right ) ^ l \mid A_{1,n} \right ] \right . \right . \nonumber \\ 
& \hspace{4em} + \sum_{k=1}^{n^{\frac{1}{d+1}}-1} \mathbb{E} \left [ \left ( \psi_n^{-1} \left \| f - \hat{f}_n \right \|_\infty \right ) ^ l \mid A_{k+1,n} \right ] \mathbb{P} \left ( A_{k,n}^c \right ) \nonumber \\
& \hspace{4em} + \left . \left . \mathbb{E} \left [ \left ( \psi_n^{-1} \left \| f - \hat{f}_n \right \|_\infty \right ) ^ l \mid A_{n^{\frac{1}{d+1}},n}^c \right ] \mathbb{P} \left ( A_{n^{\frac{1}{d+1}},n}^c \right ) \right \} \right ]
\label{equation: Markov bound on P}
\end{align}
We now turn our attention to the expectations in \eqref{equation: Markov bound on P}. Note that for each $u \in \mathcal{N}$ there is a $u'$ such that $B \left ( u', r_k / 2 \right ) \subset W \left ( u, r_k \right )$. Let $C \left ( u', r_k \right )$ be the axis aligned cube with sides of length $l_k := l_k (u)$ inscribed into $B \left ( u', r_k / 2 \right )$ and let $C' \left ( u', r_k \right )$ and $C'' \left ( u', r_k \right )$ be respectively its coordinate-wise projection onto $\mathcal{X}$ and $\mathcal{Y}$. We therefore have that
\begin{align}
\mathbb{P} \left ( A_{k,n} \right ) & \geq 1 - \sum_{l=1}^M \prod_{t=1}^n \left ( 1 - \mathbb{P} \left ( \left ( \boldsymbol{X}_t, Y_t \right ) \in W \left ( u_l, r_k \right ) \right ) \right ) \nonumber \\ 
& \geq 1 - \sum_{l=1}^M \prod_{t=1}^n \left ( 1 - \mathbb{P} \left ( Y_t \in C'' \left ( u_l', r_k \right ) \mid \boldsymbol{X}_t \in C' \left ( u_l', r_k \right ) \right ) \mathbb{P} \left ( \boldsymbol{X}_t \in C' \left ( u_l', r_k \right ) \right ) \right ) \nonumber \\ 
& \geq 1 - \sum_{l=1}^M \prod_{t=1}^n \left ( 1 - \mathbb{P} \left ( R_t \geq 1 - l_k / C_f \mid \boldsymbol{X}_t \in C' \left ( u_l', r_k \right ) \right ) \int_{C' \left ( u_l', r_k \right )} f_{\boldsymbol{X}_t} \left ( x \right ) \mathrm{d} x \right ) \nonumber \\ 
& \geq 1 - \sum_{l=1}^M \left ( 1 - C_{1,R} \left ( l_k \left ( u_l \right ) / C_f \right ) C_{1,X} \operatorname{mes} \left ( C' \left ( u_l', r_k \left ( u_l \right ) \right ) \right )  \right ) ^ n
\label{equation: P A bound}
\end{align}
Note that $M$ is bounded from above by $C_{d} \psi_n^{-d}$ where $C_d$ depends on $d$ and the constants $\left \{ \overline{x}_1, \dots, \overline{x}_d, \overline{y} \right \}$ introduced in Assumption~\ref{assumption: Z is subset of hypercube}. Note also that (i) $l_k = r_k / \sqrt{d+1}$ and that (ii) since $\operatorname{mes} \left ( B \left ( u, r_k \right ) \right ) \geq \operatorname{mes} \left ( W \left ( u, r_k \right ) \right )$ for any $u \in \mathcal{N}$ each $r_k$ will satisfy
\begin{equation}
r_k \geq \psi_n \left ( \frac{d}{d+1} + k \right ) ^ \frac{1}{d+1} \left [ \frac{C_W \Gamma \left ( \frac{d+1}{2} + 1 \right )}{\pi^{\frac{d+1}{2}}} \right ] ^ \frac{1}{d+1} 
\label{eqaution: r_k bound}
\end{equation}
Therefore plugging \eqref{equation: C_W} and \eqref{eqaution: r_k bound} into \eqref{equation: P A bound} we obtain that
\begin{align*}
\mathbb{P} \left ( A_{k,n} \right ) & \geq 1 - C_d \psi_n^{-d} \left ( 1 - \left ( \frac{d}{d+1} + k \right ) \psi_n^{d+1} \right ) ^ n \\
& = 1 - C_d \left ( \frac{n}{\log (n)} \right ) ^ \frac{d}{d+1} \left ( 1 - \frac{1}{n} \left ( \frac{d}{d+1} + k \right ) \log (n) \right ) ^ n  \\
& \geq 1 - C_d n^{-k}
\end{align*}
For the radii $r_k$ defined through \eqref{equation: mes W bound} the fact that $\operatorname{mes} \left ( B (u , r ) \right ) \leq 2^{d+1} \operatorname{mes} \left ( W \left ( u, r \right ) \right )$ along with the monotonicity of $\Psi$ imply that 
\begin{equation}
r_k \leq \frac{2}{\sqrt{\pi}} \left ( C_W \Gamma \left ( \frac{d+1}{2} + 1 \right ) \right ) ^ \frac{1}{d+1} \psi_n \left ( \frac{d}{d+1} + k \right ) := C \psi_n \left ( \frac{d}{d+1} + k \right ) 
\label{equation: bound on radii}
\end{equation}
for each $k = 1, \dots, n^{\frac{1}{d+1}}$. For every such $k$ and any $\boldsymbol{x} \in \mathcal{X}$ with $\left \| \boldsymbol{x} \right \| \leq r_k \left ( 3 + \frac{1}{L} \right )$, Assumption~\ref{assumption: production frontier function continuous} along with the facts
\begin{enumerate}[(i)]
    \item $\left | f (\boldsymbol{0}) - \hat{f}_n (\boldsymbol{x}) \right | \leq \left | f (\boldsymbol{0}) - \hat{f}_n (\tilde{\boldsymbol{x}}) \right |$ for any $\tilde{\boldsymbol{x}} \leq \boldsymbol{x}$, and
    \item $f (\boldsymbol{0}) = \hat{f}_n (\boldsymbol{0}) = 0$
\end{enumerate}
together imply that
\begin{equation}
\left | f (\boldsymbol{x}) - \hat{f}_n (\boldsymbol{x}) \right | \leq r_k \left ( 3L + 1 \right ).
\label{equation: bound when x close to zero}
\end{equation}
Let $\mathcal{N}'$ be the coordinate-wise projection of $\mathcal{N}$ onto $\mathcal{X}$. For any $\boldsymbol{x} \in \mathcal{X}$ with $\left \| \boldsymbol{x} \right \| > r_k \left ( 3 + \frac{1}{L} \right )$, let 
\begin{equation}
\boldsymbol{x'} \in \left \{ \boldsymbol{v} \in \mathcal{N}' \mid \boldsymbol{v} < \boldsymbol{x} \text{ and } r_k \leq \left \| \boldsymbol{x} - \boldsymbol{v} \right \| \leq 2 r_k \right \}. 
\end{equation}
On the event $A_{k,n}$ there must be an $\boldsymbol{x}''$ with $\left \| \boldsymbol{x}' - \boldsymbol{x}'' \right \| \leq r_k$ for which $\hat{f}_n (\boldsymbol{x}) \geq f (\boldsymbol{x}'') - r_k$. Assumption~\ref{assumption: production frontier function continuous} along with the fact that $\hat{\Psi}_n \subseteq \Psi$ imply that
\begin{align}
\left | f (\boldsymbol{x}) - \hat{f}_n (\boldsymbol{x}) \right | \leq \left | f (\boldsymbol{x}) - f (\boldsymbol{x}'') \right | + r_k \leq L \left ( \left \| \boldsymbol{x} - \boldsymbol{x}' \right \| + \left \| \boldsymbol{x}' - \boldsymbol{x}'' \right \| \right ) + r_k \leq r_k \left ( 3L + 1 \right ). 
\label{equation: bound when x far from zero}
\end{align}
Therefore \eqref{equation: bound when x close to zero} and \eqref{equation: bound when x far from zero} together with \eqref{equation: bound on radii} imply that there are constants $C_1$ and $C_2$ such that for each $l > 0$
\begin{equation}
\mathbb{E} \left [ \left ( \psi_n^{-1} \left \| f - \hat{f}_n \right \|_\infty \right ) ^ l \mid A_{k,n} \right ] \leq \left ( C_1 k + C_2 \right ) ^ l  \hspace{2em} k = 1, \dots, n^{\frac{1}{d+1}}. 
\label{equation: conditional expectation bound}
\end{equation}
Plugging \eqref{equation: conditional expectation bound} into \eqref{equation: Markov bound on P}, we therefore have that
\begin{align}
\eqref{equation: Markov bound on P} & \leq e^{-\zeta} \left [ 1 + \sum_{l=1}^\infty \frac{1}{l!} \left ( C_1 + C_2 \right )^l + \sum_{l=1}^\infty \frac{1}{l!} \sum_{k=1}^{n^{\frac{1}{d+1}}-1} \left ( C_1 k + C_2 \right ) ^ l C_\mathcal{Z} \left ( d + 1 \right ) n^ {-k} \right . \nonumber \\
& \hspace{2em} \left . + \sum_{l=1}^\infty \frac{1}{l!} \left ( C_\mathcal{Z} n ^ \frac{1}{d+1} \sqrt{C_\mathcal{Z} \left ( d + 1 \right )} \right ) ^ l C_\mathcal{Z} \left ( d + 1 \right ) n ^ {-n^{\frac{1}{d+1}}} \right ] \nonumber \\ 
& := e^{-\zeta} \left [ 1 + S_1 + S_2 + S_3 \right ]
\label{equation: further Markov bound on P}
\end{align}
We now bound each of the sums in term. For the first sum, we have that $S_1 \leq e^{C_1 + C_2} := C_3$. For the second sum, we have the following provided that $n \geq e^{C_1 +1}$
\begin{align}
S_2 & \leq C_\mathcal{Z} \left ( d + 1 \right ) \sum_{k=1}^\infty n^{-k} \sum_{l=1}^\infty \frac{1}{l!} \left ( C_1 k + C_2 \right ) ^ l  \\ 
& \leq C_\mathcal{Z} \left ( d + 1 \right ) \sum_{k=1}^\infty n^{-k} e ^ {C_1 k + C_2} \\ 
& \leq C_\mathcal{Z} \left ( d + 1 \right ) e ^{C_2} \sum_{k=1}^\infty e ^ {-k} \\
& := C_4.
\label{equation: S2 bound}
\end{align}
For the third sum, we have the following 
\begin{align}
S_3 & \leq C_\mathcal{Z} \left ( d + 1 \right ) n^{-n^{\frac{1}{d+1}}} \sum_{l=0}^\infty \frac{1}{l!} \left ( C_\mathcal{Z} n ^ \frac{1}{d+1} \sqrt{C_\mathcal{Z} \left ( d + 1 \right )} \right ) ^ l \\ 
& \leq C_\mathcal{Z} \left ( d + 1 \right ) e ^ {-n^{\frac{1}{d+1}} \left (\log (n) - C_\mathcal{Z} \sqrt{C_\mathcal{Z} \left ( d + 1 \right )} \right )} 
\label{equation: S3 bound}
\end{align}
which goes to zero as $n \rightarrow \infty$ and is therefore bounded by some absolute constant $C_5$. Combining \eqref{equation: further Markov bound on P}, \eqref{equation: S2 bound} and \eqref{equation: S3 bound}, we have that $\eqref{equation: Markov bound on P} \leq e^{-\zeta} \left [ 1 + C_3 + C_4 + C_5 \right ]$, which completes the proof. 
\end{proof}

\subsubsection{Proof of Theorem~\ref{theorem: single change detection}}

\begin{proof}[Proof: part (i)]
Let $\lambda$ be as in Lemma~\ref{lemma: max bound}. With $N = 0$ we have the crude bound $R_t \leq \hat{R}_t \leq 1$ for each $t = 1, \dots, n$, which implies that
\begin{equation*}
\hat{L}_\tau \leq - 2 N \left ( 1, \tau, \boldsymbol{x}_0\right ) \log \left [ M \left ( 1 , \tau , \boldsymbol{x}_0\right ) \right ]
\end{equation*}
for all valid choices of $\tau$. Consequently, Lemma~\ref{lemma: max bound} guarantees that with high probability $\hat{L}_\tau \leq \lambda$ uniformly over all values of $\tau$. This proves part (i) of Theorem~\ref{theorem: single change detection}. 
\label{proof: thm 3.1 part (i)}
\end{proof}

\begin{proof}[Proof: first part of (ii)]
Let $\tilde{f}_n (\cdot)$ be the FDH estimator obtained from the sample $\left ( \boldsymbol{X}_t, Y_t \right )$ with $t = \eta + 1, \dots, n$. Theorem~\ref{theorem: exponential concentration in infty distance} guarantees that with probability at least $1 - Cn^{-1}$, where $C$ is given in Theorem~\ref{theorem: exponential concentration in infty distance}, it will hold that 
\begin{equation}
\left \| \tilde{f}_n - f_{(2)} \right \|_\infty \leq \left [ \frac{\log (n - \eta)}{n - \eta} \right ]^{\frac{1}{d+1}} \log (n) \leq \log^{\frac{3}{2}} \left ( n \right ) / \delta^{\frac{1}{d+1}}.
\label{equation: sup norm bound}
\end{equation}
Moreover since for all $\boldsymbol{x} \in \mathcal{X}$ it holds that $f_{(2)} \left ( \boldsymbol{x} \right ) \geq \hat{f}_n \left ( \boldsymbol{x} \right ) \geq \tilde{f}_n \left ( \boldsymbol{x} \right )$ we also have that 
\begin{equation}
\left ( f_{(2)} \left ( \boldsymbol{x} \right ) - \hat{f}_n \left ( \boldsymbol{x} \right ) \right ) \leq \left \| \tilde{f}_n - f_{(2)} \right \|_\infty, \hspace{2em} \forall x \in \mathcal{X}. 
\label{equation: dif sup bound}
\end{equation}
Finally, we will have that 
\begin{equation}
\hat{f}_n \left ( \boldsymbol{x} \right ) \geq \tilde{f}_n \left ( \boldsymbol{x} \right ) \geq x_0/ (2L), \hspace{2em} \forall \boldsymbol{x} \in \mathcal{X} \text{ s.t. } \left \| \boldsymbol{x} \right \| \geq x_0
\label{equation: f at x_0 bound}
\end{equation}
on the event $A_n := \left \{ \exists t \in \{ \eta+1, \ldots,n \} \text{ s.t. } R_t \geq 1/2 \text{ and } \left \| \boldsymbol{X}_t \right \| \geq x_0 / 2 \right \}$ since by Assumptions~\ref{assumption: production frontier function continuous} on $A_n$ for every such $\boldsymbol{x}$ there must be an $\boldsymbol{\tilde{\boldsymbol{x}}}$ with $x_0/2 \leq \left \| \boldsymbol{\tilde{x}} \right \| < x_0$ for which $\tilde{f}_n \left ( \boldsymbol{x} \right ) \geq \tilde{f}_n \left ( \boldsymbol{\tilde{x}} \right ) \geq \frac{1}{2} f_{(2)} \left ( \boldsymbol{\tilde{x}} \right ) \geq x_0 / (2L)$. The event $A_n$ in turn holds with probability at least $1 - n^{-1}$ because
\begin{align}
\mathbb{P} \left ( A_n \right ) & = 1 - \prod_{t = \eta + 1}^{n} \left [ 1 - \mathbb{P} \left ( R_t \geq 1/2 \mid \left \| \boldsymbol{X}_t \right \| \geq x_0 / 2 \right ) \mathbb{P} \left ( \left \| \boldsymbol{X}_t \right \| \geq x_0 / 2 \right ) \right ] \label{equation: A_n prob bound} \\
& \geq 1 - \left ( \max_{1 \leq t \leq n }\left [ 1 - \mathbb{P} \left ( R_t \geq 1/2 \mid \left \| \boldsymbol{X}_t \right \| \geq x_0 / 2 \right ) \mathbb{P} \left ( \left \| \boldsymbol{X}_t \right \| \geq x_0 / 2 \right ) \right ] \right )^{\delta} \nonumber \\ 
& \geq 1 - n^{-1}, \nonumber 
\end{align}
where passage to the final line holds with $C_{x,d}$ chosen sufficiently large in part (i) of Assumption~\ref{assumption: change point spacing}. Therefore, with high probability we have that
\begin{subequations}
\begin{align}
\hat{L}_{\hat{\eta}} \geq \hat{L}_{\eta} & = 2 N \left ( 1, \eta, \boldsymbol{x}_0 \right ) \log \left (  \left [ \max_{\substack{1 \leq t \leq \eta \\ \boldsymbol{X}_t > \boldsymbol{x}_0}} R_t \left ( \frac{f_{(1)} \left ( \boldsymbol{X}_t \right )}{f_{(2)} \left ( \boldsymbol{X}_t \right )} \right ) \left ( 1 + \frac{f_{(2)} \left ( \boldsymbol{X}_t \right ) - \hat{f}_n \left ( \boldsymbol{X}_t \right )}{\hat{f}_n \left ( \boldsymbol{X}_t \right )} \right ) \right ]^{-1} \right ) \nonumber \\
& \geq 2 N \left ( 1, \eta, \boldsymbol{x}_0 \right ) \log \left ( \left ( 1 / \mu \right ) \left [ 1 + \left \| \tilde{f}_n - f_{(2)} \right \|_\infty / \tilde{f}_n \left ( \boldsymbol{x}_0 \right ) \right ] ^ {-1} \right ) \nonumber \\
& \geq 2 N \left ( 1, \eta, \boldsymbol{x}_0 \right ) \left \{ \log \left ( 1 / \mu \right ) - \log \left ( 1 + (2L / x_0) \log^{\frac{3}{2}} (n) / \delta^{\frac{1}{d+1}} \right ) \right \} \label{equation: single change detection I} \\
& \geq 2 N \left ( 1, \eta, \boldsymbol{x}_0 \right ) \left \{ \log \left ( 1 / \mu \right ) - (2L / x_0) \log^{\frac{3}{2}} (n) / \delta^{\frac{1}{d+1}} \right \} \nonumber \\
& \geq \frac{\delta}{4} \log \left ( 1 / \mu \right ), 
\label{equation: single change detection III}
\end{align}
\end{subequations}
where in particular\eqref{equation: single change detection I}  holds with high probability due to \eqref{equation: f at x_0 bound} and Theorem~\ref{theorem: exponential concentration in infty distance}, and  \eqref{equation: single change detection III} holds with high probability due to Lemma~\ref{lemma: trimming} and part (ii) of Assumption~\ref{assumption: change point spacing}. 
Since $\lambda < \frac{\delta}{8} \log (1/\mu)$ by \eqref{equation: AMOC thresh condition} the change is detected and the first part of result (ii) is proved. 
\end{proof}

\begin{proof}[Proof: second part of (ii)]
Let $\lambda$ be as in Lemma~\ref{lemma: max bound}. We will show that with high probability
\begin{equation}
\left | \hat{\eta} - \eta \right | \leq \epsilon_n \coloneq C_2 \log (n) \vee C_3 \frac{\log (n)}{\log (1/\mu)}, 
\label{equation: epsilon n}
\end{equation}
which implies the second part of (ii). In particular we put $C_2 = C_{1,R}^{-1} \left ( 1 - e^{-1} \right ) ^ {-1}$ and $C_3 = (16/3) C_{1,R}^{-1} \left ( 1 - e^{-1} \right ) ^ {-1}$. Note that \eqref{equation: epsilon n} must hold if (i) $\hat{L}_{\eta} > \hat{L}_{\eta + \epsilon_n + j}$ and (ii) $\hat{L}_{\eta} > \hat{L}_{\eta - \epsilon_n - j}$ uniformly over all permissible $j$'s. Here we only prove (i) as (ii) can be shown analogously. First, with probability at least $1 - n^{-1}$ we must have that 
\begin{equation}
- \log \left [ M \left ( \eta + 1, \eta + \epsilon_n, \boldsymbol{x}_0 \right ) \right ] \leq \frac{1}{4} \log (1/\mu). 
\label{equation: null LR signal bound}
\end{equation}
This holds because the probability of \eqref{equation: null LR signal bound} not occurring is bounded as shown below, where in particular the $p$'s, $\overline{p}$'s and $F$'s are as defined in the proof of Lemma~\ref{lemma: max bound}: 
\begin{subequations}
\begin{align}
& \mathbb{P} \left ( M \left ( \eta + 1, \eta + \epsilon_n, \boldsymbol{x}_0 \right ) \leq e^{-\frac{1}{4} \log (1/\mu)} \right ) \nonumber \\
& = \sum_{k=0}^\epsilon \sum_{A \in F_{\eta + 1, \eta + \epsilon, k}} \prod_{t \in A} p_{t,\boldsymbol{x}_0} \mathbb{P} \left ( R_t \leq e^{-\frac{1}{4} \log (1 / \mu)} \mid \boldsymbol{X}_t > \boldsymbol{x}_0 \right ) \prod_{t \in A^c} \left ( 1 - p_{t, \boldsymbol{x}_0} \right )  \\ 
& \leq \sum_{k=0 }^{\epsilon_n} \binom{\epsilon_n}{k} \overline{p}_{\eta+1, \eta + \epsilon, \boldsymbol{x}_0} ^ {k} \left ( 1 - \overline{p}_{\eta+1, \eta + \epsilon, \boldsymbol{x}_0} \right ) ^ {\epsilon_n - k} \left ( 1 - C_{1,R} \left ( 1 - e^{-\frac{1}{4} \log (1 / \mu)} \right ) \right ) ^ k \nonumber \\ 
& = \left ( 1 - \frac{3}{4} C_{1,R} \left ( 1 - e^{-\frac{1}{4} \log (1/\mu)} \right ) \right ) ^ {\epsilon_n}.  
\label{equation: M epsilon > M eta bound}
\end{align}
\end{subequations}
Moreover, using the fact that any for any $x \in [0,1]$ the curve $x \mapsto (1 - e^{-x})$ lies everywhere above the line $x \mapsto (1 - e^{-1}) x$, if $\log (1/\mu) \leq 4$ then \eqref{equation: M epsilon > M eta bound} is bounded from above as 
\begin{equation}
\left [ 1 - \frac{3 C_{1,R}}{16} \left ( 1 - e^{-1} \right ) \log (1/\mu) \right ] ^ {\epsilon_n} \leq \exp \left ( - \epsilon_n \left [ \frac{3 C_{1,R}}{16} \left ( 1 - e^{-1} \right ) \log (1/\mu) \right ] \right ) \leq \frac{1}{n}. 
\label{equation: (1-e^x) property bound i}
\end{equation}
On the other hand if $\log (1/\mu) > 4$ then \eqref{equation: M epsilon > M eta bound} is bounded as
\begin{equation}
\left [ 1 - C_{1,R} \left ( 1 - e^{-1} \right ) \right ] ^ {\epsilon_n} \leq \exp \left ( -\epsilon_n C_{1,R} \left ( 1 - e^{-1} \right ) \right ) \leq \frac{1}{n}, 
\label{equation: (1-e^x) property bound ii}
\end{equation}
where the second inequality follows from the fact that $\log (1-x) \leq x$ for $x \in [0,1]$. Note also that on \eqref{equation: sup norm bound} and \eqref{equation: f at x_0 bound} the event \eqref{equation: null LR signal bound} implies that
\begin{equation}
\hat{M} \left ( \eta + 1, \eta + \epsilon_n + j, \boldsymbol{x}_0 \right ) > \hat{M} \left ( 1, \eta, \boldsymbol{x}_0 \right ), \hspace{2em} j = 0, \dots, n - \eta - \epsilon_n, 
\label{equation: M epsilon > M eta}
\end{equation}
which holds because
\begin{align}
\eqref{equation: M epsilon > M eta} & \supseteq \cap_{j=0}^{n-\eta-\epsilon_n} \left \{ M \left ( \eta + 1, \eta + \epsilon_n + j, \boldsymbol{x}_0 \right ) > \mu \left [ 1 + (2L / x_0) \log^{\frac{3}{2}} (n) / \delta^\frac{1}{d+1} \right ] \right \} \nonumber \\
& \supseteq \cap_{j=0}^{n-\eta-\epsilon_n} \left \{ M \left ( \eta + 1, \eta + \epsilon_n + j, \boldsymbol{x}_0 \right ) > e^{-\frac{1}{2} \log (1 / \mu)} \right \} \nonumber \\
& \supseteq \left \{ M \left ( \eta + 1, \eta + \epsilon_n, \boldsymbol{x}_0 \right ) > e^{-\frac{1}{4} \log (1 / \mu)} \right \}.
\label{equation: M epsilon > M eta chain}
\end{align}
Then, \eqref{equation: M epsilon > M eta} and \eqref{equation: M epsilon > M eta chain} imply that
\begin{equation}
- \log \left [ \hat{M} \left ( \eta + 1, \eta + \epsilon_n + j, \boldsymbol{x}_0 \right ) \right ] \leq \frac{1}{4} \log \left ( 1 / \mu \right ), \hspace{2em} j = 0, \dots, n - \eta - \epsilon_n.
\label{equation: M epsilon signal strength bound}
\end{equation}
Finally, on the events \eqref{equation: M epsilon > M eta} and \eqref{equation: M epsilon signal strength bound} we have that
\begin{subequations}
\begin{align}
 \hat{L}_{\eta} - \hat{L}_{\eta + \epsilon_n + j} & = - 2 N \left ( 1, \eta, \boldsymbol{x}_0 \right ) \log \left [ \hat{M} \left ( 1, \eta, \boldsymbol{x}_0 \right ) \right ] \nonumber \\ 
 & \quad + 2 N \left ( 1, \eta + \epsilon_n + j, \boldsymbol{x}_0 \right ) \log \left [ \hat{M} \left ( 1, \eta + \epsilon_n + j, \boldsymbol{x}_0 \right ) \right ] \label{equation: localization rate argument} \\ 
& \geq 2 N \left ( 1, \eta, \boldsymbol{x}_0 \right ) \log \left [ \left ( 1 / \mu \right ) \left [ 1 + (2L / x_0) \log^{\frac{3}{2}} (n) / \delta^{\frac{1}{d+1}} \right ] ^ {-1} \right ] \nonumber \\
& \quad + 2 N \left ( 1, \eta + \epsilon_n + j, \boldsymbol{x}_0 \right ) \log \left [ M \left ( \eta + 1, \eta + \epsilon_n + j, \boldsymbol{x}_0 \right ) \right ] \label{equation: single change localization I} \\ 
& \geq N \left ( 1, \eta, \boldsymbol{x}_0 \right ) \log \left ( 1 / \mu \right ) + 2 N \left ( 1, \eta, \boldsymbol{x}_0 \right ) \log \left [ M \left ( \eta + 1, \eta + \epsilon_n + j, \boldsymbol{x}_0 \right ) \right ] \nonumber \\
& \quad + 2 N \left ( \eta + 1, \eta + \epsilon_n + j, \boldsymbol{x}_0 \right ) \log \left [ M \left ( \eta + 1, \eta + \epsilon_n + j, \boldsymbol{x}_0 \right ) \right ] \nonumber \\
& \geq 2 N \left ( 1, \eta, \boldsymbol{x}_0 \right ) \left \{ \frac{1}{2} \log \left ( 1/\mu \right ) + \log \left [ M \left ( \eta + 1, \eta + \epsilon_n + j, \boldsymbol{x}_0 \right ) \right ] \right \} - \lambda \label{equation: single change localization II} \\ 
& \geq \frac{1}{2} N \left ( 1, \eta, \boldsymbol{x}_0 \right ) \log \left ( 1/\mu \right ) - \lambda \label{equation: single change localization III} \\
& \geq \frac{\delta}{8} \log (1/\mu) - \lambda, \label{equation: single change localization IV} 
\end{align}
\end{subequations}
where in particular \eqref{equation: single change localization I} holds with high probability due to \eqref{equation: f at x_0 bound} and Theorem~\ref{theorem: exponential concentration in infty distance}, \eqref{equation: single change localization II} holds with high probability due to Lemma~\ref{lemma: max bound}, \eqref{equation: single change localization III} holds with high probability due to~\ref{equation: M epsilon signal strength bound}, and \eqref{equation: single change localization IV} holds with high probability due to Lemma~\ref{lemma: trimming} provided we choose $C_1 \geq \frac{1}{4}$ in Theorem~\ref{theorem: single change detection}. By condition \eqref{equation: AMOC thresh condition} and $C_3$ appropriately chosen we will have that $(\delta / 8) \log (1/\mu) - \lambda > 0$ which establishes \eqref{equation: epsilon n}.
\end{proof}

\subsubsection{Proof of Theorem~\ref{lemma: single change inference}}

\begin{proof}
Let $\nu$ be any quantity such that (i) $\nu \rightarrow  \infty$ as $n \rightarrow \infty$ and (ii) $\nu = o \left ( \eta \right )$. Observe that 
\begin{equation*}
\mathbb{P} \left ( \hat{\eta} - \eta < 0 \right ) \leq \mathbb{P} \left ( \cup_{\tau = 1}^{\nu-1} \left \{ \hat{L}_{\eta} < \hat{L}_{\tau} \right \} \right ) + \mathbb{P} \left ( \cup_{\tau = \nu}^{\eta-1} \left \{ \hat{L}_{\eta} < \hat{L}_{\tau} \right \} \right ). 
\end{equation*}
Recall from \eqref{equation: quasi-LR stat} that we follow the convention $0 \times \infty = 0$. Therefore, putting $\tau^*_\nu = \min \{ \tau < \nu \mid N (1, \tau, \boldsymbol{x_0}) > 0 \}$ we have that
\begin{align*}
& \mathbb{P} \left ( \cup_{\tau = 1}^{\nu-1} \left \{ \hat{L}_{\eta} < \hat{L}_{\tau} \right \} \right ) \\
& \quad \quad \leq \mathbb{P} \left ( - 2 N(1, \nu-1, \boldsymbol{x_0}) \log \left [ \hat{M} \left ( 1, \tau^*_\nu, \boldsymbol{x_0} \right ) \right ] > - 2 N(1, \eta, \boldsymbol{x_0}) \log \left [ \hat{M} \left ( 1, \eta, \boldsymbol{x_0} \right ) \right ] \right ). 
\end{align*}
Then, the facts that
\begin{enumerate}
    \item $- 2  \log \left [ \hat{M} \left ( 1, \tau^*_\nu, \boldsymbol{x_0} \right ) \right ] = \mathcal{O}_\mathbb{P} (1)$ by Assumption~\ref{assumption: X and R distribution};
    \item $N(1, \nu-1, \boldsymbol{x_0}) / N(1, \eta, \boldsymbol{x_0}) \overset{p}{\rightarrow} 0$ since $\nu = o(\eta)$, and 
    \item $\hat{M} \left ( 1, \eta, \boldsymbol{x_0} \right ) \overset{p}{\rightarrow} \rho \in (0, \mu)$ due to Scenario~\ref{scenario: AMOC}
\end{enumerate}
imply that $\mathbb{P} ( \cup_{\tau = 1}^{\nu-1} \{ \hat{L}_{\eta} < \hat{L}_{\tau} \} ) \rightarrow 0$ as $n \rightarrow \infty$. Next we have that
\begin{align*}
& \mathbb{P} \left ( \cup_{\tau = \nu}^{\eta-1} \left \{ \hat{L}_{\eta} < \hat{L}_{\tau} \right \} \right ) \\
& \quad = \mathbb{P} \left ( \cup_{\tau = \nu}^{\eta-1} \left \{ - 2 N \left ( 1, \eta, \boldsymbol{x}_0 \right ) \log  \left [ \hat{M} \left ( 1, \eta, \boldsymbol{x}_0 \right ) \right ] < - 2 N \left ( 1, \tau, \boldsymbol{x}_0 \right ) \log  \left [ \hat{M} \left ( 1, \tau, \boldsymbol{x}_0 \right ) \right ] \right \} \right ) \\
& \quad \leq \mathbb{P} \left ( - 2 N \left ( 1, \eta, \boldsymbol{x}_0 \right ) \log  \left [ \hat{M} \left ( 1, \eta, \boldsymbol{x}_0 \right ) \right ] < - 2 N \left ( 1, \eta - 1, \boldsymbol{x}_0 \right ) \log \left [ \min_{\tau = \nu, \dots, \eta - 1} \hat{M} \left ( 1 , \tau, \boldsymbol{x}_0 \right ) \right ] \right ).
\end{align*}
Then, the facts that
\begin{enumerate}
    \item $N \left ( 1, \eta - 1, \boldsymbol{x}_0 \right ) / N \left ( 1, \eta, \boldsymbol{x}_0 \right ) \overset{p}{\rightarrow} 1$,  
    \item $\hat{M} \left ( 1, \eta, \boldsymbol{x}_0 \right ) \overset{p}{\rightarrow} \rho \in (0, \mu)$, and 
    \item $\min_{\tau = \nu, \dots, \eta-1} \hat{M} \left ( 1, \tau, \boldsymbol{x}_0 \right ) \overset{p}{\rightarrow} \rho \in  (0, \mu)$
\end{enumerate}
imply that $\mathbb{P} ( \cup_{\tau = \nu}^{\eta-1} \{ \hat{L}_{\eta} < \hat{L}_{\tau} \} ) \rightarrow 0$ as $n \rightarrow \infty$. from which we obtain that $\lim_{n \rightarrow \infty} \mathbb{P} \left ( \hat{\eta} - \eta < 0 \right ) = 0$. 
For any $k \geq 1$ we consequently have that
\begin{align*}
& \mathbb{P} \left ( \hat{\eta} - \eta \geq k \right ) \\
& = \mathbb{P} \left ( \cap_{\tau = \eta}^{\eta+k-1} \left \{ \hat{L}_{\eta+k} > \hat{L}_\tau \right \} \right ) \\
& = \mathbb{P} \left ( \cap_{\tau = \eta}^{\eta + k -1} \left \{ - 2 N \left ( 1, \eta + k, \boldsymbol{x}_0 \right ) \log \left [ \hat{M} \left ( 1, \eta + k, \boldsymbol{x}_0 \right ) \right ] > - 2 N \left ( 1, \eta, \boldsymbol{x}_0 \right ) \log \left [ \hat{M} \left ( 1, \tau, \boldsymbol{x}_0 \right ) \right ] \right \}  \right ) \\
& = \mathbb{P} \left ( \cap_{\tau=\eta}^{\eta+k-1} \left \{ \hat{M} \left ( 1, \tau, \boldsymbol{x}_0 \right ) \vee \hat{M} \left ( \tau + 1, \eta + k, \boldsymbol{x}_0 \right ) < \exp \left ( \frac{N \left ( 1, \tau, \boldsymbol{x}_0 \right )}{N \left ( 1, \eta + k, \boldsymbol{x}_0 \right )} \log \left [ \hat{M} \left ( 1, \tau, \boldsymbol{x}_0 \right ) \right ] \right ) \right \}  \right ) \\
& = \mathbb{P} \left ( \cap_{\tau=\eta}^{\eta+k-1} \left \{ \left \{ \hat{M} \left ( 1, \tau, \boldsymbol{x}_0 \right ) < \exp \left ( \frac{N \left ( 1, \tau, \boldsymbol{x}_0 \right )}{N \left ( 1, \eta + k, \boldsymbol{x}_0 \right )} \log \left [ \hat{M} \left ( 1, \tau, \boldsymbol{x}_0 \right ) \right ] \right ) \right \} \cap \right. \right . \\
& \hspace{4em} \left . \left . \left \{ \hat{M} \left ( \tau + 1, \eta + k, \boldsymbol{x}_0 \right ) < \exp \left ( \frac{N \left ( 1, \tau, \boldsymbol{x}_0 \right )}{N \left ( 1, \eta + k, \boldsymbol{x}_0 \right )} \log \left [ \hat{M} \left ( 1, \tau, \boldsymbol{x}_0 \right ) \right ] \right ) \right \} \right \}  \right ) \\
& = \mathbb{P} \left ( \cap_{\tau = \eta}^{\eta + k - 1} \left \{ \hat{M} \left ( \tau + 1, \eta + k, \boldsymbol{x}_0 \right ) < \exp \left ( \frac{N \left ( 1, \tau, \boldsymbol{x}_0 \right )}{N \left ( 1, \eta + k, \boldsymbol{x}_0 \right )} \log \left [ \hat{M} \left ( 1, \tau, \boldsymbol{x}_0 \right ) \right ] \right ) \right \}  \right ).
\end{align*} 
As $n \rightarrow \infty$, for any fixed $k$, we have that
\begin{enumerate}
    \item $M \left ( 1, \eta, \boldsymbol{x}_0 \right ) \overset{p}{\rightarrow} \rho$
    \item $N \left ( 1, \eta, \boldsymbol{x}_0 \right ) / N \left ( 1, \eta + k, \boldsymbol{x}_0 \right ) \overset{p}{\rightarrow} 1$
    \item $\hat{M} \left ( t_1, t_2, \boldsymbol{x}_0 \right ) \overset{d}{\rightarrow} M \left ( t_1, t_2, \boldsymbol{x}_0 \right )$ for $\eta < t_1 \leq t_2 \leq n$ and $t_1,t_2$ fixed
\end{enumerate}
and as such combined with the fact that $\lim_{n \rightarrow \infty} \mathbb{P} \left ( \mathcal{E} \right ) = 1$ by Slutsky's theorem and the continuous mapping theorem we obtain following, where for economy of notation we let $M \left ( t_1, t_2, \boldsymbol{x}_0 \right ) = 0$ whenever $t_1 > t_2$
\begin{align*}
\lim_{n \rightarrow \infty} \mathbb{P} \left ( \hat{\eta} - \eta \geq k \right ) & = \mathbb{P} \left ( \cap_{\tau = \eta}^{\eta + k - 1} \left \{ M \left ( \tau + 1, \eta + k, \boldsymbol{x}_0 \right ) < \rho \vee M \left ( \eta + 1, \tau, \boldsymbol{x}_0 \right ) \right \} \right ) \\
& \leq \mathbb{P} \left ( \cap_{\tau = \eta}^{\eta + k - 1} \left \{ M \left ( \tau + 1, \eta + k, \boldsymbol{x}_0 \right ) < \mu \vee M \left ( \eta + 1, \tau, \boldsymbol{x}_0 \right ) \right \} \right ) \\
& = \mathbb{P} \left ( M \left ( \eta + 1, \eta + k, \boldsymbol{x}_0 \right ) < \mu \right ) \\
& = \prod_{t=\eta+1}^{\eta+k} \mathbb{P} \left ( R_t \mathbf{1}_{\left \{ \boldsymbol{X}_t > \boldsymbol{x}_0 \right \}} < \mu \right ). 
\end{align*}
Therefore if the random variables $\left ( \boldsymbol{X}_t, R_t \right )_{t=1,\dots,n}$ are identically distributed we have that
\begin{equation*}
\lim_{n \rightarrow \infty} \mathbb{P} \left ( \hat{\eta} - \eta + 1 \leq k \right ) \geq
\begin{cases}
1 - \left ( 1 - \mathbb{P} \left ( R_1 \mathbf{1}_{\left \{ \boldsymbol{X}_1 > \boldsymbol{x}_0 \right \}} \geq \mu \right ) \right )^{\left \lfloor k \right \rfloor} & \text{ if } k \geq 1 \\
0 & \text{ if } k < 1
\end{cases},
\end{equation*}
which proves part (ii) of the lemma. For part (i) if the lemma it is enough to observe that by Assumption~\ref{assumption: X and R distribution} for each $t = 1, \dots, n$
\begin{align*}
\mathbb{P} \left ( R_t \mathbf{1}_{\left \{ \boldsymbol{X}_t > \boldsymbol{x}_0 \right \}} \geq \mu \right ) & = \mathbb{P} \left ( R_t \geq \mu \mid \boldsymbol{X}_t > \boldsymbol{x}_0 \right ) \mathbb{P} \left ( \boldsymbol{X}_t > \boldsymbol{x}_0 \right ) \\
& \geq C_{1,R} \left ( 1 - \mu \right ) C_{1,X} \operatorname{mes} \left ( \mathcal{X} \setminus \{ \boldsymbol{x} \leq \boldsymbol{x}_0 \} \right ). 
\end{align*}
and consequently
\begin{align*}
& \lim_{n \rightarrow \infty} \mathbb{P} \left ( \hat{\eta} - \eta + 1 \leq k \right ) \\
& \quad \quad \geq \begin{cases}
1 - \left ( 1 - C_{1,R} \left ( 1 - \mu \right ) C_{1,X} \operatorname{mes} \left ( \mathcal{X} \setminus \{ \boldsymbol{x} \leq \boldsymbol{x}_0 \} \right ) \right )^{\left \lfloor k \right \rfloor} & \text{ if } k \geq 1 \\
0 & \text{ if } k < 1
\end{cases}. 
\end{align*}
\end{proof}

\subsubsection{Proof of Proposition~\ref{lemma: global change minimax lower bound}}

\begin{proof}

Let $P_{1:n} = \otimes_{t=1}^n P_t$ denote the joint distribution of the random variables $\left \{ \left ( \boldsymbol{X}_t, Y_t \right ) \mid t = 1, \dots, n \right \}$ with
\begin{equation*}
Y_t = \begin{cases}
\mu f \left ( \boldsymbol{X}_t \right ) R_t & \text{ for } t = 1, \dots, \delta \\
f \left ( \boldsymbol{X}_t \right ) R_t & \text{ for } t = \delta + 1, \dots, n, 
\end{cases}
\end{equation*}
the $\boldsymbol{X}$'s and $R$'s are mutually independent and identically distributed with marginal density functions $p_{\boldsymbol{X}} \left ( \cdot \right )$ and $p_R (\cdot)$ respectively, satisfying Assumption~\ref{assumption: X and R distribution}. Let $Q_{1:n} = \otimes_{t=1}^n Q_t$ denote the joint law of the random variables $\{ ( \boldsymbol{\tilde{X}}_t, \tilde{Y}_t ) \mid t = 1, \dots, n \}$ where 
\begin{equation*}
\tilde{Y}_t = \begin{cases}
\mu f ( \boldsymbol{\tilde{X}}_t ) \tilde{R}_t & \text{ for } t = 1, \dots, \delta + \tau \\
f ( \boldsymbol{\tilde{X}}_t ) \tilde{R}_t & \text{ for } t = \delta + \tau + 1, \dots, n, 
\end{cases}
\end{equation*}
the distribution of the $\boldsymbol{\tilde{X}}$'s and $\tilde{R}$'s is the same as that of the $X$'s and $R$'s, and $\tau$ is a positive integer no larger than $n - \delta - 1$. Note that $\eta \left ( P_{1:n} \right ) = \delta$ and $\eta \left ( Q_{1:n} \right ) = \delta + \tau$. Therefore applying Le Cam's two point method, see for example Section 15.2 in \cite{wainwright2019high}, it holds that
\begin{equation*}
\inf_{\hat{\eta}} \sup_{P \in \mathcal{Q}_n} \mathbb{E}_{P} \left [ \left | \hat{\eta} - \eta \right | \right ] \geq \tau \left ( 1 - \text{TV} \left ( P_{1:n}, Q_{1:n} \right ) \right ) \geq \frac{\tau}{2} \exp \left ( - \text{KL} \left ( Q_{1:n}, P_{1:n} \right ) \right )
\end{equation*}
where $\text{TV} \left ( \cdot, \cdot \right )$ denotes the total variation distance, $\text{KL} \left ( \cdot, \cdot \right )$ denotes the Kullback-Liebler divergence, and the second inequality follows from Lemmas 2.1 and 2.6 in \cite{tsybakov2009nonparametric} respectively. Since $P_{1:n}$ and $Q_{1:n}$ are product measures we have that 
\begin{align*}
\text{KL} \left ( Q_{1:n}, P_{1:n} \right ) & = \sum_{\delta < t \leq \delta + \tau} \text{KL} \left ( Q_t, P_t \right ) \\
& = \tau \int_{\mathcal{X}} \int_{\mathcal{Y}} \log \left ( \frac{f_{\boldsymbol{\tilde{\boldsymbol{X}}}_{\delta+1}, \tilde{Y}_{\delta+1}} \left ( \boldsymbol{x}, y \right )}{f_{\boldsymbol{X}_{\delta+1}, Y_{\delta+1}} \left ( \boldsymbol{x}, y \right )} \right ) f_{\boldsymbol{\tilde{X}}_{\delta+1}, \tilde{Y}_{\delta+1}} \left ( \boldsymbol{x}, y \right ) \mathrm{d} y \mathrm{d} \boldsymbol{x}
\end{align*}
where $f_{\boldsymbol{X}_{\delta+1}, Y_{\delta+1}} \left ( \cdot, \cdot \right )$ and  $f_{\boldsymbol{\tilde{X}}_{\delta+1}, \tilde{Y}_{\delta+1}} \left ( \cdot, \cdot \right )$ denote respectively the joint densities of $\left (\boldsymbol{X}_{\delta+1}, Y_{\delta+1} \right )$ and $( \boldsymbol{\tilde{X}}_{\delta+1}, \tilde{Y}_{\delta+1} )$. Moreover,
\begin{align}
f_{\boldsymbol{X}_{\delta+1}, Y_{\delta+1}} \left ( \boldsymbol{x}, y \right ) & = \frac{\partial}{\partial \boldsymbol{x} \partial y} \mathbb{P} \left ( \boldsymbol{X}_{\delta+1} \leq \boldsymbol{x}, Y_{\delta+1} \leq y \right ) \nonumber \\
& = \frac{\partial}{\partial \boldsymbol{x}} \int_{\boldsymbol{z} \leq \boldsymbol{x}} \frac{\partial}{\partial y} \mathbb{P} \left ( f (\boldsymbol{z} R_{\delta+1} \leq y \right ) p_{\boldsymbol{X}} \left ( \boldsymbol{z} \right ) \mathrm{d} \boldsymbol{z} \nonumber \\
& = \frac{\partial}{\partial\boldsymbol{x}} \int_{\boldsymbol{z} \leq \boldsymbol{x}} \frac{1}{f(\boldsymbol{z})} p_R \left ( y / f (\boldsymbol{z}) \right ) p_{\boldsymbol{X}} \left ( \boldsymbol{z} \right ) \mathrm{d} \boldsymbol{z} \nonumber \\ 
& = \frac{1}{f(\boldsymbol{x})} p_R \left ( y / f (\boldsymbol{x}) \right ) p_{\boldsymbol{X}} \left ( \boldsymbol{x} \right ) \label{equation: joint density delta + 1}
\end{align}
and likewise 
\begin{equation*}
f_{\boldsymbol{\tilde{X}}_{\delta+1}, \tilde{Y}_{\delta+1}} \left ( x, y \right ) =  \frac{1}{\mu f(\boldsymbol{x})} p_R \left ( y / f (\boldsymbol{x}) \right ) p_{\boldsymbol{X}} \left ( \boldsymbol{x} \right ).
\end{equation*}
It therefore follows that $\text{KL} \left ( Q_{1:n}, P_{1:n} \right ) = \tau \log \left ( \frac{1}{\mu} \right )$ and therefore
\begin{equation}
\inf_{\hat{\eta}} \sup_{P \in \mathcal{Q}} \mathbb{E}_{P} \left [ \left | \hat{\eta} - \eta \right | \right ] \geq \frac{\tau}{2} \exp \left ( - \tau \log \left ( \frac{1}{\mu} \right ) \right ). 
\label{equation: minimax tau bound}
\end{equation}
Finally set $\tau = \min \left (  \left \lceil \left ( \log (1/\mu) \right ) ^{-1} \right \rceil, n - \delta - 1 \right )$. Note that if $n - \delta - 1 < \left \lceil \left ( \log (1/\mu) \right ) ^{-1} \right \rceil$ then 
\begin{equation}
\delta \log \left ( \frac{1}{\mu} \right ) \leq \frac{\delta}{n - 2 - \delta} \leq \frac{n}{n/2 - 2} \leq 10 
\label{equation: change energy inequality}
\end{equation}
for $n \geq 5$. However we also have $\delta \log \left ( 1 / \mu \right ) \geq \zeta_n$ where $\zeta_n \rightarrow \infty$. Therefore, there is a constant $n_0$ depending on the sequence of $\zeta$'s such that for $n > n_0$ \eqref{equation: change energy inequality} does not hold, so we must have that  
\begin{equation*}
\tau = \left \lceil \left ( \log (1/\mu) \right ) ^{-1} \right \rceil
\end{equation*}
for $n$ sufficiently large. Plugging this value for $\tau$ into \eqref{equation: minimax tau bound} and using the fact that $\lceil x \rceil \geq x$ for any $x \geq 0$ completes the proof. 
\end{proof}

\subsubsection{Proof of Theorem~\ref{theorem: multiple change detection}}

\begin{proof}[Proof: part (i)]
Let $\lambda$ be as in Lemma~\ref{lemma: max bound}. Arguing as in the proof of part (i) of Theorem~\ref{theorem: single change detection} we have that with high probability Algorithm~\ref{algorithm: multiple change detection} will not spuriously detect a change when acting on a stretch of data free from change points. We now show that with high probability all changes are detected.  Let $\hat{f}_n \left ( \cdot \right )$ and $\tilde{f}_n \left ( \cdot \right )$ be the FDH estimators computed on data indexed on $\left \{ 1, \dots, n \right \}$ and $\left \{ \eta_{K} + 1, \dots, n \right \}$ respectively. For all $\boldsymbol{x} \in \mathcal{X}$ it holds that $f_{(K+1)} \left ( \boldsymbol{x} \right ) \geq \hat{f}_n \left ( \boldsymbol{x} \right ) \geq \tilde{f}_n \left ( \boldsymbol{x} \right )$, and by Theorem~\ref{theorem: exponential concentration in infty distance} with probability at least $1 - C n^{-2}$ it will hold that 
\begin{equation}
\left \| \tilde{f}_n - f_{(K+1)} \right \|_\infty \leq  2 \left [ \frac{\log (n - \eta_K)}{n - \eta_K} \right ]^{\frac{1}{d+1}} \log (n) \leq 2 \log^{\frac{3}{2}} \left ( n \right ) / \delta^{\frac{1}{d+1}}_K.
\label{equation: f_N sup norm bound}
\end{equation}
Moreover on the event 
\begin{equation*}
    A_{\eta_K} := \left \{ \exists t \in \left [ \eta_K+1, n \right ] \text{ s.t. } R_t \geq 1/2 \text{ and } \left \| \boldsymbol{X}_t \right \| \geq x_0 / 2 \right \},
\end{equation*}
which by part (i) of Assumption~\ref{assumption: multiple change point spacing} and calculations analogous to \eqref{equation: A_n prob bound} holds with probability at least $1 - n^{-2}$, we must have that
\begin{equation}
\tilde{f}_n \left ( \boldsymbol{x} \right ) \geq x_0 / (2L), \hspace{2em} \forall \boldsymbol{x} \in \mathcal{X} \text{ s.t. } \left \| \boldsymbol{x} \right \| \geq x_0. 
\label{equation: f_N at x_0 bound}
\end{equation}
Consequently, arguing as in \eqref{equation: single change detection III} and making use of Assumption~\ref{assumption: multiple change point spacing} we have that on the event \eqref{equation: f_N at x_0 bound}
\begin{align*}
\hat{L}_{\eta_K - \delta_K / 2 + 1, \eta_K, n} & \geq 2 N \left ( \eta_K - \delta_K / 2 + 1, \eta_K, \boldsymbol{x}_0 \right ) \log \left [ \left ( 1 / \mu_K \right ) \left [ 1 + \left \| \tilde{f}_n - f_{(K)} \right \|_\infty / \tilde{f}_n \left ( \boldsymbol{x}_0 \right ) \right ] ^ {-1} \right ] \\
& \geq 2 N \left ( \eta_K - \delta_K / 2 + 1, \eta_K, \boldsymbol{x}_0 \right ) \left \{ \log \left ( 1 / \mu_K \right ) - \log \left [ 1 + (4L / x_0) \log^{\frac{3}{2}} (n) / \delta_K^{\frac{1}{d+1}} \right ] \right \} \\
& \geq 2 N \left ( \eta_K - \delta_K / 2 + 1, \eta_K, \boldsymbol{x}_0 \right ) \left \{ \log \left ( 1 / \mu_K \right ) - (2L / x_0) \log^{\frac{3}{2}} (n) / \delta_K^{\frac{1}{d+1}} \right \} \\
& \geq \frac{\delta_K}{8} \log \left ( 1 / \mu_K \right ).
\end{align*}
Therefore, since $\lambda < (\delta_K / 8) \log \left ( 1 / \mu_K \right )$ by condition \eqref{equation: multiple change threshold thresh condition} with $C_3$ chosen appropriately, we have that Algorithm~\ref{algorithm: multiple change detection} will detect the right-most change at the latest when inspecting data indexed on $\left \{ \eta_K - \delta_K/2 + 1, \dots, n \right \}$. Consequently, Algorithm~\ref{algorithm: multiple change detection} re-starts at an index at least $\delta_{K-1}/2$ from the change point location $\eta_{K-1}$. By induction each change is detected with probability at least $1 - Cn^{-2}$. Since necessarily $K \leq n$ all changes are detected uniformly with probability at least $1 - Cn^{-1}$. This proves part (i) of the Theorem. 
\end{proof}

\begin{proof}[Proof: part (ii)]
From the proof of part (i) and definition \eqref{equation: s e refit points} it is clear that with high probability $\eta_k - \tilde{s}_k \geq \delta_{k}/2$ and $\tilde{e}_k - \eta_k \geq \delta_k / 2$ uniformly over all $k = 1, \dots, K$. Then, arguments identical to \eqref{equation: localization rate argument}-\eqref{equation: single change localization III} give that for any $k \in \left \{ 1, \dots, K \right \}$ and any integers $1 \leq \check{s}_k < \check{e}_k' \leq n$ with $\check{s}_k' - \check {e}_k' \geq \delta_k$ and $\eta_k - \check{s}_k \geq \nu$ putting $\check{\eta}_k = \argmax_{\check{s}_k + \nu \leq \tau \leq \check{e}_k} \hat{L}_{\check{s}_k, \tau, \check{e}_k}$ that with probability at least $1 - C_1 n^{-2}$
\begin{equation*}
\left | \check{\eta}_k - \eta_k \right | \leq C_2 \log (n) \vee C_3 \frac{\log (n)}{\log (1/ \mu_k)}
\end{equation*}
for absolute constants $C_1, C_2, C_3 > 0$. Since the above result holds for the $\tilde{s}$'s and $\tilde{e}$'s recovered in the proof of part (i) and $K < n$ the theorem is proved. 
\end{proof}

\subsubsection{Proof of Theorem~\ref{theorem: consistent detection under local alternatives}}

\begin{proof}[Proof: part (i)]
Let $\lambda$ be as in Lemma~\ref{lemma: multiscale gird max bound}. Arguing as in the proof of part (i) of Theorem~\ref{theorem: single change detection} the procedure outlined in Section~\ref{section: local change methodology} will not spuriously detect a change when acting on a stretch of data free from change points. We now show that with high probability all changes are detected. Introduce the quantity
\begin{equation}
\mathcal{X}_K' \in \argmax_{\mathcal{X}' \in \mathfrak{X}} \operatorname{mes} \left ( \mathcal{X}' \right ) \quad \text{subject to} \quad \mathcal{X}' \subseteq \mathcal{X}_{(K)}. 
\label{equation: X_K'}
\end{equation}
Note that by part (i) of Assumption~\ref{assumption: size of regions},  with $A_n$ chosen as in Theorem~\ref{theorem: consistent detection under local alternatives}, $\mathcal{X}_K'$ exists for $C_\chi$ sufficiently large. The quantity may not be unique, in which case the ensuing argument holds for any element in the admissible set. We have that
\begin{subequations}
\begin{align}
& \hat{L}_{\eta_K - \delta_K / 2 + 1, \eta, n}^{\mathcal{X}_K'} \nonumber \\
& \geq 2 N \left ( \eta_K - \delta_K / 2 + 1, \eta_K, \mathcal{X}_K' \right ) \left \{ \log \left ( 1 / \mu_K \right ) - (2L / x_0) \log^{\frac{3}{2}} (n) / \delta_K^{\frac{1}{d+1}} \right \} \label{equation: multiscale detection I} \\
& \geq N \left ( \eta_K - \delta_K / 2 + 1, \eta_K, \mathcal{X}_K' \right ) \log ( 1 / \mu_K ) \label{equation: multiscale detection II} \\
& \geq 2^{-1} C_{1,X} \operatorname{mes} \left ( \mathcal{X}_K' \right ) \delta_K \log \left ( 1 / \mu_K \right ) \label{equation: multiscale detection III} \\
& \geq 2^{-(d+1)} C_{1,X} \operatorname{mes} \left ( \mathcal{X}_{(K)} \right ) \delta_K \log \left ( 1 / \mu_K \right ) \label{equation: multiscale detection IV}
\end{align}
\end{subequations}
where in particular \eqref{equation: multiscale detection I} holds with high probability due to \eqref{equation: f_N sup norm bound} and \eqref{equation: f_N at x_0 bound}, \eqref{equation: multiscale detection II} holds due to Assumption~\ref{assumption: multiple change point spacing}, \eqref{equation: multiscale detection III} holds with high probability due to Lemma~\ref{lemma: trimming multiscale grid}, and \eqref{equation: multiscale detection IV} holds since necessarily $\operatorname{mes} \left ( \mathcal{X}_K' \right ) \geq 2^{-d} \operatorname{mes} \left ( \mathcal{X}_{(K)} \right )$. Therefore, on a high probability set the first change is detected provided $C_4$ is sufficiently large, and the algorithm restarts no further than $\delta_{K-1} / 2$ from the $(K-1)$-th change point. By induction each subsequent change is detected with high probability. This proves part (i) of the Theorem.
\end{proof}

\begin{proof}[Proof: part (ii)]
For each $k = 1, \dots, K$ introduce the quantity
\begin{equation*}
\epsilon_{n,k} = 2 C_{1,R}^{-1} C_{1,X}^{-1} \left ( 1 - e^{-1} \right )^{-1} \log^2 (n) \vee 2 \frac{C_{1,R}^{-1} C_{1,X}^{-1} \left ( 1 - e^{-1} \right )^{-1} 2^{d+3}}{\operatorname{mes} \left ( \mathcal{X}_{(k)} \right ) \log (1 / \mu_k)} \log^2 (n). 
\end{equation*}
Similar to \eqref{equation: epsilon n} we will show that (i) $\max_{\mathcal{X}' \in \mathfrak{X}}\hat{L}^{\mathcal{X}'}_{\tilde{s}_k, \eta_k, \tilde{e}_k} > \max_{\mathcal{X}' \in \mathfrak{X}} \hat{L}^{\mathcal{X}'}_{\tilde{s}_k, \eta_k + \epsilon_{n,k} + j, \tilde{e}_K}$ and (ii) $\max_{\mathcal{X}' \in \mathfrak{X}}\hat{L}^{\mathcal{X}'}_{\tilde{s}_k, \eta_k, \tilde{e}_k} > \max_{\mathcal{X}' \in \mathfrak{X}} \hat{L}^{\mathcal{X}'}_{\tilde{s}_k, \eta_k + \epsilon_{n,k} - j, \tilde{e}_k}$ uniformly over all permissible $j$'s, and for each $k = 1, \dots, K$. This implies the stated localization rate. Again, we only prove (i) as (ii) can be shown analogously. We begin by showing that the refitted estimator for the $K$-th change point location attains the stated localization rate.
Let $\mathcal{X}_K'$ be as defined in \eqref{equation: X_K'}. We have that with probability at least $1 - n^{-2}$ that 
\begin{equation*}
\max_{\mathcal{X}' \in \mathfrak{X}} \left \{ - \frac{2^{d+1}}{C_{1,X} \operatorname{mes} \left ( \mathcal{X}_{(K)} \right )} \log \left [ M \left ( \eta_K + 1, \eta_{K} + \epsilon_{n,K}, \mathcal{X}' \right ) \right ] \right \} \leq \frac{1}{4} \log (1 / \mu_K). 
\end{equation*}
This holds because
\begin{align*}
& \mathbb{P} \left ( \max_{\mathcal{X}' \in \mathfrak{X}} \left \{ - \frac{2^{d+1}}{C_{1,X} \operatorname{mes} \left ( \mathcal{X}_{(K)} \right )} \log \left [ M \left ( \eta_K + 1, \eta_{K} + \epsilon_{n,K}, \mathcal{X}' \right ) \right ] \right \} > \frac{1}{4} \log (1 / \mu_K) \right ) \\
& \leq \sum_{\mathcal{X}' \in \mathfrak{X}} \mathbb{P} \left ( M \left ( \eta_K + 1, \eta_{K} + \epsilon_{n,K}, \mathcal{X}' \right ) \leq e^{-\frac{1}{4} \frac{C_{1,X}}{2^{d+1}} \operatorname{mes} \left ( \mathcal{X}_{(K)} \right ) \log (1 / \mu_K)} \right ) \\
& \leq \sum_{\mathcal{X}' \in \mathfrak{X}} \sum_{l=0}^{\epsilon_{n,K}} \binom{\epsilon_{n,K}}{l} \bar{p}_{\eta_K, \eta_K + \epsilon_{n,K}, \mathcal{X}'}^l \left ( 1 - \bar{p}_{\eta_K, \eta_K + \epsilon_{n,K}, \mathcal{X}'} \right )^{\epsilon_{n,K} - l } \\
& \hspace{10em} \times \left ( 1 - C_{1,R} \left ( 1 - e^{-\frac{1}{4} \frac{C_{1,X}}{2^{d+1}} \operatorname{mes} \left ( \mathcal{X}_{(K)} \right ) \log (1 / \mu_K)} \right ) \right )^{\epsilon_{n,K} - l} \\
& = \sum_{\mathcal{X}' \in \mathfrak{X}} \left ( 1 - \bar{p}_{\eta_K, \eta_K + \epsilon_{n,K}, \mathcal{X}'} C_{1,R} \left ( 1 - e^{-\frac{1}{4} \frac{C_{1,X}}{2^{d+1}} \operatorname{mes} \left ( \mathcal{X}_{(K)} \right ) \log (1 / \mu_K)} \right ) \right )^{\epsilon_{n,K}} \\
& \leq \left | \mathfrak{X} \right | \left ( 1 - \frac{C_{1,X}}{\log (n)} C_{1,R} \left ( 1 - e^{-\frac{1}{4} \frac{C_{1,X}}{2^{d+1}} \operatorname{mes} \left ( \mathcal{X}_{(K)} \right ) \log (1 / \mu_K)} \right ) \right )^{\epsilon_{n,K}}, 
\end{align*}
and so arguing as in \eqref{equation: (1-e^x) property bound i} and using the definition of $\epsilon_{n,K}$, if 
\begin{equation*}
\frac{1}{4} \frac{C_{1,X}}{2^{d+1}} \operatorname{mes} \left ( \mathcal{X}_{(K)} \right ) \log (1 / \mu_K) \leq 1
\end{equation*}
we will have that 
\begin{align*}
& \left ( 1 - \frac{C_{1,X}}{\log (n)} C_{1,R} \left ( 1 - e^{-\frac{1}{4} \frac{C_{1,X}}{2^{d+1}} \operatorname{mes} \left ( \mathcal{X}_{(K)} \right ) \log (1 / \mu_K)} \right ) \right )^{\epsilon_{n,K}} \\
\quad & \leq \left ( 1 - \frac{C_{1,X}}{\log (n)} C_{1,R} \left ( 1 - e^{-1} \right ) \frac{1}{4} \frac{C_{1,X}}{2^{d+1}} \operatorname{mes} \left ( \mathcal{X}_{(K)} \right ) \log (1 / \mu_K) \right ) ^{\epsilon_{n,K}} \\
& \leq \exp \left ( -\epsilon_{n,K} \frac{C_{1,X}}{\log (n)} C_{1,R} \left ( 1 - e^{-1} \right ) \frac{1}{4} \frac{C_{1,X}}{2^{d+1}} \operatorname{mes} \left ( \mathcal{X}_{(K)} \right ) \log (1 / \mu_K) \right ) \leq \frac{1}{n^2}
\end{align*}
with $C_1$ chosen appropriately larger. Further, arguing as in \eqref{equation: (1-e^x) property bound ii} if 
\begin{equation*}
\frac{1}{4} \frac{C_{1,X}}{2^{d+1}} \operatorname{mes} \left ( \mathcal{X}_{(K)} \right ) \log (1 / \mu_K) > 1
\end{equation*}
we will have that
\begin{align*}
& \left ( 1 - \frac{C_{1,X}}{\log (n)} C_{1,R} \left ( 1 - e^{-\frac{1}{4} \frac{C_{1,X}}{2^{d+1}} \operatorname{mes} \left ( \mathcal{X}_{(K)} \right ) \log (1 / \mu_K)} \right ) \right )^{\epsilon_{n,K}} \\
& \leq \left ( 1 - \frac{C_{1,X}}{\log (n)} C_{1,R} \left ( 1 - e^{-1} \right ) \right )^{\epsilon_{n,K}}  \leq \exp \left ( - {\epsilon_{n,K}} \frac{C_{1,X}}{\log (n)} C_{1,R} \left ( 1 - e^{-1} \right ) \right ) \leq \frac{1}{n^2}. 
\end{align*}
Consequently, arguing as in \eqref{equation: localization rate argument} - \eqref{equation: single change localization IV} we have that for any $\mathcal{X}'' \in \mathfrak{X}$
\begin{align*}
& \hat{L}^{\mathcal{X}'_K}_{\tilde{s}_K, \eta_K, \tilde{e}_K} - \hat{L}^{\mathcal{X}''}_{\tilde{s}_K, \eta_K + \epsilon_{n,K} + j, \tilde{e}_K} \nonumber \\
& \quad \geq N \left ( \tilde{s}_K, \eta_K, \mathcal{X}'_K \right ) \log \left ( 1 / \mu_{K} \right ) + 2 N \left ( \tilde{s}_K, \eta_K, \mathcal{X}'' \right ) \log \left [ M \left ( \eta_K, \eta_K + \epsilon_{n,K} + j, \mathcal{X}'' \right ) \right ] - \lambda \nonumber \\ 
& \quad =  2 N \left ( \tilde{s}_K, \eta_K, \mathcal{X}'_K \right ) \left \{ \frac{1}{2} \log (1/\mu_K) - \frac{N \left ( \tilde{s}_K, \eta_K, \mathcal{X}'' \right )}{N \left ( \tilde{s}_K, \eta_K, \mathcal{X}'_K \right )} \log \left [ M \left ( \eta_K + 1, \eta_K + \epsilon_{n,K} + j, \mathcal{X}'' \right ) \right] \right \} - \lambda \nonumber \\ 
& \quad \geq 2 N \left ( \tilde{s}_K, \eta_K, \mathcal{X}'_K \right ) \left \{ \frac{1}{2} \log (1/\mu_K) - \frac{2^{d+1}}{C_{1,X} \operatorname{mes} ( \mathcal{X}_{(k)} )} \log \left [ M \left ( \eta_K + 1, \eta_K + \epsilon_{n,K} + j, \mathcal{X}'' \right ) \right] \right \} - \lambda. 
\end{align*}
Therefore we have that for any $\mathcal{X}'' \in \mathfrak{X}$
\begin{align*}
\hat{L}^{\mathcal{X}'_K}_{\tilde{s}_K, \eta_K, \tilde{e}_K} - \hat{L}^{\mathcal{X}''}_{\tilde{s}_K, \eta_K + \epsilon_{n,K} + j, \tilde{e}_K} & \geq N \left ( \tilde{s}_K, \eta_K, \mathcal{X}'_K \right ) \log \left ( \mu_K \right ) - \lambda  \\
& \geq \frac{\delta_K \operatorname{mes} (\mathcal{X}_{(k)})}{2^{d+2}} \log \left ( 1 / \mu_K \right ) - \lambda > 0,
\end{align*}
with $C_4$ appropriately chosen. Consequently on the event \eqref{equation: N mathcal X event} with probability at least $1 - n^{-2}$, for all valid $j$'s
\begin{align*}
& \max_{\mathcal{X}' \in \mathfrak{X}} \hat{L}^{\mathcal{X}'}_{\tilde{s}_K, \eta_K, \tilde{e}_K} - \max_{\mathcal{X}' \in \mathfrak{X}} \hat{L}^{\mathcal{X}'}_{\tilde{s}_K, \eta_K + \epsilon_{n,K} + j, \tilde{e}_K} \\
& \quad \quad \quad \geq \hat{L}^{\mathcal{X}_K'}_{\tilde{s}_K, \eta_K, \tilde{e}_K} - \max_{\mathcal{X}'' \in \mathfrak{X}} \hat{L}^{\mathcal{X}''}_{\tilde{s}_K, \eta_K + \epsilon_{n,K} + j, \tilde{e}_K} > 0.  
\end{align*}
The localization rate is therefore established for the $K$-th change point location. Repeating the argument for $\eta_{K-1}, \dots, \eta_1$ and using the fact that necessarily $K < n$, a union bound argument is sufficient to establish that all change points are localized at the desired rate. Since the event \eqref{equation: N mathcal X event} holds with probability $1 - \mathcal{O} (n^{-1} \log^3 (n) )$ by the bound on $A_n$ in Theorem~\ref{theorem: consistent detection under local alternatives}, the proof is now completed. 
\end{proof}

\subsubsection{Proof of Proposition~\ref{lemma: local change minimax lower bound}}

\begin{proof}
The proof proceeds along the same lines as the proof of Proposition~\ref{lemma: global change minimax lower bound}. Let $P_{1:n} = \otimes_{t=1}^n P_t$ and $Q_{1:n} = \otimes_{t=1}^n Q_t$ denote respectively the joint laws of the random variables $\left \{ \left ( \boldsymbol{X}_t, Y_t \right ) \mid t = 1, \dots, n \right \}$ and $ \{ ( \boldsymbol{\tilde{X}}_t, \tilde{Y}_t ) \mid t = 1, \dots, n \}$, where 
\begin{equation}
Y_t = \begin{cases}
g ( \boldsymbol{X}_t ) R_t & 1 \leq t \leq \delta \\
f ( \boldsymbol{X}_t ) R_t & \delta < t \leq n
\end{cases}
\hspace{2em} \text{and} \hspace{2em}
\tilde{Y}_t = \begin{cases}
g ( \boldsymbol{\tilde{X}}_t ) \tilde{R}_t & 1 \leq t \leq \delta + \tau \\
f ( \boldsymbol{\tilde{X}}_t ) \tilde{R}_t & \delta + \tau < t \leq n 
\end{cases}
\end{equation}
and $\tau$ is a positive integer no larger than $n - \delta - 1$. Let the $X$'s and $R$'s have the same distributions as in the proof of Proposition~\ref{lemma: global change minimax lower bound}. Let $\mathcal{X}''$ be such that $\mathcal{X}' \subset \mathcal{X}''$ and $\operatorname{mes} \left ( \mathcal{X}'' \right ) = \min \left ( \operatorname{mes} \left ( \mathcal{X} \right ), 2 \operatorname{mes} \left ( \mathcal{X}' \right ) \right )$, and put $\mathcal{X}''' = \mathcal{X}'' \setminus \mathcal{X}'$. Let $f \left ( \cdot \right )$ be a production frontier function satisfying Assumption~\ref{assumption: production frontier function continuous}, and let 
\begin{equation}
g \left ( x \right ) = 
 \begin{cases}
 f \left ( x \right ) & \text{ if } x \in \mathcal{X} \setminus \mathcal{X}'' \\
\mu f \left ( x \right ) & \text{ if } x \in \mathcal{X}' \\
h \left ( x \right ) & \text{ if } x \in \mathcal{X}''' 
\end{cases}
\end{equation}
where $h \left ( \cdot \right )$ is chosen so that $g \left ( \cdot \right )$ is $K$-bilipschitz while maintaining $g \left ( x \right ) \leq f \left ( x \right )$ for all $x \in \mathcal{X}$. We again have that 
\begin{align}
\inf_{\hat{\eta}} \sup_{P \in \mathcal{Q}_n} \mathbb{E}_{P} \left [ \left | \hat{\eta} - \eta \right | \right ] & \geq \tau \left ( 1 - \text{TV} \left ( P_{1:n}, Q_{1:n} \right ) \right ) \nonumber \\ 
& \geq \frac{\tau}{2} \exp \left ( - \text{KL} \left ( Q_{1:n}, P_{1:n} \right ) \right ) = \frac{\tau}{2} \exp \left ( - \tau \text{KL} \left ( Q_{\delta+1}, P_{\delta+1} \right ) \right ), 
\label{equation: minimax bound again}
\end{align}
We have that $f_{\boldsymbol{X}_{\delta+1}, Y_{\delta+1}} \left ( \boldsymbol{x}, y \right )$ is as in \eqref{equation: joint density delta + 1} and moreover
\begin{equation*}
f_{\boldsymbol{\tilde{X}}_{\delta+1}, \tilde{Y}_{\delta+1}} \left ( \boldsymbol{x}, y \right ) = p_{\boldsymbol{X}} \left ( \boldsymbol{x} \right ) \times 
\begin{cases}
p_R \left ( y / f (\boldsymbol{x}) \right ) \left [ f(\boldsymbol{x}) \right ] ^{-1} & \text{if}\; \boldsymbol{x} \in \mathcal{X} \setminus \mathcal{X}'' \\
p_R \left ( y / \mu f (\boldsymbol{x}) \right ) \left [ \mu f(\boldsymbol{x}) \right ] ^{-1} & \text{if}\; \boldsymbol{x} \in \mathcal{X}' \\
p_R \left ( y / h (\boldsymbol{x}) \right ) \left [ h(\boldsymbol{x}) \right ] ^{-1} & \text{if}\; \boldsymbol{x} \in \mathcal{X}'''
\end{cases}. 
\end{equation*}
Therefore we have that 
\begin{align}
\text{KL} \left ( Q_{\delta+1}, P_{\delta+1} \right ) & = \tau \int_{\mathcal{X}'} \int_{\mathcal{Y}} \log \left ( \frac{f_{\boldsymbol{\tilde{X}}_{\delta+1}, \tilde{Y}_{\delta+1}} \left ( x, y \right )}{f_{X_{\delta+1}, Y_{\delta+1}} \left ( x, y \right )} \right ) f_{\boldsymbol{\tilde{X}}_{\delta+1}, \tilde{Y}_{\delta+1}} \left ( x, y \right ) \mathrm{d} y \mathrm{d} \boldsymbol{x} \nonumber \\
& \hspace{4em} + \tau \int_{\mathcal{X}'''} \int_{\mathcal{Y}} \log \left ( \frac{f_{\boldsymbol{\tilde{X}}_{\delta+1}, \tilde{Y}_{\delta+1}} \left ( x, y \right )}{f_{X_{\delta+1}, Y_{\delta+1}} \left ( x, y \right )} \right ) f_{\boldsymbol{\tilde{X}}_{\delta+1}, \tilde{Y}_{\delta+1}} \left ( x, y \right ) \mathrm{d} y \mathrm{d} \boldsymbol{x} \nonumber \\
& \leq \tau \log \left ( \frac{1}{\mu} \right ) \int_{\mathcal{X}''} p_{\boldsymbol{X}} (\boldsymbol{x}) \mathrm{d} \boldsymbol{x} \nonumber \\
& \leq 2 C_{L,X} \tau \log \left ( \frac{1}{\mu} \right ) \operatorname{mes} \left ( \mathcal{X}' \right ),
\label{equation: KL product space bound}
\end{align}
where $C_{L,X}$ is an absolute constant depending on the $C_{2,X}$ and the Lipschitz constant in Assumption~\ref{assumption: production frontier function continuous}. Therefore plugging \eqref{equation: KL product space bound} into \eqref{equation: minimax bound again} we have that
\begin{equation}
\inf_{\hat{\eta}} \sup_{P \in \mathcal{Q}_n} \mathbb{E}_{P} \left [ \left | \hat{\eta} - \eta \right | \right ] \geq \frac{\tau}{2} \exp \left ( -2 C_{L,X} \tau \operatorname{mes}  (\mathcal{X}') \log \left ( \frac{1}{\mu} \right ) \right ). 
\end{equation}
Finally we put $\tau = \min \left ( \left \lceil \left ( \operatorname{mes} (\mathcal{X}') \log (1/\mu) \right ) ^ {-1} \right \rceil, n - \delta - 1 \right )$ and argue as was done in \eqref{equation: change energy inequality} to complete the proof. 
\end{proof}

\subsubsection{Proof of Corollary~\ref{corollary: exponential concentration of FDH estimator strong mixing}}

\begin{proof}
We repeat the analysis from Section~\ref{section: proof of exponential concentration in Haus. distance} but this time for each $u \in \mathcal{N}$ let the radii $r_k := r_k (u)$ be defined so that
\begin{equation*}
\operatorname{mes} \left ( W \left ( u, r_k \right ) \right ) = C_W \psi_n^{d+1} \left [ \left ( \frac{d}{d+1} + k \right ) \left ( \frac{2d + 1}{d+1} + k \right ) \log ^ {-1} \left ( 1/A \right ) \right ]. 
\end{equation*}
For each $A_{k,n}$ defined as in \eqref{equation: psi-net events} introduce the quantity
\begin{equation*}
\beta_{k,n} = \log ^ {-1} \left ( 1/A \right ) \left (\frac{2d+1}{d+1} + k \right ) \log (n) 
\end{equation*}
and define the set $\mathcal{I}_{\beta_{k,n}} = \left \{ t = s \beta_{k,n} \mid 1 \leq s \leq n / \beta_{k,n} \right \}$. Using Assumption~\ref{assumption: strong mixing} and arguing as in \eqref{equation: P A bound} we have that 
\begin{align*}
& \mathbb{P} \left ( A_{k,n} \right ) \geq 1 - \sum_{l=1}^M \mathbb{P} \left ( \cap_{t \in \mathcal{I}_{\beta_{k,n}}} \left \{ \left ( \boldsymbol{X}_t, Y_t \right ) \not\in W \left ( u_l, r_k \right ) \right \} \right ) \\ 
& \geq 1 - \sum_{l=1}^M \left [ \prod_{t \in \mathcal{I}_{\beta_{k,n}}} \mathbb{P} \left ( \left ( \boldsymbol{X}_t, Y_t \right ) \not\in W \left ( u_l, r_k \right ) \right ) + n A^{\beta_{k,n}} \right ] \\
& \geq 1 - \sum_{l=1}^M \left [ \prod_{t \in \mathcal{I}_{\beta_{k,n}}} \left ( 1 - \mathbb{P} \left ( Y_t \in C'' \left ( u_l', r_k \right ) \mid X_t \in C' \left ( u_l', r_k \right ) \right ) \mathbb{P} \left ( X_t \in C' \left ( u_l', r_k \right ) \right ) \right ) + n A^{\beta_{k,n}} \right ] \\ 
& \geq 1 - C_d \psi_n \left [ \left ( 1 - \frac{\left ( \frac{d}{d+1} + k \right ) \log (n) }{n / \beta_{k,n}} \right ) ^ {n / \beta_{k,n}} + n A ^ {\beta_{k,n}}\right ] \\ 
& \geq 1 - 2 C_d n^{-k}
\end{align*}
Arguing with (i) and (ii) used to show \eqref{equation: bound on radii} we have that for some absolute $C > 0$ depending on $A$ we have that $r_k \leq C \left ( (2d+1) / (d+1) + k \right ) \psi_n$. Therefore, arguing as in \eqref{equation: bound when x close to zero} and \eqref{equation: bound when x far from zero} we have that for any $l > 0$
\begin{equation*}
\mathbb{E} \left [ \left ( \psi_n^{-1} \left \| f - \hat{f}_n \right \|_\infty \right ) ^ l \mid A_{k,n} \right ] \leq \left ( C_1 k + C_2 \right ) ^ l , \hspace{2em} k = 1, \dots, n^{\frac{1}{d+1}}, 
\end{equation*}
The rest of the proof goes through as in \eqref{equation: further Markov bound on P} and is therefore omitted for brevity. 
\end{proof}

\end{document}